\documentclass[a4paper,11pt]{article}

\usepackage{a4wide}

\usepackage[noadjust]{cite}
\usepackage{amsmath}
\usepackage{amssymb}
\usepackage{amsthm} 

\usepackage{algorithm}
\usepackage[noend]{algpseudocode} 

\usepackage{array}
\usepackage{enumerate}
\usepackage{caption} 
\usepackage{mathpartir}
\usepackage{longtable}

\usepackage{tikz}
\usetikzlibrary{fit,positioning, intersections, shapes.geometric,shapes.symbols,arrows,decorations.markings,decorations.pathreplacing,calc}

\newcommand{\executeiffilenewer}[3]{%
	\ifnum\pdfstrcmp{\pdffilemoddate{#1}}%
	{\pdffilemoddate{#2}}>0%
	{\immediate\write18{#3}}\fi%
} 

\newcommand{\svg}[2]{\def\svgwidth{#1}%
	\executeiffilenewer{#2.svg}{#2.pdf}%
	{inkscape -z -D --file=#2.svg %
		--export-pdf=#2.pdf --export-latex}%
	{\input{#2.pdf_tex}}}

\theoremstyle{plain}
\newtheorem{thm}{Theorem}
\newtheorem*{thm*}{Theorem}
\newtheorem{cor}[thm]{Corollary}
\newtheorem{lem}[thm]{Lemma}
\newtheorem{obs}[thm]{Observation}
\newtheorem{conj}[thm]{Conjecture}

\newtheorem*{clm*}{Claim}
\theoremstyle{definition}
\newtheorem{defn}[thm]{Definition}
\newtheorem{rem}[thm]{Remark}
\newtheorem*{thmmain}{Theorem~\ref{thm:main}}

\newenvironment{claimproof}{\par\noindent{Proof:}\space}{\hfill $\blacksquare$ \newline}

\newcommand{\parHom}[1]{\mathrm{\oplus\textsc{Hom}}\left(#1\right)}
\newcommand{\parRet}[1]{\mathrm{\oplus\textsc{Ret}}\left(#1\right)}
\newcommand{\parLHom}[1]{\mathrm{\oplus\textsc{LHom}}\left(#1\right)}

\newcommand{\paris}{\mathrm{\oplus\textsc{IS}}}
\newcommand{\parbis}{\mathrm{\oplus\textsc{BIS}}}

\renewcommand{\hom}[2]{\mathrm{hom}\left(#1 \to #2\right)}

\def\cycles{\text{Cy}}

\def\calH{\mathcal{H}}

\newcommand{\nesetril}{Ne\v{s}et\v{r}il}
\newcommand{\Zivny}{{\v{Z}}ivn{\'y}}

\newcommand{\Goebel}{G{\"o}bel}

\newcommand{\poly}{\mathrm{poly}}

\newcommand{\NP}{\mathrm{NP}} 
\newcommand{\parP}{\oplus\mathrm{P}}

\newcommand{\NH}[1]{\Gamma_H(#1)}
\newcommand{\Nof}[1]{\Gamma_{#1}}

\newcommand{\boldS}{\mathbf{S}}

\newcommand{\calC}{\mathcal{C}}

\newcommand{\calI}{\mathcal{I}}
\newcommand{\calL}{\mathcal{L}}
\newcommand{\calP}{\mathcal{P}}

\newcommand{\TV}{\mathsf{V}}

\newcommand{\Oy}{\ensuremath{\Omega_y}}
\newcommand{\Oz}{\ensuremath{\Omega_z}}
\newcommand{\Sox}{\ensuremath{\Sigma_{o,x}}}
\newcommand{\Sos}{\ensuremath{\Sigma_{o,s}}}
\newcommand{\Six}{\ensuremath{\Sigma_{i,x}}}
\newcommand{\Sis}{\ensuremath{\Sigma_{i,s}}}

\newcommand{\bc}{\mathsf{BC}}
\newcommand{\SP}[1]{\textsc{P}_H(#1)}
\newcommand{\walk}[1]{\textsc{Walk}(#1)}
\newcommand{\subwalk}[1]{\textsc{ExitWalk}(#1)}
\newcommand{\cseg}{\textsc{W}}

\newcommand{\dist}{\mathsf{dist}}
\newcommand{\abs}[1]{\left\vert #1 \right\vert}

\let\epsilon=\varepsilon

\let\originalleft\left
\let\originalright\right
\renewcommand{\left}{\mathopen{}\mathclose\bgroup\originalleft}
\renewcommand{\right}{\aftergroup\egroup\originalright}

\newcommand{\prob}[3]{
	\vbox{
		\begin{description}\setlength{\itemsep}{0pt}
		\setlength{\parskip}{0pt}
		\setlength{\parsep}{0pt}   			
		        \item[\bf Name:] #1
			\item[\bf Input:] #2  
			\item[\bf Output:] #3
		\end{description}
	}
}

\newenvironment{myitemize}
{ \begin{itemize}
		\setlength{\itemsep}{0pt}
		\setlength{\parskip}{0pt}
		\setlength{\parsep}{0pt}     }
	{ \end{itemize}                  } 

\newenvironment{myenumerate}
{ \begin{enumerate}
		\setlength{\itemsep}{0pt}
		\setlength{\parskip}{0pt}
		\setlength{\parsep}{0pt}     }
	{ \end{enumerate} 				 }

\usepackage{filecontents}

\usepackage[hidelinks, bookmarks=true]{hyperref}

\title{Counting Homomorphisms to $K_4$-minor-free Graphs, modulo 2\thanks{A preliminary version of this paper (without the proofs) appeared in the proceedings of SODA 2021. The research leading to these results has received funding from the European Research Council (ERC) under the European Union's Horizon 2020 research and innovation programme (grant agreement No 714532). Jacob Focke has received funding from the Engineering and Physical Sciences Research Council (grant ref: EP/M508111/1). Stanislav \Zivny\ was supported by a Royal Society University Research Fellowship. The paper reflects only the authors' views and not the views of the ERC or the European Commission. The European Union is not liable for any use that may be made of the information contained therein.}}
\author{Jacob Focke\thanks{Department of Computer Science, University of Oxford} \and Leslie Ann Goldberg$^\dagger$ \and Marc Roth\thanks{Merton College, University of Oxford} \and Stanislav {\v{Z}}ivn{\'y}$^\dagger$}

\date{16 July 2021\vspace{-5mm}}

\begin{document}
	\maketitle

\begin{abstract}
We study the problem of computing the parity of the number of homomorphisms from an input graph $G$ to a fixed graph $H$. Faben and Jerrum  [ToC'15]
introduced an explicit criterion on the graph $H$ and conjectured that, if satisfied, the problem is solvable in polynomial time and, otherwise, 
the problem is complete for the complexity class $\parP$ of parity problems. 
		
We verify their conjecture for all graphs $H$ that exclude the complete graph on $4$ vertices as a minor. Further, we rule out the existence of a subexponential-time algorithm for the $\parP$-complete cases, assuming the randomised Exponential Time Hypothesis.
		
Our proofs introduce a novel method of deriving hardness from globally defined substructures of the fixed graph $H$. Using this, we subsume all  
prior progress towards  resolving the conjecture (Faben and Jerrum [ToC'15]; G\"obel, Goldberg and Richerby [ToCT'14,'16]). 
As special cases, our machinery also yields a proof of the conjecture for graphs 
with maximum degree at most  $3$, as well as a full classification for the problem of counting list homomorphisms, modulo $2$.\end{abstract}

\section{Introduction}

A homomorphism 
\label{page:homomorphism}from a graph~$G$ to a graph~$H$
is a map~$h$ from~$V(G)$ to~$V(H)$ that preserves edges in the sense that, for every edge
$\{u,v\}$ of~$G$, the image $\{h(u),h(v)\}$ is an edge of~$H$.
Many combinatorial structures can be modelled using graph homomorphisms.
For this reason, graph homomorphisms 
are  ubiquitous in both classical and modern-day complexity theory with applications in areas such as
constraint satisfaction problems~\cite{Grohe07},
evaluations of spin systems in statistical physics~\cite{BW,Borgs},
database theory~\cite{KolaitisV00,GroheSS01}, and parameterised algorithms~\cite{CurticapeanDM17,RothW20}.
The computational problems of finding and counting homomorphisms 
are therefore amongst   the most well-studied computational problems; the analysis of their complexity dates back to the intractability result for computing the chromatic number, one of Karp's original 21 $\NP$-complete problems~\cite{Karp72}.
More recent work builds on  Hell and \nesetril's  celebrated dichotomy theorem~\cite{HellN90}, 
which shows that determining whether an input graph~$G$ has a homomorphism to a fixed graph~$H$
is polynomial-time solvable if $H$ is bipartite, or if $H$ has a self-loop. For any other graph~$H$, they show that the problem is NP-complete.

This paper focusses on the problem of counting homomorphisms.
Applications of this problem are discussed in \cite{Borgs}. The complexity of the problem 
has   been the focus of much research  (see, for example, \cite{DG, BIShard, ListApx,KelkThesis,ChenCD19}).\footnote{
There is also a huge literature on generalisations of this problem such as counting \emph{weighted} homomorphisms
(computing partition functions of spin systems or holant problems), counting homomorphisms to \emph{directed} graphs, counting partition functions
of constraint satisfaction problems, and counting homomorphisms with restrictions such as surjectivity.
These generalisations and restrictions are beyond the scope of this paper.}

The complexity of counting homomorphisms was initiated by Dyer and Greenhill~\cite{DG}, who   gave 
a complete dichotomy theorem. The complexity of counting the homomorphisms from an input graph~$G$ to a fixed graph~$H$
is polynomial-time solvable if  every component of~$H$ is either a complete bipartite graph with no self-loops or a complete graph in which every vertex has a self-loop. For any other graph~$H$, they show that the problem is \#P-complete.

Given that (exactly) counting the homomorphisms to~$H$ is \#P-complete for almost every graph~$H$,
research has focussed on restrictions of the problem. Instead of  
determining the exact number of homomorphisms from~$G$
to~$H$,  compute an approximation to this number \cite{BIShard, ListApx, KelkThesis}, 
or determine whether it is odd or even \cite{Faben2008, FJ, squarefree, cactus}, 
or determine its value modulo any prime~$p$ \cite{goebeltrees, BulatovModp}.
Alternatively, consider the parameterised complexity \cite{FlumGrohe}.
 For example, the problem can be studied when the input~$G$ is assumed to have bounded treewidth~\cite{DiazST02}
or when $H$ has a bounded treewidth, for example when $H$ is a tree \cite{FJ, goebeltrees,BulatovModp,MarkTrees}.

Restricting the input~$G$ to have bounded treewidth 
makes counting homomorphisms tractable --- given this restriction, the problem is solvable in polynomial time for any 
fixed~$H$ \cite[Corollary 5.1]{DiazST02}. 
Restricting the fixed target graph~$H$ to have bounded treewidth 
 leads to a more nuanced complexity classification, even for treewidth~$1$ (when $H$ is a tree).
 For example, the complexity of approximately counting homomorphisms to a tree~$H$ has
 still not been fully resolved, and it is known that different trees lead to vastly different complexities.
 For example,  approximately counting homomorphisms to the very simple tree that is a path of length~$3$
 is equivalent to \#BIS, which is the canonical open problem in approximate counting \cite{APred}.
 Moreover, \cite{MarkTrees} shows that for every integer~$q\geq 3$ there is a tree~$J_q$ such that approximately counting homomorphisms to~$J_q$
 is equivalent to classic problem of
 approximating the partition function of the $q$-state Potts model from statistical physics.
 Also, it shows that there are trees~$H$ such that approximately counting homomorphisms to~$H$ is NP-hard.

\subsection{Counting modulo 2 and Past Work}

Faben and Jerrum~\cite{FJ} combined  the restriction that $H$ is a tree with the 
restriction that counting is modulo~$2$.
Their result will be important for our work, so we next give the definitions that we need to state their result.

 The complexity class $\parP$  \label{page:parP} \cite{PZ82:Counting, GP86:Parallel}  contains all problems of the form
``compute $f(x) \bmod 2$'' such that computing $f(x)$ is a problem in \#P.
Toda~\cite{Tod91:PP-PH} 
 showed that there is   a randomised polynomial-time reduction
from every
problem in the polynomial hierarchy to 
some problem in~$\parP$.
Thus, $\parP$-hardness is viewed as a stronger kind of intractability than NP-hardness.
We use $\parHom{H}$ \label{page:ParHom}
to denote the computational problem
of computing the number of homomorphisms from~$G$ to~$H$, modulo~$2$, given
an input graph~$G$.
It is  immediate from the definition that $\parHom{H}$ is in $\parP$.

 The \emph{involution-free reduction} of a graph~$H$, from \cite{FJ},
 is defined as follows. 
 An \emph{involution}\label{page:involution}~$\sigma$ of~$H$ is an automorphism of~$H$ whose order is  at most~$2$ (that is,  
	$\sigma\circ \sigma$ is the identity permutation).
	An involution is   \emph{non-trivial} if it is not the identity permutation.  
 A graph $H$ is \emph{involution-free}\label{page:involutionfree} if it  has no non-trivial involutions. 
  $H^\sigma$ denotes the subgraph of~$H$ induced by the fixed points of~$\sigma$
(the vertices $v$ with $\sigma(v)=v$).
We write $H\rightarrow K$ if there is a non-trivial involution~$\sigma$ of~$H$ such that $K= H^\sigma$.
The relation $\rightarrow^*$ is the reflexive-transitive closure of the relation~$\rightarrow$.
Thus, $H \rightarrow^* K$ means that either $K=H$, or there is 
a positive integer~$j$ and
a sequence $H_1,\ldots,H_j$ of graphs such that $H=H_1$, $K=H_j$
and, for all $i\in[j]$, $H_i \rightarrow H_{i+1}$.
Faben and Jerrum~\cite[Theorem~3.7]{FJ} showed that every graph $H$ has, up to isomorphism, 
exactly one involution-free graph~$H^*$ such that $H \rightarrow^*H^*$.  This graph $H^\ast$ 
(labelled in a canonical way) is   the \emph{involution-free reduction}\label{page:involutionfreereduction} of $H$.
The relevance of the involution-free reduction is given by the following theorem.
\begin{thm}[{\cite[Theorem 3.4]{FJ}}]\label{thm:ReductionSuffices}\label{thm:star}
For all graphs $G$ and $H$,  the number of homomorphisms from~$G$ to~$H$
has the same parity as the number of homomorphisms from~$G$ to~$H^*$. 
\end{thm}
Thus, the computational problem $\parHom{H}$ reduces to~$\parHom{H^*}$.
Faben and Jerrum made the following conjecture~\cite{FJ}.
  
 \begin{conj} 
 [Faben-Jerrum Conjecture]  
 \label{conj:FJ}  
 Let $H$ be a  graph.  If its involution-free
 reduction~$H^*$ has at most one vertex, then $\parHom{H}$  can be solved in polynomial time.
 Otherwise, $\parHom{H}$
 is $\parP$-complete.\end{conj}  
  
The following progress has been made on the Faben-Jerrum conjecture.
\begin{myitemize}
\item Faben and Jerrum \cite[Theorem 3.8, Theorem 6.1]{FJ} proved the conjecture for the case where 
every connected component of $H$ is a tree.
\item \Goebel, Goldberg and Richerby \cite[Theorem 3.8]{cactus} proved the conjecture for the case where 
every connected component of $H$ is a 
cactus graph, which is a connected, simple graph in which every edge belongs to at most one cycle.
\item \Goebel, Goldberg and Richerby \cite[Theorem 1.2]{squarefree} proved the conjecture for the case where 
$H$ is a simple graph whose involution-free reduction
$H^*$
is square-free (meaning that it has no $4$-cycle).
\end{myitemize}
The cactus-graph result generalises the tree result, and is incomparable with the square-free result.

 \subsection{Contributions and Techniques}  
  
 Our first (and main) contribution is to prove the Faben-Jerrum conjecture for every simple graph $H$ that  does not have a $K_4$-minor.
 
 Here, $K_4$ \label{page:K4} denotes the complete graph with four vertices.
 The concept of graph minors is well known
 (see, for example, \cite{DiestelGT}). In short,
 a graph~$H$ is $K_4$-minor-free 
 if $K_4$ cannot be obtained from~$H$ by a sequence of vertex deletions, edge deletions, and
 edge contractions (removing any self-loops and multiple-edges that are formed by the contraction).  
 Graph classes based on excluded minors 
 form the basis of the graph structure theory of Robertson and Seymour (see \cite{Lovaszsurvey}).
  
The class of $K_4$-minor-free graphs is a rich and well-studied class. It is equivalent to  
the class of graphs with treewidth at most~$2$  
and it includes all outerplanar and series-parallel graphs \cite{Duffin}.

Both trees and cactus graphs are $K_4$-minor free, so
 our  result
subsumes the tree result of Faben and Jerrum \cite{FJ} and also the cactus-graph result of \Goebel\ et al.~\cite{cactus}.
 $K_4$-minor-free graphs can contain a $4$-cycle and, going the other way, square-free graphs
 can have a $K_4$-minor. Thus, our result is incomparable with the result of~\cite{squarefree}.
  (As a more minor contribution, our techniques also give a shorter proof of the result of~\cite{squarefree} --- see Remark~\ref{A1}.)  
  
 Our second contribution is to extend $\parP$-hardness, using the
 randomised version of the Exponential Time Hypothesis  
 of Impagliazzo and Paturi (rETH) \label{page:rETH} to rule out subexponential algorithms.
 In order to state our result, we first state the hypothesis.
  \begin{conj}[rETH,~\cite{ImpagliazzoP01}]
There is a positive constant $c$ such that no algorithm, deterministic or randomised, can decide 
the satisfiability of an $n$-variable
$3\textsc{-SAT}$ instance in time $\exp(c\cdot n)$.
\end{conj}
	  
 Using the rETH, we can now state our main result. Here (and in the rest of the paper)
 we denote the size of the input graph~$G$ as $|G|=|V(G)|+|E(G)|$.
  \newcommand{\statethmmain}{Let $H$ be a simple graph  whose involution-free reduction $H^\ast$ is $K_4$-minor free.
If~$H^*$ contains at most one vertex, then $\parHom{H}$ can be solved in polynomial time.
Otherwise,   $\parHom{H}$ is $\parP$-complete and, assuming the randomised Exponential Time Hypothesis,
it  cannot be solved in time $\exp(o(\abs{G}))$. }
 \begin{thm}\label{thm:main} \statethmmain
\end{thm}

  As an example  of an application of Theorem~\ref{thm:main}, consider the following $K_4$-minor-free graphs~$H_1$ and~$H_2$.

	
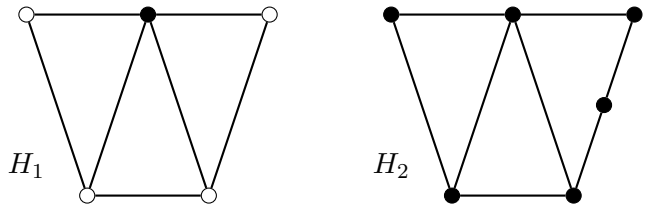
\begin{center} 
	\begin{tikzpicture}
	[scale=0.8, node distance = 1.4cm]
	\tikzstyle{dot}   =[fill=black, draw=black, circle, inner sep=0.15mm]
	\tikzstyle{vertex}=[fill=black, draw=black, circle, inner sep=2pt]
	\tikzstyle{dist}  =[fill=white, draw=black, circle, inner sep=2pt]
	\tikzstyle{pinned}=[draw=black, minimum size=10mm, circle, inner sep=0pt]
	
	\node[vertex] (t) at (0  ,0)  {};
	\node[vertex] (s) at (-1  ,3) {};
	\node[vertex] (l) at (-2  ,0)  {};
	\node[vertex] (ru) at (1  ,3) {};
	\node[vertex] (rd) at (0.5  ,1.5) {};
	\node[vertex] (sl) at (-3  ,3) {};
	\draw[thick] (s) -- (l); \draw[thick] (s) -- (ru); \draw[thick] (s) -- (sl);
	\draw[thick] (t) -- (l); \draw[thick] (s) -- (t); \draw[thick] (t) -- (rd);
	\draw[thick] (l) -- (sl); \draw[thick] (ru) -- (rd);
	\node (cl) at (-1.5  ,1.5)  {};
	\node () at (-3,0.5) {$H_2$};
	
	\begin{scope}[xshift=-6cm]
	\node[fill=white, draw=black, circle, inner sep=2pt] (t) at (0  ,0)  {};
	\node[vertex] (s) at (-1  ,3) {};
	\node[fill=white, draw=black, circle, inner sep=2pt] (l) at (-2  ,0)  {};
	\node[fill=white, draw=black, circle, inner sep=2pt] (ru) at (1  ,3) {};
	\node[fill=white, draw=black, circle, inner sep=2pt] (sl) at (-3  ,3) {};
	\draw[thick] (s) -- (l); \draw[thick] (s) -- (ru); \draw[thick] (s) -- (sl);
	\draw[thick] (t) -- (l); \draw[thick] (s) -- (t); \draw[thick] (t) -- (ru);
	\draw[thick] (l) -- (sl); 
	\node (cl) at (-1.5  ,1.5)  {};
	\node () at (-3,0.5) {$H_1$};
	\end{scope}
	
	\end{tikzpicture}
 \end{center}

\noindent 
The graph~$H_1$ has a non-trivial involution whose   only fixed-point is the solid vertex,
\label{page:H2}
so $H_1^*$ has one vertex.  By Theorem~\ref{thm:main},  $\parHom{H_1}$ can be solved in polynomial time.
The graph~$H_2$ does not have any non-trivial involutions, so  $H_2^*=H_2$.
By Theorem~\ref{thm:main},  $\parHom{H_2}$ is $\parP$-complete and 
it cannot be solved in time $\exp(o(|G|))$, unless the rETH fails.

 Before describing our techniques, we mention that they lead easily to a couple of other results ---
a proof of the Faben-Jerrum conjecture for graphs whose involution-free reduction have degree at most~$3$
(Theorem~\ref{thm:bddeg})
and a complete complexity classification for counting list homomorphisms modulo~$2$ (Theorem~\ref{cor:LHomclassification}).

\paragraph*{Technical Overview}
Given Theorem~\ref{thm:star}, we focus on the case where~$H$ is involution-free.
In general, our proof proceeds in two steps. Given an involution-free $K_4$-minor-free graph $H$, in step 1 we try to find a biconnected component of $H$, let us call it $B$, that allows us to derive $\parP$-hardness of $\parHom{H}$ by exploiting the \emph{local} structure of $B$  to construct a reduction from counting independent sets, modulo $2$. The latter problem, denoted by $\paris$,
\label{page:ParIS}
is known to be $\parP$-complete~\cite{Valiant2006} and cannot be solved in subexponential time, unless the rETH fails~\cite{DellHMTW14}. 

A careful analysis of biconnected and $K_4$-minor-free graphs, which crucially relies on the absence of 
non-trivial involutions, shows that the first step is always possible, unless all biconnected components of $H$ have a very restricted form; examples are depicted in Figure~\ref{fig:introBiconnectedComps}.
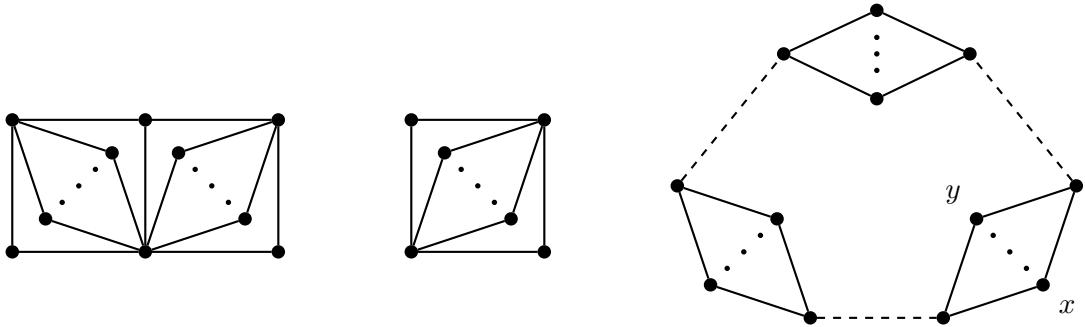
\begin{figure}[t]
	\centering
	\begin{tikzpicture}[scale=1.75, node distance = 1.4cm,thick]
	\tikzstyle{dot}   =[fill=black, draw=black, circle, inner sep=0.15mm]
	\tikzstyle{vertex}=[draw=black, fill=black, circle, inner sep=1.5pt]
	\tikzstyle{terminal}=[fill=black, draw=black, circle, inner sep=1.5pt]
	\tikzstyle{dist}  =[fill=white, draw=black, circle, inner sep=2pt]
	\tikzstyle{pinned}=[draw=black, minimum size=10mm, circle, inner sep=0pt]	
	\begin{scope}
	\node[vertex] (s) at (0  ,1)  {};
	\node[vertex] (i) at (0  ,0)  {};
	\node[vertex] (o) at (-1  ,1)  {};
	\node[terminal] (u) at (-1  ,0)  {};
	\node[vertex] (v) at (1  ,1)  {};
	\node[terminal] (x) at (1  ,0)  {};	
	\draw (s)--(i); \draw (s)--(o); \draw (s)--(v); \draw (o)--(u);
	\draw (u)--(i); \draw (x)--(i); \draw (v)--(x);	
	\node[vertex] (y1) at (-0.75  ,0.25)  {};
	\node[vertex] (yk) at (-0.25  ,0.75)  {};
	\node[vertex] (zl) at (0.25  ,0.75)   {};
	\node[vertex] (z1) at (0.75  ,0.25) {};	
	\draw (o)--(y1); \draw (o)--(yk);  \draw (i)--(y1); \draw (i)--(yk);
	\draw (v)--(z1); \draw (v)--(zl); \draw (i)--(z1); \draw (i)--(zl);    	
	\node[dot] (d) at (-0.5,0.5) {};  
	\node[dot] (d) at (-0.375,0.625) {}; 
	\node[dot] (d) at (-0.625,0.375) {}; 
	\node[dot] (d) at (0.5,0.5) {};  
	\node[dot] (d) at (0.375,0.625) {}; 
	\node[dot] (d) at (0.625,0.375) {}; 
	\end{scope}
	
	\begin{scope}[xshift=2cm] 
	\node (l) at (-1,0) {};
	\node[vertex] (s) at (0  ,1)  {};
	\node[terminal] (i) at (0  ,0)  {};
	\node[vertex] (v) at (1  ,1)  {};
	\node[terminal] (x) at (1  ,0)  {};	
	\draw (s)--(i); \draw (s)--(v);
	\draw (x)--(i); \draw (v)--(x);	
	\node[vertex] (zl) at (0.25  ,0.75)   {};
	\node[vertex] (z1) at (0.75  ,0.25) {};	
	\draw (v)--(z1); \draw (v)--(zl); \draw (i)--(z1); \draw (i)--(zl);    	
	\node[dot] (d) at (0.5,0.5) {};  
	\node[dot] (d) at (0.375,0.625) {}; 
	\node[dot] (d) at (0.625,0.375) {}; 
	\end{scope}
	
	\begin{scope}[yshift=-0.5cm, xshift=5cm] 
	\node[vertex] (n) at (-0.2  ,2)  {};
	\node[vertex] (m) at (1.2  ,2)  {};
	\node[vertex] (i) at (0  ,0)  {};
	\node[vertex] (i') at (1  ,0) {};
	\node[vertex] (o) at (-1  ,1)  {};
	\node[vertex] (v) at (2  ,1)  {};
	\node[vertex] (x1) at (0.5  ,2.33)  {};
	\node[vertex] (xk) at (0.5  ,1.66)   {};
	\node[vertex] (y1) at (-0.75  ,0.25)  {};
	\node[vertex] (yk) at (-0.25  ,0.75)   {};
	\node[vertex] (zl) at (1.25  ,0.75)  [label=135: $y$]   {};
	\node[vertex] (z1) at (1.75  ,0.25)  [label=-45: $x$]  {};	
	\draw (o)--(y1); \draw (o)--(yk);  \draw (i)--(y1); \draw (i)--(yk);
	\draw (v)--(z1); \draw (v)--(zl); \draw (i')--(z1); \draw (i')--(zl);  
	\draw (n)--(x1); \draw (n)--(xk); \draw (m)--(x1); \draw (m)--(xk);  
	\draw[dashed] (o) -- (n); \draw[dashed] (v) -- (m);
	\draw[dashed] (i) -- (i'); 	
	\node[dot] (d) at (-0.5,0.5) {};  
	\node[dot] (d) at (-0.375,0.625) {}; 
	\node[dot] (d) at (-0.625,0.375) {}; 
	\node[dot] (d) at (1.5,0.5) {};  
	\node[dot] (d) at (1.375,0.625) {}; 
	\node[dot] (d) at (1.625,0.375) {}; 
	\node[dot] (d) at (0.5,2) {};  
	\node[dot] (d) at (0.5,2.125) {}; 
	\node[dot] (d) at (0.5,1.875) {}; 
	\end{scope}
	\end{tikzpicture} 
	\caption{Examples of three types of biconnected and $K_4$-minor-free graphs that, if viewed as biconnected components, do not yield a reduction from $\paris$. We will re-encounter those graphs as ``impasses'' \emph{(left)}, ``diamonds'' \emph{(centre)}, and ``obstructions'' \emph{(right)}.}
	\label{fig:introBiconnectedComps}
\end{figure}

The second step of the proof exploits the global structure of $H$ and deals with the case where step 1 fails. Note that all of the depicted biconnected components have non-trivial involutions; consider for example the  involution given by swapping 
the vertices~$x$ and~$y$ in Figure~\ref{fig:introBiconnectedComps}.  Since the overall graph $H$ is promised to be free of such  involutions, we infer that at least one of $x$ and $y$ has a neighbour in a further biconnected component of $H$, which will allow us to successively construct a global ``walk-like'' structure in $H$ that eventually yields a reduction from $\paris$. 

We consider the construction of those global substructures as our main technical contribution. While the formal specifications are beyond the scope of the introduction, we give an illustrated example which we hope  gives some flavour of the graph theory that we will encounter in this work:

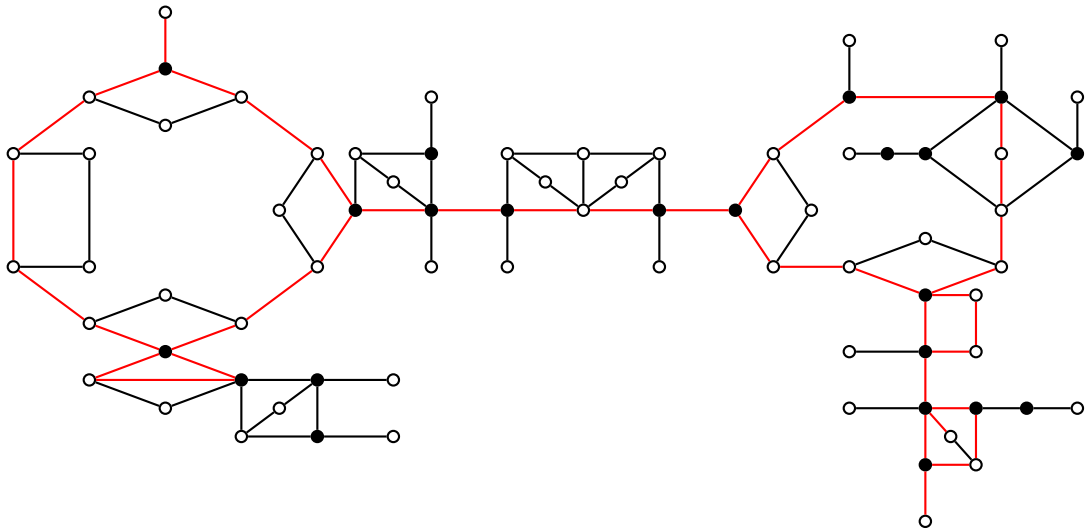
\begin{center}
	\begin{tikzpicture}[xscale=1, yscale=0.75, node distance = 1.4cm,thick]
	\tikzstyle{dot}   =[fill=black, draw=black, circle, inner sep=0.15mm]
	\tikzstyle{vertex}=[  draw=black, circle, inner sep=1.5pt]
	\tikzstyle{apoint}=[  draw=black, circle, fill, inner sep=1.5pt]
	\tikzstyle{dist}  =[fill=white, draw=black, circle, inner sep=2pt]
	\tikzstyle{pinned}=[draw=black, minimum size=10mm, circle, inner sep=0pt]	
	
	\begin{scope}[yshift=-15cm]
	\node[vertex] (a) at (0  ,2) {};
	\node[vertex] (b) at (1  ,3) {};
	\node[apoint] (c) at (2  ,3.5)  {};
	\node[vertex] (d) at (2  ,2.5) {};
	\node[vertex] (e) at (3  ,3) {};
	\node[vertex] (f) at (4  ,2) {};
	\node[apoint] (g) at (4.5  ,1)  {};
	\node[vertex] (h) at (3.5  ,1) {};
	\node[vertex] (i) at (4  ,0) {};
	\node[vertex] (j) at (3  ,-1) {};
	\node[vertex] (k) at (2  ,-0.5) {};
	\node[apoint] (l) at (2  ,-1.5)  {};
	\node[vertex] (m) at (1  ,-1) {};
	\node[vertex] (n) at (0  ,0) {};
	
	\draw[thick,red] (a) -- (b); \draw[thick,red] (c) -- (b); \draw (d) -- (b);
	\draw[thick,red] (c) -- (e); \draw (d) -- (e); \draw[thick,red] (e) -- (f);
	\draw[thick,red] (f) -- (g); \draw (f) -- (h); \draw (h) -- (i);
	\draw[thick,red] (g) -- (i); \draw[thick,red] (i) -- (j); \draw (j) -- (k);
	\draw[thick,red] (j) -- (l); \draw (k) -- (m); \draw[thick,red] (l) -- (m);
	\draw[thick,red] (m) -- (n); \draw[thick,red] (a) -- (n); 
	
	\node[apoint] (o) at (5.5  ,1)  {};
	\node[apoint] (p) at (5.5  ,2)  {};
	\node[vertex] (q) at (4.5  ,2) {};
	\node[vertex] (r) at (5  ,1.5) {};
	\node[apoint] (s) at (6.5  ,1)   {};
	\node[vertex] (t) at (7.5  ,1) {};
	\node[apoint] (u) at (8.5  ,1) {};
	\node[vertex] (v) at (8.5  ,2) {};
	\node[vertex] (w) at (7.5  ,2) {};
	\node[vertex] (x) at (6.5  ,2) {};
	\node[vertex] (y) at (7  ,1.5) {};
	\node[vertex] (z) at (8  ,1.5) {};
	
	\draw[thick,red] (g) -- (o); \draw (o) -- (p); \draw (p) -- (q);
	\draw (q) -- (r); \draw (q) -- (g); \draw (r) -- (o);
	\draw[thick,red] (u) -- (t); \draw[thick,red] (t) -- (s); 
	\draw[thick,red] (o) -- (s);
	\draw (u) -- (v); \draw (v) -- (w); \draw (w) -- (x);
	\draw (t) -- (y); \draw (y) -- (x); \draw (s) -- (x);
	\draw (t) -- (z); \draw (z) -- (v); \draw (t) -- (w);
	
	\node[apoint] (1) at (9.5  ,1)  {};
	\node[vertex] (2) at (10  ,2) {};
	\node[apoint] (3) at (11  ,3) {};
	\node[apoint] (4) at (13  ,3)  {};
	\node[apoint] (5) at (12  ,2) {};
	\node[vertex] (6) at (13  ,2) {};
	\node[apoint] (7) at (14  ,2)  {};
	\node[vertex] (8) at (13  ,1) {};
	\node[vertex] (9) at (13  ,0) {};
	\node[vertex] (10) at (12  ,0.5) {};
	\node[apoint] (11) at (12  ,-0.5)  {};
	\node[vertex] (12) at (11  ,0) {};
	\node[vertex] (13) at (10  ,0) {};
	\node[vertex] (14) at (10.5  ,1) {};
	
	\draw[thick,red] (u) -- (1); \draw[thick,red] (1) -- (2); 
	\draw[thick,red] (2) -- (3);
	\draw[thick,red] (6) -- (4); \draw (5) -- (4); \draw[thick,red] (4) -- (3);
	\draw (4) -- (7); \draw (5) -- (8); \draw[thick,red] (6) -- (8);
	\draw (7) -- (8); \draw[thick,red] (8) -- (9); \draw (9) -- (10);
	\draw[thick,red] (12) -- (11); \draw (12) -- (10); \draw[thick,red] (11) -- (9);
	\draw[thick,red] (12) -- (13); \draw (13) -- (14); \draw (14) -- (2);
	\draw[thick,red] (13) -- (1);

	\node[vertex] (c1) at (2  ,4.5) {};
	\node[vertex] (l1) at (1  ,-2) {};
	\node[apoint] (l2) at (3  ,-2)   {};
	\node[vertex] (l3) at (2  ,-2.5) {};
	\node[vertex] (111) at (12.666  ,-0.5) {};
	\node[vertex] (112) at (12.666  ,-1.5) {};
	\node[apoint] (113) at (12  ,-1.5)  {};
	\node[apoint] (114) at (12  ,-2.5)  {};
	\node[apoint] (115) at (12.666  ,-2.5) {};
	\node[vertex] (116) at (12.333  ,-3) {};
	\node[vertex] (117) at (12.666  ,-3.5) {};
	\node[apoint] (118) at (12  ,-3.5)  {};
	\node[vertex] (119) at (12  ,-4.5) {};
	
	\draw[thick,red] (c) -- (c1); \draw[thick,red] (l1) -- (l); \draw[thick,red] (l2) -- (l);
	\draw[thick,red] (11) -- (111); \draw (l2) -- (l3); \draw (l1) -- (l3);
	\draw[thick,red] (11) -- (113); \draw[thick,red] (111) -- (112); \draw[thick,red] (112) -- (113);
	\draw[thick,red] (116) -- (114); \draw[thick,red] (115) -- (114); \draw[thick,red] (114) -- (113);
	\draw[thick,red] (114) -- (118); \draw[thick,red] (118) -- (117); \draw (117) -- (116);
	\draw[thick,red] (119) -- (118); \draw [thick,red](115) -- (117); \draw[thick,red] (l2) -- (l1);

	\node[vertex] (l4) at (3  ,-3) {};
	\node[apoint] (l5) at (4  ,-3) {};
	\node[apoint] (l6) at (4  ,-2) {};
	\node[vertex] (l7) at (5  ,-2) {};
	\node[vertex] (l8) at (3.5  ,-2.5) {};
	
	\draw (l2) -- (l4); \draw (l2) -- (l6); \draw (l4) -- (l8);
	\draw (l5) -- (l6); \draw (l7) -- (l6); \draw (l6) -- (l8);
	\draw (l4) -- (l5); 
	
	\node[vertex] (p1) at (5.5  ,3) {};
	\draw (p) -- (p1); 
	
	\node[vertex] (n1) at (1  ,0) {};
	\node[vertex] (a1) at (1  ,2) {};
	\draw (n) -- (n1); \draw (a) -- (a1);
	\draw (a1) -- (n1);

	\node[vertex] (3a) at (11  ,4) {};
	\node[vertex] (4a) at (13  ,4) {};
	\node[vertex] (7a) at (14  ,3) {};
	\node[vertex] (5b) at (11  ,2) {};
	\node[apoint] (5a) at (11.5  ,2) {};
	\draw (3) -- (3a); \draw (4) -- (4a); \draw (5) -- (5a);
	\draw (5a) -- (5b); \draw (7) -- (7a);
	\node (o1) at (2,1) {};
	\node (o2) at (2,-3) {};
	\node (o3) at (11.5,1) {};
	
	\node[vertex] (l9) at (5  ,-3) {};
	\node[vertex] (1111) at (14  ,-2.5) {};
	\node[apoint] (1110) at (13.333  ,-2.5) {};
	\draw (l5) -- (l9); \draw (115) -- (1110); \draw (1110) -- (1111);
	
	\node[vertex] (no2) at (5.5  ,0) {};
	\node[vertex] (ns2) at (6.5  ,0) {};
	\node[vertex] (nu2) at (8.5  ,0) {};
	\node[vertex] (1132) at (11  ,-1.5) {};
	\node[vertex] (1142) at (11  ,-2.5) {};
	\draw (o) -- (no2); \draw (ns2) -- (s); \draw(nu2) -- (u);
	\draw (113) -- (1132); \draw (114) -- (1142);
	\end{scope}
	\end{tikzpicture} 
\end{center}

\noindent The above illustration depicts a $K_4$-minor-free graph $H'$ without non-trivial  involutions, together with a subgraph, highlighted in red, that allows for a reduction from~$\paris$. Solid vertices depict articulation points, i.e., vertices that lie in the intersection of at least $2$ biconnected components. Note that each biconnected component of~$H'$ that is not an edge is of one of the three types given in Figure~\ref{fig:introBiconnectedComps}.  Also, each biconnected component of $H'$  has an  involution. These non-trivial involutions prevent us from exploiting the local structure
of the biconnected components to derive $\parP$-hardness. Instead, we will see that the highlighted subgraph is what makes $\parHom{H'}$ hard.

 In  the next section we provide an overview of the general framework that allows us to reduce $\paris$ to $\parHom{H}$.
The structures used in such reductions are 
 captured by the so-called hardness gadgets  introduced by G\"obel, Goldberg and Richerby~\cite{cactus,squarefree}. 
 Prior applications of hardness gadgets could only be  used to construct a reduction from $\paris$ to $\parHom{H}$ if $H$ has certain local substructures,  based around a path or a cycle. 
 In contrast, our analysis will establish 
 global walks such as the one highlighted in~$H'$.
 As far as we can tell, none of the prior machinery~\cite{FJ,cactus,squarefree,BulatovModp} is capable of proving the $\parP$-hardness of $\parHom{H'}$, 
 however, this will follow as a result of our abstract consideration of   global substructures of $K_4$-minor-free graphs.

\section{Warm-up:  useful Ideas from Previous Papers --- Retractions and Hardness Gadgets }
\label{sec:warmup}

 Instead of directly reducing $\paris$ to $\parHom{H}$, it is
 useful to consider the intermediate problem $\parRet{H}$, the problem
 of counting \emph{retractions}\footnote{In some definition versions a retraction is surjective. However, for algorithmic problems this surjectivity requirement is not
 	important \cite{retr,FGZRet}} to~$H$, modulo~$2$. 
 Given a graph~$H$, a \emph{partially $H$-labelled graph} $J=(G,\tau)$ 
consists of an \emph{underlying graph}~$G$ and a corresponding \emph{pinning
function}~$\tau$,  which is a partial function from $V(G)$ to~$V(H)$. 
 A homomorphism from~$J$ to~$H$ is
 a homomorphism $h$ from $G$ to $H$ such that, for all vertices $v$ in the domain of $\tau$, $h(v) = \tau(v)$.  
 
 A homomorphism from a partially $H$-labelled graph~$J$ to~$H$ is
 also called a \emph{retraction} to~$H$ \label{page:retractions} because 
 we can think of the pinning function~$\tau$ as a way of identifying an induced subgraph~$H$ of~$G$
 which must ``retract'' to~$H$ under the action of the homomorphism --- see \cite{retr} for details.
 We use $\parRet{H}$ 
 \label{page:ParRet}
 to denote the computational problem of computing the number of homomorphisms from~$J$ to~$H$, modulo~$2$,
 given as input a partially $H$-labelled graph~$J$.

  It is known~\cite{squarefree} that $\parRet{H}$ reduces to $\parHom{H}$ whenever $H$  is involution-free. Since~$\tau$ allows us to pin vertices of $G$ to vertices of $H$ arbitrarily, it is much easier to construct a reduction from $\paris$ to $\parRet{H}$ than to construct
  a direct reduction from~$\paris$ to~$\parHom{H}$.

Consider the following example. Suppose that $H$ is the $4$-vertex path $(o,s,i,x)$ 
and that  our goal is to reduce $\paris$ to $\parRet{H}$. Let $G$ be 
an input to~$\paris$. That is, $G$ is a graph whose independent sets we wish to count, modulo~$2$.
For ease of presentation, suppose that~$G$ is bipartite,\footnote{The case of general graphs will be discussed later in the paper.} that is, the vertices of $G$ can be partitioned into two independent sets~$U$ and~$V$. Let $\widehat{G}$ be the graph obtained from $G$ by adding two additional vertices $u$ and $v$, and by connecting $u$ to all vertices in $U$, and $v$ to all vertices in $V$, respectively. 
Let $\tau$ be the pinning function defined by
  $\tau(u)=s$ and $\tau(v)= i$. We provide an illustration of the construction in Figure~\ref{fig:intro_red_1}.

\begin{figure}[h]
	\centering
	\begin{tikzpicture}[xscale=1.5,yscale=1.4, node distance = 1.4cm,thick]
	\tikzstyle{dot}   =[fill=black, draw=black, circle, inner sep=0.15mm]
	\tikzstyle{vertex}=[draw=black, fill=black, circle, inner sep=1.5pt]
	\tikzstyle{terminal}=[fill=black, draw=black, circle, inner sep=1.5pt]
	\tikzstyle{dist}  =[fill=white, draw=black, circle, inner sep=2pt]
	\tikzstyle{pinned}=[draw=black, minimum size=10mm, circle, inner sep=0pt]	
	\begin{scope}
	\node[vertex] (o) at (0  ,-1.5) [label=-90:{$o$}] {};
	\node[vertex] (s) at (1  ,-1.5) [label=-90:{$s$}] {};
	\node[vertex] (i) at (2  ,-1.5) [label=-90:{$i$}] {};
	\node[vertex] (x) at (3  ,-1.5) [label=-90:{$x$}] {};
	
	\node (U) at (1,3.5) {\large $U$}; \node (U) at (2,3.5) {\large $V$};
	\node (U) at (1.5,2) {\LARGE $G$};
	\node (U) at (0.25,0.5) {\large $\tau$}; 
	\node (U) at (2.75,0.5) {\large $\tau$};

	\node[vertex] (u) at (0  ,2) [label=-180:{$u$}] {};
	\node[vertex] (v) at (3  ,2) [label=0:{$v$}] {};
	\node[vertex] (u1) at (1  ,3)  {};
	\node[vertex] (u2) at (1  ,1)  {};
	\node[vertex] (v1) at (2  ,3)  {};
	\node[vertex] (v2) at (2  ,1)  {};
	
	\node[dot] (d) at (1,1.5) {};
	\node[dot] (d) at (1,2) {};
	\node[dot] (d) at (1,2.5) {};
	\node[dot] (d) at (2,1.5) {};
	\node[dot] (d) at (2,2) {};
	\node[dot] (d) at (2,2.5) {};
	
	\draw[rounded corners] (0.75, 3.25) rectangle (2.25, 0.75) {};
	\draw (u1) -- (v1); \draw (u2)--(v2);
	\draw (u1) -- (1.875,2.25); \draw (u2)--(1.825,1.5);
	\draw (v1) -- (1.125,2.5); \draw (v2)--(1.125,1.75);
	\draw (u1) -- (u); \draw (u2)--(u); \draw (v) -- (v1); \draw (v)--(v2);
	\draw (o) -- (s); \draw (s)--(i); \draw (i) -- (x); 
	
	\draw[dashed,->] (u) -- (0.95,-1.4); \draw[dashed,->] (v) -- (2.05,-1.4);
	
	\end{scope}
	
	\begin{scope}[xshift = 5cm]
	\node[vertex] (o) at (0  ,-0.5) [label=-180:{$o$}] {};
	\node[vertex] (s) at (1.5  ,-0.5) [label=45:{$s$}] {};
	\node[vertex] (i) at (1.5  ,-1.5) [label=-90:{$i$}] {};
	\node[vertex] (x) at (3  ,-1.5) [label=-90:{$x$}] {};
	\node[vertex] (y) at (0  ,-1.5) [label=-90:{$y$}] {};
	\node[vertex] (z) at (3  ,-0.5) [label=-45:{$z$}] {};
	
	\node (U) at (1,3.5) {\large $U$}; \node (U) at (2,3.5) {\large $V$};
	\node (U) at (1.5,2) {\LARGE $G$};
	\node (U) at (0.25,0.75) {\large $\tau$}; 
	\node (U) at (2.85,0.75) {\large $\tau$}; 
	\node (U) at (-0.7,0.75) {\large $\tau$}; 
	\node (U) at (3.5,0.75) {\large $\tau$};

	\node[vertex] (u) at (0  ,2.5) [label=-180:{$u_1$}] {};
	\node[vertex] (v) at (3  ,2.5) [label=0:{$v_1$}] {};
	\node[vertex] (u') at (0  ,1.5) [label=-180:{$u_2$}] {};
	\node[vertex] (v') at (3  ,1.5) [label=0:{$v_2$}] {};
	\node[vertex] (u1) at (1  ,3)  {};
	\node[vertex] (u2) at (1  ,1)  {};
	\node[vertex] (v1) at (2  ,3)  {};
	\node[vertex] (v2) at (2  ,1)  {};
	
	\node[dot] (d) at (1,1.5) {};
	\node[dot] (d) at (1,2) {};
	\node[dot] (d) at (1,2.5) {};
	\node[dot] (d) at (2,1.5) {};
	\node[dot] (d) at (2,2) {};
	\node[dot] (d) at (2,2.5) {};
	
	\draw[rounded corners] (0.75, 3.25) rectangle (2.25, 0.75) {};
	\draw (u1) -- (v1); \draw (u2)--(v2);
	\draw (u1) -- (1.875,2.25); \draw (u2)--(1.825,1.5);
	\draw (v1) -- (1.125,2.5); \draw (v2)--(1.125,1.75);
	\draw (u1) -- (u); \draw (u2)--(u); \draw (v) -- (v1); \draw (v)--(v2);
	\draw (u1) -- (u'); \draw (u2)--(u'); \draw (v') -- (v1); \draw (v')--(v2);
	\draw (o) -- (s); \draw (s)--(i); \draw (i) -- (x); 
	\draw (o) -- (y); \draw (y)--(i); \draw (s) -- (y); 
	\draw (x) -- (z); \draw (s)--(z); 
	
	\draw[dashed,->] (u') -- (1.45,-0.4); \draw[dashed,->] (v') -- (1.575,-1.4);
	\draw[dashed,->] (u) to[in=135,out=-135] (-0.05,-1.45);
	\draw[dashed,->] (v) to[in=45,out=-45] (3.1,-0.4);
	
	\end{scope}
	\end{tikzpicture} 
	\caption{\label{fig:intro_red_1} Illustration of the reduction from (bipartite) $\paris$ to $\parRet{H}$ 
where $H$ is the $4$-vertex path	  \emph{(left)}, and 
$H$ is the graph $H_2$ from page~\pageref{page:H2}   \emph{(right)}.}
\end{figure}

 Observe that any homomorphism~$\varphi$ 
from~$(\widehat{G},\tau)$ to~$H$  must map every vertex in $U$ to either~$o$ or~$i$, and every vertex in $V$ to either~$s$ or~$x$. 
Since $H$ has no edge from~$o$ to~$x$, the definition of homomorphism ensures that
 $\varphi^{-1}(o)\cup \varphi^{-1}(x)$ is an independent set of $G$.
 It is easy to verify that the function $\varphi \mapsto \varphi^{-1}(o)\cup \varphi^{-1}(x)$ 
 is a bijection between 
 the homomorphisms from~$(\widehat{G},\tau)$ to~$H$
 and the independent sets of~$G$,  which gives a reduction from (bipartite) $\paris$ to $\parRet{H}$.
 
The observant reader might notice that the $4$-vertex path    has a non-trivial  involution, and thus, we cannot 
further reduce   $\parRet{H}$ to $\parHom{H}$ in this case.\footnote{In fact, the problem $\parHom{H}$ is trivial 
when $H$ is the $4$-vertex path  since the number of homomorphisms will always be even.} 
However, the construction works for \emph{any} graph $H$  with an induced path $(o,s,i,x)$ such that $s$ and~$i$ each only have two neighbours.

The notion of a \emph{hardness gadget}, which we formally introduce in Section~\ref{sec:prelims}, is essentially a generalisation of the previous construction. For example, we could substitute each of $o$, $s$, $i$ and $x$ with an odd number of copies, since we are only interested in the parity of the number of independent sets. Furthermore, we could identify $o$ and $x$, since we only need the edge $\{o,x\}$ to be absent in $H$. A more sophisticated generalisation is obtained by observing that we can, to some extent, substitute the edges $\{o,s\}$, $\{s,i\}$ and $\{i,x\}$ with more complicated graphs, e.g.\ with  length-$2$ paths, if we substitute the edges in $\widehat{G}$ accordingly. 
Finally,   observe that   the 
construction $(\widehat{G},\tau)$ uses the partial function~$\tau$  in a very simple manner: By adding a common neighbour $u$ for all vertices in $U$ and setting $\tau(u)=s$, 
the construction enforces 
the constraint that any homomorphism  from $(\widehat{G},\tau)$ to~$H$ must map every vertex in $U$ to a neighbour of $s$. More sophisticated  
constructions
 will allow us to enforce much stronger constraints on  homomorphisms.
 We will need this flexibility to
  construct reductions from $\paris$ to $\parRet{H}$ for more general graphs $H$.

We conclude by making a generalisation explicit for one further example --- the graph $H_2$  from page~\pageref{page:H2}. We provide a more convenient drawing of $H_2$, including a labelling of its vertices and an illustration of the reduction in Figure~\ref{fig:intro_red_1}. Again, we will assume for ease of presentation that the input~$G$ to~$\paris$ is bipartite. To construct~$\widehat{G}$,
 we add two additional vertices $u_1$ and $u_2$ and make them adjacent to every vertex in $U$.
 Similarly, we add two additional vertices $v_1$ and $v_2$ and make them adjacent to every vertex in $V$.  
 Let~$\tau$ be the pinning function defined by $\tau(u_1)=y$, $\tau(u_2)=s$, $\tau(v_1)=z$, and $\tau(v_2)=i$. 
 
 Consider any homomorphism~$\varphi$ from $(\widehat{G},\tau)$ to~$H_2$.  Since $\varphi$ is edge-preserving, it must map every vertex in $U$ to a common neighbour of $s$ and $y$ in $H_2$. Consequently, $\varphi(U)\subseteq\{o,i\}$. Similarly, we obtain $\varphi(V)\subseteq \{s,x\}$. Again, it is easy to see that the mapping $\varphi \mapsto \varphi^{-1}(o)\cup \varphi^{-1}(x)$ is a bijection 
 between 
 the homomorphisms from~$(\widehat{G},\tau)$ to~$H$
 and the independent sets of~$G$,  which gives a reduction from (bipartite) $\paris$ to $\parRet{H}$.  

Note that the second example, while being less straightforward than the first, is still a very simple reduction. The proof of Theorem~\ref{thm:main} requires us to consider much more intricate ``hardness gadgets''; the necessary tools will be carefully introduced in Sections~\ref{sec:toolbox} and~\ref{sec:chordalBipartiteComps}.

\section{Proof Outline and Organisation of the Paper}

We start with the formal definitions  that we need in Section~\ref{sec:prelims}; in particular we set up the framework of hardness gadgets.
Section~\ref{sec:toolbox}, our ``toolbox'', presents the most important class of hardness gadgets that we use.  

Sections~\ref{sec:chordalBipartiteComps}-\ref{sec:k4MinorFreeMain} constitute the proof of our main result and should be considered the technical core of this paper. In Section~\ref{sec:chordalBipartiteComps} we deal with biconnected $K_4$-minor-free graphs that are additionally chordal bipartite graphs
(that is, they have the property that every induced cycle is a square). The reason for our separate treatment of these graphs is   that our main gadget from Section~\ref{sec:toolbox} cannot be applied to graphs without an induced cycle of length $\neq 4$. We identify two families of such graphs, \emph{impasses} and \emph{diamonds}, that prevent us from  constructing a local hardness gadget; examples of an impasse and a diamond are depicted in Figure~\ref{fig:introBiconnectedComps}.

After that, we dedicate Section~\ref{sec:generalisingCaterpillar} to the analysis of $K_4$-minor-free graphs that contain certain sequences of biconnected components, each of which is either an edge, an impasse, or a diamond. In Section~\ref{sec:K4freeGraphs} we consider biconnected $K_4$-minor-free graphs that are not necessarily chordal bipartite. We identify 
another family of graphs that does not allow for a local, i.e., an ``internal'', hardness gadget; we call such graphs \emph{obstructions}; obstructions always contain an induced cycle of length other than $4$, and an example of an obstruction is depicted in Figure~\ref{fig:introBiconnectedComps}. 

In combination, Sections~\ref{sec:generalisingCaterpillar} and~\ref{sec:K4freeGraphs} reveal the structure of 
involution-free
$K_4$-minor-free graphs  that do not allow for a local hardness gadget. In Section~\ref{sec:k4MinorFreeMain} we  use this structure, which allows us to constructively prove the existence of \emph{global} hardness gadgets for all remaining $K_4$-minor-free graphs without non-trivial involutions. Our main theorem, Theorem~\ref{thm:main},  follows. 

In Sections~\ref{sec:degreeAtmost3} and~\ref{sec:parLHom} we explore the applicability of our machinery to further classes of graphs and problems: Section~\ref{sec:degreeAtmost3} presents a full classification for counting homomorphisms to graphs of degree at most $3$, modulo $2$. Section~\ref{sec:parLHom} considers the problem of counting list homomorphisms, modulo $2$, a variation of the homomorphism problem that generalises retractions as follows: Let $H$ be a fixed graph. The problem $\parLHom{H}$ expects as input a graph $G$ and a function $\tau$ that maps every vertex of $G$ to a list of vertices of $H$. The goal is then to compute the parity of the number of homomorphisms from $G$ to $H$ which additionally map every vertex $v$ of $G$ to a vertex contained in $\tau(v)$. We provide a full classification of $\parLHom{H}$ for all graphs $H$, even if self-loops are allowed.

Finally, in Section~\ref{sec:index},
we provide an index containing the most important symbols and definitions.

	\section{Preliminaries}\label{sec:prelims}
An index of notation and terminology is in Section~\ref{sec:index}. 
Given a positive integer $q$ let $[q]=\{1, \dots, q\}$. Given a finite set $S$, we write $|S|$ for its cardinality.
	
	\paragraph*{Graph theory}
	Graphs in this work are simple, that is, without multiple edges, and do not contain self-loops, unless stated otherwise. The size of a graph $G$ is defined as $\abs{G}=\abs{V(G)}+\abs{E(G)}$. Given a graph $H$ and a subset $S$ of its vertices, we write $H[S]$ for the subgraph of $H$ \emph{induced} by $S$.\label{page:inducedgraph}
	
	Given a non-negative integer $k$, a \emph{walk} of length $k$ in a graph $H$ is a sequence of (not necessarily distinct) vertices $(v_0, \dots, v_k)$ such that, for all $i\in [k]$, $\{v_{i-1}, v_{i}\}\in E(H)$.   The walk is \emph{closed} if $v_0=v_k$. Note that for $k=0$, the single vertex $(v_0)$ is a closed walk of length $0$. A \emph{path} of length $k$ is a walk of length $k$ for which $v_0, \dots, v_k$ are distinct. 
For $k\geq 3$, a \emph{cycle} of length $k$ is a closed walk of length $k$  such that $v_1, \dots, v_k$ are distinct. A \emph{square} is a cycle of length $4$. 
	
	For $i, j \in \{0, \dots, k\}$ with $i\le j$, we say that $(v_i, v_{i+1} \dots, v_j)$ is a \emph{subwalk} of $(v_0, \dots, v_k)$. 
	For vertices $u,v \in V(H)$, $\dist_H(u,v)$ is the length of a shortest path between $u$ and $v$. 
	 
	\begin{defn}[chordal bipartite graph, see e.g.~\cite{ChordalBipartite99}]\label{def:chordalbipartite}
		A graph in which every induced cycle is a square is called a \emph{chordal bipartite graph}.
	\end{defn}

\noindent Given a graph $H$ and a vertex $v\in V(H)$, we write $\NH{v}$ for the \emph{neighbourhood} of~$v$ in~$H$
and we write $\deg_H(v)$ for the \emph{degree} of~$v$.
 That is, 
 $\NH{v} = \{u\in V(H) \mid \{u,v\} \in E(H)\}$ and 
   $\deg_H(v) =\abs{\NH{v}}$.\label{page:degree}
	Given a subset $S$ of $V(H)$, we set $\NH{S}= \bigcap_{v\in U} \NH{v}$ and note that $\NH{v} = \NH{\{v\}}$.\ \label{page:neighbourhood}
	
	\begin{defn}[walk-neighbour-set]\label{def:walkneighbourset}
		Given a closed walk $W = (w_0,\ldots,w_{q-1}, w_0)$  
		in a graph~$H$
		we use $N_{W,H}(w_i)$ to denote $\Gamma_H(w_{i-1}) \cap \Gamma_H(w_{i+1})$, where the indices are taken modulo~$q$.
		We refer to the sets $N_{W,H}(w_0),\ldots,N_{W,H}(w_{q-1})$ as the 
		\emph{walk-neighbour-sets} of~$W$ in~$H$.
	\end{defn}

	\begin{defn}[articulation point, biconnected, block-cut tree]\label{def:blockcuttree}
		An \emph{articulation point} of a graph is a vertex whose removal increases the number of connected components.
		A graph is \emph{biconnected} if it has at least $2$~vertices and has no articulation point.
		A \emph{biconnected component} is a maximal biconnected subgraph.
		
		Let $H$ be a connected graph. The \emph{block-cut tree} of $H$ is the tree $\bc(H)$ that has a vertex for each biconnected component of~$H$ (such vertices are called \emph{blocks}) and a vertex
		for each articulation point of~$H$ (such vertices are also called \emph{cut vertices})  such that $T$ has an edge between each biconnected component $B$ and each articulation point $a$ in $B$.
	\end{defn}

	\paragraph{Partially labelled graphs}
	Let $H$ be a graph. Recall from Section~\ref{sec:warmup} that
   a  {partially $H$-labelled graph} $J=(G,\tau)$\label{page:partiallyHlabelled}
	consists of an  {underlying graph}~$G$ and a corresponding  {pinning
		function}~$\tau$,\label{page:pinningfunction} which is a partial function from $V(G)$ to~$V(H)$.
	A vertex~$v$ in the domain of the pinning function is said to be
	\emph{pinned}, \emph{pinned to $\tau(v)$}, or a \emph{$\tau(v)$-pin}.
	We write
	partial functions as sets of pairs, for example, writing $\tau =
	\{a\mapsto s,b\mapsto t\}$ for the partial function~$\tau$ with domain $\{a,b\}$ such that $a$ is an $s$-pin and $b$ is a $t$-pin. The size of a partially $H$-labelled graph $J=(G,\tau)$ is defined as $\abs{J}=\abs{G}$.

	\paragraph*{Homomorphisms and Counting (mod 2)} 
	
	Let $G$ and $H$ be graphs. 
	Then $\hom{G}{H}$\label{page:homsGH} denotes the \emph{set of homomorphisms} from $G$ to $H$ and $\hom{J}{H}$ denotes the set of homomorphisms from $J$ to $H$.

	It will sometimes be convenient to consider a graph $G$ together with some number of
	distinguished vertices $x_1, \dots, x_r$ of $G$. We denote such a graph by
	$(G, x_1, \dots, x_r)$. The distinguished vertices need not
	be distinct.
	A homomorphism from a graph $(G, x_1, \dots,
	x_r)$ to $(H, y_1, \dots, y_r)$ is a homomorphism~$h$ from $G$
	to~$H$ with the property that, for each $i\in[r]$, $h(x_i)=y_i$.
	Correspondingly, we write $\hom{(G, x_1, \dots, x_r)}{(H, y_1, \dots, y_r)}$ for the set of such homomorphisms.
	
	Given a partially
	labelled graph~$J=(G,\tau)$ and distinguished vertices $x_1, \dots, x_r$ of $G$ that are not in the domain of $\tau$, a homomorphism from $(J, x_1, \dots, x_r)$ to $(H,
	y_1, \dots, y_r)$ is a homomorphism from
	$J'=(G, \tau \cup \{x_1\mapsto y_1, \dots, x_r\mapsto y_r\})$ to~$H$.
	 The set of such homomorphisms is denoted by  $\hom{(J, x_1, \dots, x_r)}{(H, y_1, \dots, y_r)}$\label{page:partialhomset}.

\paragraph*{Useful tools}

The following theorem of \Goebel, Goldberg and Richerby will be of crucial importance in this work, as it will allow us to derive hardness of $\parHom{H}$ from hardness of $\parRet{H}$.

	\begin{thm}[{\cite[Theorem 3.1]{squarefree}}]\label{thm:ETH_RetToHoms}
Let $H$ be an involution-free graph. Then there is an algorithm with oracle access to $\parHom{H}$ that takes as input a partially $H$-labelled graph $J$ and 
	computes $\abs{\hom{J}{H}}\mod 2$ in time $\poly(\abs{J})$.
	The size of the input to every oracle query is 
	$O(\abs{J})$.\end{thm}
	
The statement of Theorem~\ref{thm:ETH_RetToHoms} in
\cite[Theorem 3.1]{squarefree} does not mention the fact that 
the size of the input to every oracle query is 
$O(\abs{J})$. Nevertheless, it is easy to see, by examining the proof in~\cite{squarefree} that this linearity requirement is met (without
making any changes to the proof).
The reason 
that we introduce this linearity constraint  is	  
so
that our hardness results  can also rule out subexponential-time algorithms for $\parHom{H}$ in the $\parP$-hard cases, 
using the rETH.

The following theorem of Faben and Jerrum will also be useful, as it will allow us to focus on
connected graphs. The statement of \cite[Theorem 6.1]{FJ} does not mention the  
linearity requirement on the size of oracle queries, but this requirement does not present any difficulties. Faben and Jerrum's proof is given in a slightly different setting (pinning to orbits of vertices of~$H$ rather than to vertices) so, for completeness, we give a short proof.

	\begin{lem}[{\cite[Theorem 6.1]{FJ}}]\label{lem:ETH_connectivity}
		Let $H$ be an involution-free graph and let $H'$ be a connected component of $H$. Then there exists an algorithm with oracle access to $\parHom{H}$ that takes as input a graph $G$ and computes $\abs{\hom{G}{H'}}\mod 2$ in time $\poly(\abs{G})$. The size of every oracle query is $O(\abs{G})$.	
	\end{lem}
	
	\begin{proof}
		Let $G$ be a graph. If $G$ is the empty graph then the algorithm returns $1$, which is the number of homomorphisms from $G$ to $H'$.
		Otherwise, there exists a vertex $u\in V(G)$. For each $v\in V(H')$ we define the partially $H'$-labelled $J_v=(G,\{u \mapsto v\})$. Note that $\abs{\hom{G}{H'}}= \sum_{v \in V(H')} \abs{\hom{J_v}{H}}$.
		
		By Theorem~\ref{thm:ETH_RetToHoms}, there is an algorithm~$A$  with oracle access to $\parHom{H}$ that takes as input a partially $H$-labelled graph $J$ and computes $\abs{\hom{J}{H}}\mod 2$ in time $\poly(\abs{J})$ such that the size of every oracle query is bounded by $O(\abs{J})$. 
		Our algorithm uses algorithm~$A$ as a subroutine to compute the parity of~$\abs{\hom{J_v}{H}}$ for each $v \in V(H')$. This requires $\abs{V(H')}$ executions of the subroutine $A$. Thus, the algorithm runs in time 
$$O\Big(\sum_{v\in V(H')}\poly(|J_v|)\Big)=\poly(\abs{G}).$$ 
Moreover, for each $v\in V(H')$, the size of each $\parHom{H}$ oracle query is bounded by $O(\abs{J_v})=O(\abs{G})$.
	\end{proof}

\paragraph*{Hardness Gadgets}	
The following is a slightly generalised version of the \emph{hardness gadget} introduced in~\cite[Definition 4.1]{squarefree}.
The only difference between their definition and ours is that they require
the sets~$I$ and~$S$ to have size~$1$. 
	
	\begin{defn} \cite[Definition 4.1]{squarefree}
		\label{defn:hardness-gadget}
		A \emph{hardness gadget} $(I,S,(J_1,y),(J_2,z),(J_3,y,z))$ for a
		graph~$H$ consists
		of odd-cardinality sets $I,S \subseteq V(H)$ together with three connected, partially $H$-labelled
		graphs with distinguished vertices  
		$(J_1,y)$, $(J_2,z)$ and $(J_3,y,z)$ 
		that satisfy certain
		properties
		as explained below.
		Let
		\begin{align*}
		\Oy &= \{ a \in V(H) \mid |\hom{(J_1, y)}{(H,a)}|\text{ is odd}\},\\
		\Oz &= \{ b \in V(H) \mid |\hom{(J_2, z)}{(H,b)}|\text{ is odd}\}, \mbox{ and}\\
		\Sigma_{a,b} &=\hom{(J_3,y,z)}{(H,a,b)}\,.
		\end{align*}
		The properties 
		that
		we require are 
		the following.
		\begin{enumerate}
			\item $|\Oy|$ is even and $I\subset \Oy$.
			\item $|\Oz|$ is even and $S\subset \Oz$.
			\item 
			For each $i\in I$, $o\in \Oy\setminus I$, $s\in S$ and each $x\in \Oz\setminus S$,
			\begin{itemize}
				\item$|\Sox|$ is even.
				\item $|\Sis|$, $|\Sos|$ and $|\Six|$ are odd.
			\end{itemize}
		\end{enumerate}
	\end{defn}

The following theorem of \Goebel, Goldberg and Richerby
establishes intractability of $\parRet{H}$  whenever $H$ has a hardness gadget.
	\begin{thm}[{\cite[Theorem 4.2]{squarefree}}]
		\label{thm:hardness-gadget} 
	Let $H$ be an involution-free graph that has a hardness gadget. 
		Then $\parRet{H}$ is $\parP$-hard. Also, assuming the randomised Exponential Time Hypothesis, $\parRet{H}$   cannot be solved in time $\exp(o(|J|))$.

\end{thm}

\begin{proof}
Although  the hardness gadgets from~\cite{squarefree} 
are more constrained than the ones that we use, the proof of 
\cite[Theorem 4.2]{squarefree} establishes the $\parP$-hardness in Theorem~\ref{thm:hardness-gadget} with only very minor changes, which we now
describe.

As noted in the introduction,
Valiant~\cite{Valiant2006} showed that the problem $\paris$ is $\parP$-complete.
The proof of \cite[Theorem 4.2]{squarefree}
gives a polynomial-time Turing reduction from 
$\paris$ to~$\parRet{H}$.
The reduction uses~$G$ and the hypothesised hardness gadget for~$H$
to construct a partially $H$-labelled graph~$J$ such that the
number of independent sets of $G$, which we denote $\abs{\calI(G)}$,
 is equal, modulo~$2$, to 
 $\abs{\hom{J}{H}}$. 
The reduction concludes by making a single oracle call to $\parRet{H}$ with input~$J$.

In our case, the construction of~$J$ is exactly as it is in~\cite{squarefree}.
The proof that 
$\abs{\calI(G)} = 
\abs{\hom{J}{H}}\mod 2$  needs only a very minor modification to account for the fact that the sets~$I$ and~$S$ in the hardness gadget
 may have more than one element.
 At some point in the proof of \cite{squarefree}, it is argued that a certain quantity $n(a,a')$ is even if $a$ and $a'$ are both in~$I$, and odd otherwise.
 This is still true even when~$I$  and~$S$ contain more than one element --- it follows from item~3 in the definition of hardness gadget
 (and from the fact that~$I$ and~$S$ have odd cardinality).

The final sentence in the statement of Theorem~\ref{thm:hardness-gadget}, asserting
that $\parRet{H}$ cannot be solved in time $\exp(o(|J|))$ unless the rETH fails,
was not contained in the original theorem of~\cite{squarefree},
however it follows immediately from the fact that 
 $|J| = O(|G|)$ (which is easily checked) 
 and from the fact that 
$\paris$ cannot be solved in time $\exp(o(|G|))$, unless the rETH fails,
which was proved by Dell, Husfeldt, Marx, Taslaman and Wahlen~\cite{DellHMTW14}. \footnote{In more detail,
 Dell et al.\ established that counting independent sets cannot be done in time $\exp(o(|E(G)|))$, unless the rETH fails~\cite[Theorem~1.2]{DellHMTW14}. 
 They point out explicitly that their reduction also works in the case of counting modulo $2$. 
 Furthermore, their reduction always yields a graph without isolated vertices --- for such graphs we have $|E(G)|=\Theta(|G|)$.}
 
 More precisely, for establishing the conditional lower bound for $\parRet{H}$, let us assume for contradiction that we can solve $\parRet{H}$ in time $\exp(o(|J|))$. We obtain an algorithm for $\paris$ that, on input $G$, runs the (polynomial-time) Turing-reduction from~\cite[Theorem 4.2]{squarefree} and then simulates each oracle query $J$ for $\parRet{H}$ in time $\exp(o(|J|))$; note that this simulation is possible by our assumption. Since each oracle query $J$ has size at most $O(|G|)$, the total running time of our algorithm for $\paris$ is bounded by $\mathsf{poly}(|G|)\cdot \exp(o(|G|)) = \exp(o(|G|))$, contradicting rETH. This concludes the proof of the conditional lower bound.
 \end{proof}	
	
	\section{Toolbox}\label{sec:toolbox}

\subsection{Path Gadget}

We will use the following path gadget, which  is called a   ``caterpillar gadget'' in \cite{squarefree}.

\begin{defn}
\label{defn:caterpillar} 
Given a path $P=(v_0,\dots, v_q)$ in~$H$ with $q\geq 1$,
define the \emph{path gadget} $J_P=(G,\tau)$ as follows.
$V(G)=\{u_1, \dots, u_{q-1}, w_1,\dots, w_{q-1},y,z\}$
and $G$ is the path $(y,u_1,\ldots,u_{q-1},z)$ together with edges $\{u_j,w_j\}$ for $j\in[q-1]$.
 $\tau=\{w_1\mapsto v_1, \dots, w_{q-1}\mapsto v_{q-1}\}$.
\end{defn}

We will use the following lemma of \Goebel, Goldberg and Richerby.
The original lemma was stated for square-free graphs, but the proof only uses
the fact that no edge of~$P$ is part of a square in~$H$.

\begin{lem}[{\cite[Lemma 4.5]{squarefree}}]\label{lem:SquarefreePathPrehardness}
For an integer $q\ge 1$, let $P=(v_0, \dots, v_q)$ be a path in a graph $H$. Suppose that no edge of $P$ is part of a square in $H$ and that $\deg_H(v_j)$ is odd for all $j\in [q-1]$. Let $\Omega_y\subseteq \NH{v_0}$ and $\Omega_z \subseteq \NH{v_q}$, with $I=\{v_1\}\subset \Omega_y$ and $S=\{v_{q-1}\}\subset \Omega_z$. For $i=v_1$, $s=v_{q-1}$ and for each $o\in \Omega_y\setminus I$ and $x\in \Omega_z\setminus S$ we have the following:
	\begin{myitemize}
		\item $\abs{\hom{(J_p, y,z)}{(H,o,x)}}=0$,
		\item $\abs{\hom{(J_p, y,z)}{(H,o,s)}}=1$,
		\item $\abs{\hom{(J_p, y,z)}{(H,i,x)}}=1$, and
		\item $\abs{\hom{(J_p, y,z)}{(H,i,s)}}$ is odd.
\end{myitemize} 
\end{lem}

\subsection{Cycle Gadget}

We will use the following cycle gadget, which is a generalisation of  the cycle gadget in~\cite{BulatovModp}.

\begin{defn}[Cycle gadget]\label{def:cyclegadget}
	For an integer $q\ge 3$, let $\calC=(\calC_0, \dots, \calC_{q-1})$ where, for $i=0, \dots, q-1 $, $s_i$ is a positive integer and $\calC_i=\{c_i^1, \dots, c_i^{s_i}\}$ is a set of $s_i$ vertices. We define the \emph{cycle gadget} $J_{\calC}=(G,\tau)$ as follows (see Figure~\ref{fig:cycleGadget}). For $i=0, \dots, q-1$, let $U_i=\{u_i^1, \dots, u_i^{s_i}\}$ be a set of $s_i$ vertices. Then $V(G)=\{v_0, \dots, v_{q-1}\}\cup U_0 \cup \dots \cup U_{q-1}$ (where all named vertices are assumed to be distinct) and $E(G)=\{\{v_i, v_{i+1 \bmod {q}}\} \mid i \in \{0, \dots, q-1\} \} \cup \{\{v_i, u_i^j\} \mid i \in \{0, \dots, q-1\}, j\in \{1, \dots, s_i\} \}$. $\tau =\{u_i^j\mapsto c_i^j \mid \forall i \in \{0, \dots q-1\}, j\in \{1, \dots s_i\}\}$.
	
	\begin{figure}[ht]
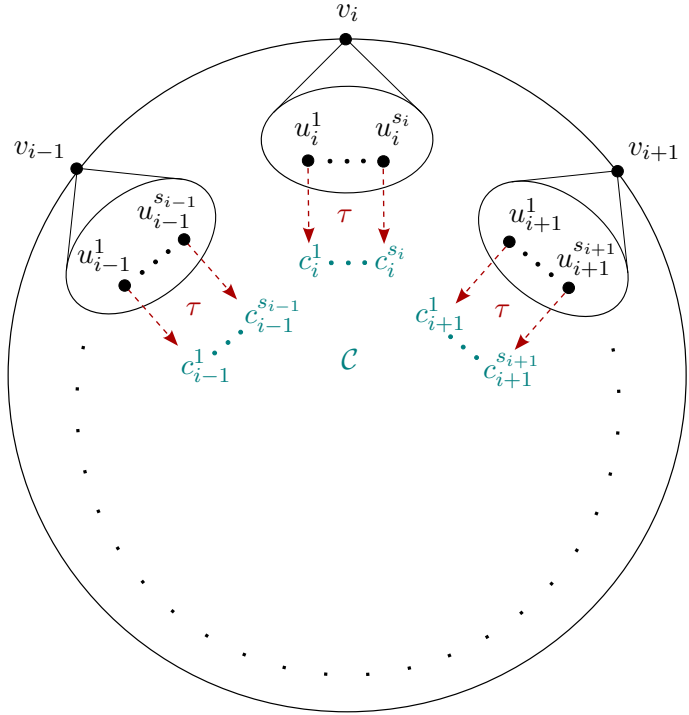

		\centering
        \svg{0.8\linewidth}{cyclegadget}
		\caption{The cycle gadget $J_{\calC}$ is depicted in black. The corresponding pinning function $\tau$ is depicted in red, it maps to the vertices of $\calC$, depicted in blue. 
		}
		\label{fig:cycleGadget}
	\end{figure}
\end{defn}	

In fact, we will also need a further generalisation of the cycle gadget from Definition~\ref{def:cyclegadget}.

\begin{defn}[Generalised cycle gadget]\label{def:generalisedcyclegadget}
	Intuitively, we generalise the cycle gadget by attaching, at each vertex $v_i$, a gadget $J_i$.
	Let $H$ be a graph.
	For an integer $q\ge 3$, let $\calC=(\calC_0, \dots, \calC_{q-1})$ where, for $i=0, \dots, q-1 $, $s_i$ is a positive integer and $\calC_i=\{c_i^1, \dots, c_i^{s_i}\}$ is a set of $s_i$ vertices of $H$. Let $J_{\calC}$ be the cycle gadget from Definition~\ref{def:cyclegadget}. Let $(J_0,z_0), \dots, (J_{q-1},z_{q-1})$ be partially $H$-labelled graphs with distinguished vertices.
	Then the \emph{generalised cycle gadget} $J(J_{\calC}, J_0, \dots, J_{q-1})$ is the gadget obtained from $J_\calC, J_0, \dots, J_{q-1}$ by identifying, for each $i\in \{0, \dots,q-1\}$ the vertex $v_i$ from $J_\calC$ with the vertex $z_i$ from $J_i$.  
\end{defn}

\begin{lem}\label{lem:CycleGadget+}
	For an integer $q\ge 3$, let $H$ be a graph which contains sets of vertices $\calC_0, \dots, \calC_{q-1}$ (not necessarily disjoint or even distinct). Let $(J_0,z_0), \dots, (J_{q-1},z_{q-1})$ be partially $H$-labelled graphs with distinguished vertices, and, for each $i\in\{0, \dots, q-1\}$, let 
	$\Omega_i = \{ a \in V(H) \mid |\hom{(J_i, z_i)}{(H,a)}|\text{ is odd}\}$.
	Suppose that for all $i\in \{0, \dots, q-1\}$ we have the following.
	{
		\renewcommand{\theenumi}{L\thethm.\arabic{enumi}}
		\renewcommand{\labelenumi}{(\theenumi)}
		\setlength{\leftmargini}{5em}
		\begin{myenumerate}
			\item $\abs{\calC_{i-1 \bmod q}\cap\Omega_i}$ and $\abs{\calC_{i+1 \bmod q}\cap\Omega_i}$ are odd. \label{equ:CycleGadget+2}
			\item If $w\in \calC_{i-1 \bmod q}$ then $\NH{w} \cap \NH{\calC_{i+1 \bmod q}}= \calC_i$. \label{equ:CycleGadget+3a}
			\item If $w\in \calC_{i+1 \bmod q}$ then $\NH{\calC_{i-1 \bmod q}} \cap \NH{w} = \calC_i$. \label{equ:CycleGadget+3b}
			\item There is no walk of the form $D=(d_0, \dots, d_{q-1}, d_0)$ such that, for all $i\in \{0,\dots, q-1\}$, $d_i\in \NH{\calC_i}\setminus(\calC_{i-1 \bmod q}\cup \calC_{i+1 \bmod q})$.\label{equ:CycleGadget+4}
		\end{myenumerate} 
	}
	\noindent Let $J_\calC$ be the cycle gadget (Definition~\ref{def:cyclegadget}) and let $J^*=J(J_{\calC}, J_0, \dots, J_{q-1})$ be the generalised cycle gadget (Definition~\ref{def:generalisedcyclegadget})
	Then, for all $k\in \{0, \dots, q-1\}$,
	\[
	\{a \in V(H) \mid \abs{\hom{(J^*, v_k)}{(H,a)}} \text{ is odd}\} = \left(\calC_{k-1 \bmod q} \cup \calC_{k+1 \bmod q}\right) \cap \Omega_k.
	\] 
\end{lem}

\begin{proof}
	To simplify notation, all indices in this proof are considered to be modulo $q$. 
	For $a\in V(H)$, let $k\in \{0, \dots, q-1\}$ and $h\in \hom{(J^*, v_k)}{(H,a)}$. By construction of $J^*$ and the fact that $h$ has to preserve edges, for all $i\in \{0,\dots, q-1\}$, we obtain
	\begin{itemize}
		\item $h(v_i)\in \NH{\calC_i}$,
		\item $h(v_i)\notin \calC_i$ (since we do not allow self-loops in $H$),
		\item $h(v_i)$ is adjacent to $h(v_{i+1})$ in $H$,
		\item $h(v_i)\neq h(v_{i+1})$.
	\end{itemize}
	
	Consequently, it holds that $h(v_{i+1})\in \NH{h(v_i)} \cap \NH{\calC_{i+1}}$.
	Suppose, for some $i\in \{0, \dots, q-1\}$, that $h(v_i)\in \calC_{i-1}$. Then, by \eqref{equ:CycleGadget+3a}, we have $h(v_{i+1})\in \calC_i$. Therefore,
	\begin{equation}\label{equ:CycleGadget+A}
	\text{If }h(v_i)\in \calC_{i-1} \text{ then } h(v_{i+1})\in \calC_i.
	\end{equation}
	Analogously, using \eqref{equ:CycleGadget+3b},
	\begin{equation}\label{equ:CycleGadget+B}
	\text{If }h(v_i)\in \calC_{i+1} \text{ then } h(v_{i-1})\in \calC_i.
	\end{equation}
	
	Thus, if there exists some $\ell\in\{0, \dots, q-1\}$ such that $h(v_\ell)\in \calC_{\ell-1}$ then we can use~\eqref{equ:CycleGadget+A} iteratively to obtain $h(v_i)\in \calC_{i-1}$ for all $i\in\{0, \dots, q-1\}$. In particular, $h(v_k) \in \calC_{k-1}$. Analogously, if there exists some $\ell\in\{0, \dots, q-1\}$ such that $h(v_\ell)\in \calC_{\ell+1}$ then we can use~\eqref{equ:CycleGadget+B} iteratively to obtain $h(v_i)\in \calC_{i+1}$ for all $i\in\{0, \dots, q-1\}$. This means that $h(v_k) \in \calC_{k+1}$. 
	
	Suppose that $h(v_k) \notin \calC_{k-1} \cup \calC_{k+1}$. We have established that, using \eqref{equ:CycleGadget+A} and~\eqref{equ:CycleGadget+B} iteratively, we obtain, for all $i\in\{0, \dots, q-1\}$, $h(v_i)\notin \calC_{i-1} \cup \calC_{i+1}$ and consequently $h(v_i)\in \NH{\calC_i}\setminus(\calC_{i-1} \cup \calC_{i+1})$. However, $(h(v_0), \dots, h(v_{q-1}),h(v_0))$ is a walk in $H$, which gives a contradiction to~\eqref{equ:CycleGadget+4}.
	
	We have shown that $h(v_k) \in \calC_{k-1} \cup \calC_{k+1}$. Moreover, for each $a\in \calC_{k-1}$, $\abs{\hom{(J^*, v_k)}{(H,a)}} = \abs{\hom{(J_k, z_k)}{(H,a)}} \cdot \prod_{i \in \{0, \dots, q-1\}\setminus\{k\}} \abs{\calC_{i-1}\cap\Omega_i}$, which is odd if and only if $a\in \calC_{k-1}\cap\Omega_k$ by \eqref{equ:CycleGadget+2}. The statement for $a\in \calC_{k+1}$ is analogous.
\end{proof}

\begin{lem}\label{lem:CycleHardness+}
	For an integer $q\ge 3$, let $H$ be a graph which contains sets of vertices $\calC_0, \dots, \calC_{q-1}$ (not necessarily disjoint or even distinct). Let $(J_0,z_0), \dots, (J_{q-1},z_{q-1})$ be partially $H$-labelled graphs with distinguished vertices, and, for each $i\in\{0, \dots, q-1\}$, let 
	$\Omega_i = \{ a \in V(H) \mid |\hom{(J_i, z_i)}{(H,a)}|\text{ is odd}\}$.
	Suppose that for all $i\in \{0, \dots, q-1\}$ we have the following properties from the statement of Lemma~\ref{lem:CycleGadget+}.
	{
		\renewcommand{\theenumi}{L\ref{lem:CycleGadget+}.\arabic{enumi}}
		\renewcommand{\labelenumi}{(\theenumi)}
		\setlength{\leftmargini}{5em}
		\begin{myenumerate}
			\item $\abs{\calC_{i-1 \bmod q}\cap\Omega_i}$ and $\abs{\calC_{i+1 \bmod q}\cap\Omega_i}$ are odd. \label{equ:CycleHardness+1}
			\item If $w\in \calC_{i-1 \bmod q}$ then $\NH{w} \cap \NH{\calC_{i+1 \bmod q}}= \calC_i$.\label{equ:CycleHardness+2}
			\item If $w\in \calC_{i+1 \bmod q}$ then $\NH{\calC_{i-1 \bmod q}} \cap \NH{w} = \calC_i$. \label{equ:CycleHardness+3}
			\item There is no walk of the form $D=(d_0, \dots, d_{q-1}, d_0)$ such that, for all $i\in \{0,\dots, q-1\}$, $d_i\in \NH{\calC_i}\setminus(\calC_{i-1 \bmod q}\cup \calC_{i+1 \bmod q})$.\label{equ:CycleHardness+4}
		\end{myenumerate} 
	}
	\noindent Furthermore, there exists $k\in\{0, \dots, q-1\}$ such that
	{
		\renewcommand{\theenumi}{L\thethm.\arabic{enumi}}
		\renewcommand{\labelenumi}{(\theenumi)}
		\setlength{\leftmargini}{5em}
		\begin{myenumerate}
			\item  there are no edges between $\calC_k$ and $\calC_{k+3 \bmod q}$,\label{equ:CycleHardness+5}
			\item  $\abs{\left(\calC_k \cup \calC_{k+2\bmod q}\right) \cap \Omega_{k+1}}$ and $\abs{\left(\calC_{k+1\bmod q} \cup \calC_{k+3\bmod q}\right) \cap \Omega_{k+2}}$ are even. \label{equ:CycleHardness+6}
		\end{myenumerate}
	}
	\noindent Then $H$ has a hardness gadget.
\end{lem}
\begin{proof}
	To simplify notation all indices in this proof are considered to be modulo $q$. 
	We construct a hardness gadget $(I,S,(J'_1, y), (J'_2, z), (J'_3, y,z))$ for $H$, as defined in Definition~\ref{defn:hardness-gadget}.
	
	Let $\calC=(\calC_0, \dots, \calC_{q-1})$.
	Let $J'_1$ and $J'_2$ each be an instance of the generalised cycle gadget $J(J_{\calC}, J_0, \dots, J_{q-1})$, let $y=v_{k+1}$, and let $z=v_{k+2}$. Then we have $\Omega_y=\left(\calC_k \cup \calC_{k+2}\right) \cap \Omega_{k+1}$ and $\Omega_z=\left(\calC_{k+1} \cup \calC_{k+3}\right) \cap \Omega_{k+2}$ by Lemma~\ref{lem:CycleGadget+}. It follows that $\abs{\Omega_y}$ and $\abs{\Omega_z}$ are even by~\eqref{equ:CycleHardness+6}. Let $I=\calC_{k+2}\cap\Omega_{k+1}$ and $S=\calC_{k+1}\cap\Omega_{k+2}$. We note that $I$ and $S$ have odd size by~\eqref{equ:CycleHardness+1} and that $I\subset \Omega_y$ and $S \subset \Omega_z$.
	
	Let $J_3$ be an edge from $y$ to $z$. For each $o\in \Omega_y\setminus I\subseteq\calC_k$, $s\in S\subseteq\calC_{k+1}$, $i\in I\subseteq\calC_{k+2}$ and $x\in \Omega_z\setminus S\subseteq\calC_{k+3}$,
	\begin{itemize}
		\item $\abs{\Sigma_{ox}}=0$ since no edge exists between $\calC_k$ and $\calC_{k+3}$ according to~\eqref{equ:CycleHardness+5}.
		\item $\abs{\Sigma_{is}}=\abs{\Sigma_{ix}}=\abs{\Sigma_{os}}=1$ since, by~\eqref{equ:CycleHardness+2}, for all $\ell\in \{0, \dots, q-1\}$ we have $\calC_{\ell} \subseteq \NH{\calC_{\ell+1}}$.
	\end{itemize}
\end{proof}

We point out a corollary which is more easily accessible and does not use the full generality of the gadget $J(J_{\calC}, J_0, \dots, J_{q-1})$ but rather only uses the cycle gadget $J_{\calC}$.	

\begin{cor}\label{cor:CycleOfSquares}
	For an integer $q=3$ or $q\ge 5$, let $H$ be a graph which contains a cycle $C=c_0, \dots, c_{q-1}, c_0$ such that 
	\begin{myitemize}
		\item for all $i\in \{0, \dots, q-1\}$, we have $\abs{N_{C,H}(c_i)}=1$, and
		\item there is no walk of the form $D=d_0, \dots, d_{q-1}, d_0$ with $d_i\in \NH{c_i}\setminus(c_{i-1}\cup c_{i+1})$ ($\forall i\in \{0, \dots, q-1\}$).
	\end{myitemize}
	Then $H$ has a hardness gadget.
\end{cor}
\begin{proof} All indices in this proof are considered to be modulo $q$. 
	For $i\in \{0, \dots, q-1\}$ we choose $\calC_i=N_{C,H}(c_i)$, which by the fact that $\abs{N_{C,H}(c_i)}=1$ implies $\calC_i=\{c_i\}$. We choose $k=0$. For each $i\in \{0, \dots, q-1\}$, let $(J_i, z_i)$ be the partially $H$-labelled graph that only contains the single vertex $z_i$ and has an empty pinning function. It follows that $\Omega_i=V(H)$ and that $J(J_{\calC}, J_0, \dots, J_{q-1})$ is essentially $J_{\calC}$. We check that the requirements of Lemma~\ref{lem:CycleHardness+} are met. \eqref{equ:CycleHardness+1} holds since $\calC_{i-1}\cap\Omega_i = \calC_{i-1} = \{c_{i-1}\}$ and $\calC_{i+1}\cap\Omega_i =\calC_{i+1} = \{c_{i+1}\}$. \eqref{equ:CycleHardness+2} and~\eqref{equ:CycleHardness+3} hold since $\abs{N_{C,H}(c_i)}=1$ and therefore $c_i$ is the only common neighbour of $c_{i-1}$ and $c_{i+1}$. 
	There is no walk of the form $D=d_0, \dots, d_{q-1}, d_0$ with $d_i\in \NH{c_i}\setminus(c_{i-1}\cup c_{i+1})$, as required by~\eqref{equ:CycleHardness+4}.
	Since $q\ge 3$ and $C$ is a cycle, the vertices $c_0, c_1, c_2$ are distinct. If $q=3$, as $C$ is a cycle, we have $c_0=c_3$, and~\eqref{equ:CycleHardness+5} holds since we do not allow self-loops in $H$. If otherwise $q\ge 5$ then~\eqref{equ:CycleHardness+5} holds since $\NH{c_1}\cap \NH{c_3}= N_{C,H}(c_2)=\{c_2\}$ and therefore $c_0$ (which is a neighbour of $c_1$) cannot be a neighbour of $c_3$. 
	Since $q\ge 3$~\eqref{equ:CycleHardness+6} holds as $\left(\calC_0 \cup \calC_{2}\right) \cap \Omega_{1}=\{c_0, c_2\}$ and $\left(\calC_{1} \cup \calC_{3}\right) \cap \Omega_{2}=\{c_1,c_3\}$ are sets of $2$ distinct vertices.
\end{proof}

\begin{rem}\label{rem:squarefreeSimplification} \label{A1}
Suppose that a square-free graph $H$ contains a cycle $C$. Clearly, the requirements of Corollary~\ref{cor:CycleOfSquares} are met and, by Theorem~\ref{thm:hardness-gadget}, we obtain $\parP$-hardness for $\parRet{H}$. If, in addition, $H$ is involution-free $\parP$-hardness carries over to $\parHom{H}$ by Theorem~\ref{thm:ETH_RetToHoms} (from~\cite[Theorem 3.1]{squarefree}). This argument, together with the classification of $\parHom{H}$ for trees by Faben and Jerrum~\cite{FJ} (or alternatively the shorter~\cite[Lemmas 5.1 and 5.3]{squarefree}) implies the dichotomy for square-free graphs presented in~\cite{squarefree}.
\end{rem}

\section{Chordal Bipartite Components}\label{sec:chordalBipartiteComps}

Our main strategy for proving $\parP$-hardness of $\parHom{H}$ for $K_4$-minor-free graphs will rely on finding induced cycles 
whose lengths are not equal to $4$. However, this requires us to treat the case of ($K_4$-minor-free) graphs that include \emph{only} squares as induced cycles separately; recall that such graphs are called chordal bipartite graphs.

In the current section we will construct a hardness gadget for every involution-free, $K_4$-minor-free, biconnected chordal bipartite graph $H$, 
unless $H$ has a very restricted form (this is Lemma~\ref{lem:main_chordal_bipartite}). In this restricted case we call~$H$ an \emph{impasse} (which will be formally defined in Definition~\ref{def:impasse}). The main tool that we use to construct  hardness gadgets relies on two squares that share one edge. More formally, we will consider the following graph:

\begin{defn}[The graph $F$, $\Gamma_{H\setminus F}(i,j)$]\label{def:graphF}
	The graph $F$ is defined to be the graph depicted in Figure~\ref{fig:F}.
	
	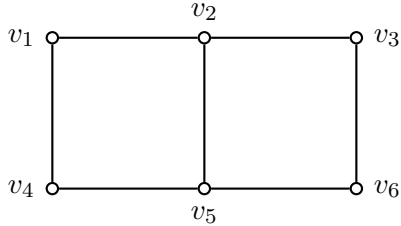
\begin{figure}[H]
		\centering
		\begin{tikzpicture}[scale=2, node distance = 1.4cm, thick] 
		\tikzstyle{vertex}=[draw=black, circle, inner sep=1.5pt] 
		\begin{scope}[xshift=2cm] 
		\node[vertex] (s) at (0  ,1) [label=90: $v_2$] {};
		\node[vertex] (i) at (0  ,0) [label=270: $v_5$] {};
		\node[vertex] (o) at (-1  ,1) [label=180: $v_1$] {};
		\node[vertex] (u) at (-1  ,0) [label=180: $v_4$] {};
		\node[vertex] (v) at (1  ,1) [label=0: $v_3$] {};
		\draw (s)--(i); \draw (s)--(o); \draw (s)--(v); \draw (o)--(u);
		\draw (u)--(i); 
		\node[vertex] (x) at (1  ,0) [label=0: $v_6$] {};
		\draw (x)--(i); \draw (v)--(x);
		\end{scope}
		\end{tikzpicture}
		\caption{The graph $F$.}
		\label{fig:F}
	\end{figure}

	Given a graph $H$ that contains $F$ as a subgraph, and $i\neq j\in [6]$, we define
	\[\Gamma_{H\setminus F}(i,j) = (\Gamma_H(v_i) \cap \Gamma_H(v_j))\setminus V(F) \,.\]
\end{defn}

\begin{defn}[Type $\TV$]\label{def:typeV}
	Let $H$ be a $K_4$-minor-free graph that contains $F$ as a subgraph. We say that $F$ has type $\TV$ in $H$ if one of the following is true
	\begin{itemize}
		\item $\Gamma_{H\setminus F}(1,5)$ and $\Gamma_{H\setminus F}(3,5)$ are non-empty and $\Gamma_{H\setminus F}(2,4)$ and $\Gamma_{H\setminus F}(2,6)$ are empty.
		\item $\Gamma_{H\setminus F}(2,4)$ and $\Gamma_{H\setminus F}(2,6)$ are non-empty and $\Gamma_{H\setminus F}(1,5)$ and $\Gamma_{H\setminus F}(3,5)$ are empty.
	\end{itemize}
	An illustration of the former case is given in Figure~\ref{fig:typeV}.
\end{defn}

The following observation will be useful in the remainder of this section:
\begin{lem}\label{lem:K_4-free}
	Let $H$ be a $K_4$-minor-free graph containing $F$ as a subgraph. At least one of $\Gamma_{H\setminus F}(1,5)$ and $\Gamma_{H\setminus F}(2,4)$ is empty, and at least one of $\Gamma_{H\setminus F}(2,6)$ and $\Gamma_{H\setminus F}(3,5)$ is empty.
\end{lem}
\begin{proof}
	If $\Gamma_{H\setminus F}(1,5)$ and $\Gamma_{H\setminus F}(2,4)$ are both non-empty, then the vertices $v_1,v_2,v_4$ and $v_5$ yield a $K_4$-minor. 
	If $\Gamma_{H\setminus F}(2,6)$ and $\Gamma_{H\setminus F}(3,5)$ are both non-empty, then the vertices $v_2,v_3,v_5$ and $v_6$ yield a $K_4$-minor. 
\end{proof}

\begin{lem}\label{lem:typeV}
	Let $H$ be a $K_4$-minor-free graph containing $F$ as a subgraph. If $F$ does not have type $\TV$ in $H$ then either $\Gamma_{H\setminus F}(1,5)=\Gamma_{H\setminus F}(2,6)=\emptyset$ or $\Gamma_{H\setminus F}(2,4)=\Gamma_{H\setminus F}(3,5)=\emptyset$.	
\end{lem}
\begin{proof}
	Note that either $\Gamma_{H\setminus F}(2,6)$ or $\Gamma_{H\setminus F}(3,5)$ are empty by Lemma~\ref{lem:K_4-free}. Assume w.l.o.g.\ that the former is empty; the other case is symmetric.	
	We distinguish two cases:
	\begin{enumerate}[(I)]
		\item $\Gamma_{H\setminus F}(3,5)\neq \emptyset$. Now assume for contradiction that $\Gamma_{H\setminus F}(1,5)\neq\emptyset$. Then, again by Lemma~\ref{lem:K_4-free}, we obtain $\Gamma_{H\setminus F}(2,4)=\emptyset$, which implies that $F$ has type $\TV$ in $H$, yielding the desired contradiction. In combination with the previous assumption, we thus have $\Gamma_{H\setminus F}(1,5)=\Gamma_{H\setminus F}(2,6)=\emptyset$.
		\item $\Gamma_{H\setminus F}(3,5) = \emptyset$. By Lemma~\ref{lem:K_4-free} we have that either $\Gamma_{H\setminus F}(1,5)$ or $\Gamma_{H\setminus F}(2,4)$ is empty. This concludes the proof as the current case provides additionally $\Gamma_{H\setminus F}(3,5)=\emptyset$ and $\Gamma_{H\setminus F}(2,6)=\emptyset$.
	\end{enumerate}
\end{proof}
\begin{figure}[t!]
	\centering
	\begin{tikzpicture}[scale=2, node distance = 1.4cm, thick] 
	\tikzstyle{vertex}=[draw=black, circle, inner sep=1.5pt] 
	\begin{scope}[xshift=2cm] 
	\node[vertex] (s) at (0  ,1) [label=90: $v_2$] {};
	\node[vertex] (i) at (0  ,0) [label=270: $v_5$] {};
	\node[vertex] (o) at (-1  ,1) [label=180: $v_1$] {};
	\node[vertex] (u) at (-1  ,0) [label=180: $v_4$] {};
	\node[vertex] (v) at (1  ,1) [label=0: $v_3$] {};
	\draw (s)--(i); \draw (s)--(o); \draw (s)--(v); \draw (o)--(u);
	\draw (u)--(i); 
	\node[vertex] (x) at (1  ,0) [label=0: $v_6$] {};
	\draw (x)--(i); \draw (v)--(x);
	\node[vertex] (v7) at (-0.5  ,0.5)  {};
	\node[vertex] (v8) at (0.5  ,0.5)  {};
	\draw (o) -- (v7); \draw (v7) -- (i);
	\draw (i) -- (v8); \draw (v8) -- (v);
	\end{scope}
	\end{tikzpicture}
	\caption{\label{fig:typeV} A $K_4$-minor-free graph containing $F$ of type $\TV$.} 
\end{figure}
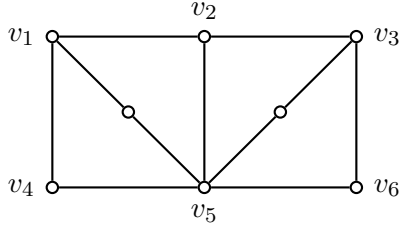

\begin{lem}\label{lem:typeV_hard}
	Let $H$ be a $K_4$-minor-free graph containing $F$ as a subgraph. Then $H$ has a hardness gadget, unless $F$ has type $\TV$ in $H$.
\end{lem}
\begin{proof}
	Using Lemma~\ref{lem:typeV} and the fact that $H$ is $K_4$-minor free, we can w.l.o.g.\ assume that
	\begin{enumerate}[(a)]
		\item The edges $\{v_1,v_6\}$ and $\{v_3,v_4\}$ are \emph{not} present in $H$ as, otherwise, we obtain a $K_4$-minor.
		\item $\Gamma_H(v_1)\cap\Gamma_H(v_5) = \{v_2,v_4\}$.
		\item $\Gamma_H(v_2)\cap\Gamma_H(v_6) = \{v_3,v_5\}$.
	\end{enumerate}
	This allows us to construct a hardness gadget:
	\begin{itemize}
		\item $S=\{v_5\}$ and $I=\{v_2\}$.
		\item $J_1$ is the graph where $y$ is adjacent to a $v_1$-pin and a $v_5$-pin 
		Note that $\Omega_y=\{v_2,v_4\}$ by (b).
		\item $J_2$ is the graph where $z$ is adjacent to a $v_2$-pin and a $v_6$-pin. Note that $\Omega_z=\{v_3,v_5\}$ by (c).
		\item $J_3$ is just the edge $\{y,z\}$.
	\end{itemize}
	We have $|\Sigma_{v_4,v_5}|=|\Sigma_{v_5,v_2}|=|\Sigma_{v_2,v_3}|= 1$. Furthermore, $|\Sigma_{v_4,v_3}|=0$ by (a).
\end{proof} 

\begin{defn}[The graph {$S_{k,\ell}$}]\label{def:Skl}
	For positive integers $k$ and $\ell$, $S_{k,\ell}$ is the graph depicted in Figure~\ref{fig:Skl}.
	
	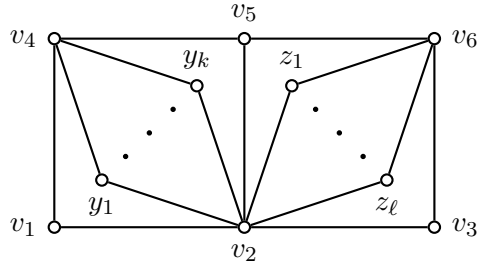
\begin{figure}[H]
		\centering
		\begin{tikzpicture}[scale=2.5, node distance = 1.4cm,thick,yscale=-1]
		\tikzstyle{dot}   =[fill=black, draw=black, circle, inner sep=0.15mm]
		\tikzstyle{vertex}=[  draw=black, circle, inner sep=1.5pt]
		\tikzstyle{dist}  =[fill=white, draw=black, circle, inner sep=2pt]
		\tikzstyle{pinned}=[draw=black, minimum size=10mm, circle, inner sep=0pt]	
		\begin{scope}[xshift=2cm] 
		\node[vertex] (s) at (0  ,1) [label=270: $v_2$] {};
		\node[vertex] (i) at (0  ,0) [label=90: $v_5$] {};
		\node[vertex] (o) at (-1  ,1) [label=180: $v_1$] {};
		\node[vertex] (u) at (-1  ,0) [label=180: $v_4$] {};
		\node[vertex] (v) at (1  ,1) [label=0: $v_3$] {};
		\node[vertex] (x) at (1  ,0) [label=0: $v_6$] {};	
		\draw (s)--(i); \draw (s)--(o); \draw (s)--(v); \draw (o)--(u);
		\draw (u)--(i); \draw (x)--(i); \draw (v)--(x);	
		\node[vertex] (y1) at (-0.75  ,0.75) [label=270: $y_1$] {};
		\node[vertex] (yk) at (-0.25  ,0.25) [label=90: $y_k$]  {};
		\node[vertex] (zl) at (0.25  ,0.25)  [label=90: $z_1$] {};
		\node[vertex] (z1) at (0.75  ,0.75) [label=270: $z_\ell$] {};	
		\draw (s)--(y1); \draw (s)--(yk);  \draw (u)--(y1); \draw (u)--(yk);
		\draw (x)--(z1); \draw (x)--(zl); \draw (s)--(z1); \draw (s)--(zl);    	
		\node[dot] (d) at (-0.5,0.5) {};  
		\node[dot] (d) at (-0.625,0.625) {}; 
		\node[dot] (d) at (-0.375,0.375) {}; 
		\node[dot] (d) at (0.5,0.5) {};  
		\node[dot] (d) at (0.625,0.625) {}; 
		\node[dot] (d) at (0.375,0.375) {}; 
		\end{scope}
		\end{tikzpicture} 
		\caption{The graph $S_{k, \ell}$.}
		\label{fig:Skl}
	\end{figure}
	
\end{defn}

\begin{lem}\label{lem:even_diamond}
	Let $H$ be a $K_4$-minor-free graph containing $F$ as a subgraph. If~$F$ has type $\TV$ in $H$ and $|\Gamma_{H\setminus F}(1,5)|$ and $|\Gamma_{H\setminus F}(2,4)|$ are even, then $H$ has a hardness gadget.
\end{lem}
\begin{proof}
	As $F$ has type $\TV$ in $H$ we can assume w.l.o.g.\ that $\Gamma_{H\setminus F}(2,4)\neq \emptyset$ and $\Gamma_{H\setminus F}(2,6)\neq \emptyset$, and that $\Gamma_{H\setminus F}(1,5) = \Gamma_{H\setminus F}(3,5)=\emptyset$; the other case is symmetric. In other words, there exist positive integers $k$ and $\ell$ such that $H$ contains the subgraph $S_{k, \ell}$ (Definition~\ref{def:Skl}) with $\Gamma_{H\setminus F}(2,4)=\{y_1,\dots,y_k\}$ and $\Gamma_{H\setminus F}(2,6)=\{z_1,\dots,z_\ell\}$. 
	
	\noindent By the premise of the lemma, $k$ must be even. 
	We will emphasise some crucial properties of~$H$:
	\begin{enumerate}[(a)]
		\item $\Gamma_H(v_3)\cap\Gamma_H(v_5) = \{v_2,v_6\}$, since $\Gamma_{H\setminus F}(3,5)=\emptyset$. 
		\item $v_6$ is not adjacent to any vertex in $\{y_1,\dots,y_k,v_1\}$: Assuming otherwise, let $w\in \{y_1,\dots,y_k,v_1\}$ be adjacent to $v_6$. We obtain the following $K_4$-minor of $H$:
		
		
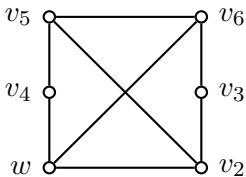
\medskip
		\begin{tikzpicture}[scale=2, node distance = 1.4cm, thick] 
		\tikzstyle{vertex}=[draw=black, circle, inner sep=1.5pt] 
		\begin{scope}[xshift=2cm] 
		\node[vertex] (s) at (1  ,1) [label=0: $v_6$] {};
		\node[vertex] (i) at (0  ,0) [label=180: $w$] {};
		\node[vertex] (o) at (0  ,1) [label=180: $v_5$] {};
		\node[vertex] (u) at (0  ,0.5) [label=180: $v_4$] {};
		\node[vertex] (v) at (1  ,0.5) [label=0: $v_3$] {};
		\draw (s)--(i); \draw (s)--(o); \draw (s)--(v); \draw (o)--(u);
		\draw (u)--(i); 
		\node[vertex] (x) at (1  ,0) [label=0: $v_2$] {};
		\draw (x)--(i); \draw (v)--(x); \draw (o) -- (x);
		
		\end{scope}
		\end{tikzpicture}
		\medskip
	\end{enumerate}
	We proceed by constructing a hardness gadget:
	\begin{itemize}
		\item $S=\{v_2\}$ and $I=\{v_5\}$.
		\item $J_1$ is the graph where $y$ is adjacent to a $v_2$-pin and a $v_4$-pin. Note that \[\Omega_y=\{v_1,v_5\} \cup \Gamma_{H\setminus F}(2,4) =\{v_1,v_5,y_1,\dots,y_k\}\,.\] In particular, $|\Omega_y|$ is even as $k$ is.
		\item $J_2$ is the graph where $z$ is adjacent to a $v_3$-pin and a $v_5$-pin. Note that $\Omega_z=\{v_2,v_6\}$ by (a).
		\item $J_3$ is just the edge $\{y,z\}$.
	\end{itemize}
	We have $|\Sigma_{v_2,v_5}|=|\Sigma_{v_5,v_6}|=1$ and, for every $o\in \Omega_y\setminus\{v_5\}$, $|\Sigma_{o,v_2}|=1$. Furthermore, by (b), $|\Sigma_{o,v_6}|=0$. 
\end{proof}

\subsection{Strong Hardness Gadgets}

\begin{defn}[strong hardness gadget]\label{def:stronghardnessgadget}
	A graph $J$ is called a \emph{strong hardness gadget} if every $K_4$-minor-free graph that contains $J$ as a 
	subgraph has a hardness gadget.
\end{defn}

\begin{lem}\label{lem:good_triangle_gadget}
	The following graph $J$ is a strong hardness gadget:
	
	\medskip
	
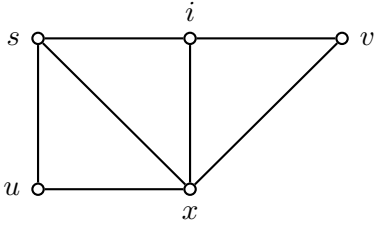
\begin{tikzpicture}[scale=2, node distance = 1.4cm, thick] 
	\tikzstyle{vertex}=[draw=black, circle, inner sep=1.5pt] 
	\begin{scope}[xshift=2cm] 
	\node[vertex] (s) at (0  ,1) [label=90: $i$] {};
	\node[vertex] (i) at (0  ,0) [label=270: $x$] {};
	\node[vertex] (o) at (-1  ,1) [label=180: $s$] {};
	\node[vertex] (u) at (-1  ,0) [label=180: $u$] {};
	\node[vertex] (v) at (1  ,1) [label=0: $v$] {};
	\draw (s)--(i); \draw (s)--(o); \draw (s)--(v); \draw (o)--(u);
	\draw (u)--(i); \draw (v) -- (i); \draw (o) -- (i);
	\end{scope}
	\end{tikzpicture}
	\medskip
\end{lem} 
\begin{proof}
	Let $H$ be a $K_4$-minor-free supergraph of $J$. We construct a hardness gadget of $H$:
	\begin{itemize}
		\item $S=\{s\}$ and $I=\{i\}$.
		\item $J_1$ is the graph where $y$ is adjacent to a $u$-pin and an $i$-pin. Note that $\Omega_y=\{x,s\}$ as $H$ is $K_4$-minor free.
		\item $J_2$ is the graph where $z$ is adjacent to a $v$-pin and an $s$-pin. Note that $\Omega_z=\{x,i\}$ as $H$ is $K_4$-minor free.
		\item $J_3$ is just the edge $\{y,z\}$.
	\end{itemize}
	We have $|\Sigma_{s,i}|=|\Sigma_{s,x}|=|\Sigma_{x,i}|=1$ and $|\Sigma_{x,x}|=0$ --- recall that we do not allow self-loops. 
\end{proof}

For the proof of the following lemma recall the definition of walk-neighbour-sets from Definition~\ref{def:walkneighbourset}.

\begin{lem}\label{lem:adj_odd_neighbours}
	Let $H$ be a $K_4$-minor-free graph containing two adjacent vertices $a$ and $b$ such that $|\Gamma_H(a)\cap\Gamma_H(b)|$ is odd and at least $3$. Then $H$ has a hardness gadget.
\end{lem}
\begin{proof} 
	Let $c$ be a common neighbour of $a$ and $b$ and consider the triangle $C=(a,b,c,a)$: If $a$ and $c$ have a common neighbour apart from $b$, or if $b$ and $c$ have a common neighbour apart from $a$ then Lemma~\ref{lem:good_triangle_gadget} applies, as $a$ and $b$ have a common neighbour apart from $c$ by assumption. Otherwise, we have that
	$\abs{N_{C,H}(a)}=1$, $\abs{N_{C,H}(b)}=1$, and $\abs{N_{C,H}(c)}=j\ge 3$, where $j$ is odd.  
	For any $w\in N_{C,H}(c)$   we can assume that
	\begin{equation}\label{equ:adj_odd_neighbours}
	\NH{w}\cap \NH{a}=\{b\} \text{ and } \NH{w}\cap \NH{b}=\{a\},
	\end{equation}
	as otherwise we obtain a hardness gadget from Lemma~\ref{lem:good_triangle_gadget} (choose $w$ instead of $c$). 
	Next we can apply Lemma~\ref{lem:CycleHardness+} to obtain a hardness gadget as follows.
	
	Let $q=3$ and $\calC_0=N_{C,H}(a)=\{a\}$, $\calC_1=N_{C,H}(b)=\{b\}$, $\calC_2=N_{C,H}(c)$.
	For each $i\in \{0, 1,2\}$, let $(J_i, z_i)$ be the partially $H$-labelled graph that only contains the single vertex $z_i$ and has an empty pinning function. It follows that $\Omega_i=V(H)$. We choose $k=0$ and check that the requirements of Lemma~\ref{lem:CycleHardness+} are met.
	\begin{myitemize}
		\item \eqref{equ:CycleHardness+1} holds since, for each $i\in \{0, 1,2\}$, $\Omega_i=V(H)$ and $\calC_i$ has odd cardinality (either $1$ or $j$).
		\item \eqref{equ:CycleHardness+2} and~\eqref{equ:CycleHardness+3} hold by \eqref{equ:adj_odd_neighbours} and the fact that $\NH{a}\cap\NH{b}=N_{C,H}(c)=\calC_2$.
		\item  Suppose for contradiction that there exists a walk $D=(d_a,d_b,d_c,d_a)$ with $d_a\in \NH{a}\setminus\{b,c\}$, $d_b\in \NH{b}\setminus\{a,c\}$ and $d_c\in \NH{c}\setminus\{a,b\}$. Consequently, as we do not allow self-loops in $H$, $d_a\neq a$, $d_b\neq b$ and $d_c\neq c$. Then the vertices $d_a, a, b, c$ induce a $K_4$-minor (where the path from $d_a$ to $b$ goes via $d_b$, and the path from $d_a$ to $c$ goes via $d_c$). Hence \eqref{equ:CycleHardness+4} holds.
		\item Since $\calC_0=\calC_{3\mod q}=\{a\}$, \eqref{equ:CycleHardness+5} holds by the fact that we do not allow self-loops in $H$.
		\item \eqref{equ:CycleHardness+6} holds since $(\calC_0\cup \calC_2)\cap\Omega_1= \calC_0\cup \calC_2$, which has cardinality $j+1$ (as we do not allow self-loops in $H$ and therefore $\calC_0=\{a\}$ and $\calC_2=N_{C,H}(c)=\NH{a}\cap\NH{b}$ are disjoint), and $j+1$ is even. Analogously, $(\calC_1\cup \calC_2)\cap\Omega_1= \calC_1\cup \calC_2$ has even cardinality.
	\end{myitemize}
\end{proof}

\begin{lem}\label{lem:diamond_square_gadget}
	The following graph $J$ is a strong hardness gadget:
	
	\medskip
	
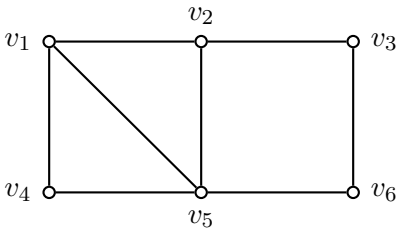
\begin{tikzpicture}[scale=2, node distance = 1.4cm, thick] 
	\tikzstyle{vertex}=[draw=black, circle, inner sep=1.5pt] 
	\begin{scope}[xshift=2cm] 
	\node[vertex] (s) at (0  ,1) [label=90: $v_2$] {};
	\node[vertex] (i) at (0  ,0) [label=270: $v_5$] {};
	\node[vertex] (o) at (-1  ,1) [label=180: $v_1$] {};
	\node[vertex] (u) at (-1  ,0) [label=180: $v_4$] {};
	\node[vertex] (v) at (1  ,1) [label=0: $v_3$] {};
	\draw (s)--(i); \draw (s)--(o); \draw (s)--(v); \draw (o)--(u);
	\draw (u)--(i); \draw (o) -- (i);
	\node[vertex] (x) at (1  ,0) [label=0: $v_6$] {};
	\draw (x)--(i); \draw (v)--(x);
	\end{scope}
	\end{tikzpicture}
	\medskip
\end{lem} 
\begin{proof}
	Let $H$ be a $K_4$-minor-free supergraph of $J$. In particular, the graph $F$ is a subgraph of $J$ and thus of $H$. Note that, due to the edge $\{v_1,v_5\}$, the vertices $v_2$ and $v_4$ have no common neighbours apart from $v_1$ and $v_5$ in $H$, as we would obtain a $K_4$-minor otherwise. In other words, $\Gamma_{H\setminus F}(2,4)=0$. By Lemma~\ref{lem:typeV_hard} we are done, unless $F$ has type $\TV$ in $H$. In particular, as $\Gamma_{H\setminus F}(2,4)=\emptyset$, only the following case remains:

	\medskip
	
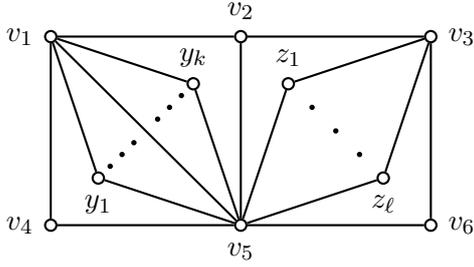
\begin{tikzpicture}[scale=2.5, node distance = 1.4cm,thick]
	\tikzstyle{dot}   =[fill=black, draw=black, circle, inner sep=0.15mm]
	\tikzstyle{vertex}=[  draw=black, circle, inner sep=1.5pt]
	\tikzstyle{dist}  =[fill=white, draw=black, circle, inner sep=2pt]
	\tikzstyle{pinned}=[draw=black, minimum size=10mm, circle, inner sep=0pt]	
	\begin{scope}[xshift=2cm] 
	\node[vertex] (s) at (0  ,1) [label=90: $v_2$] {};
	\node[vertex] (i) at (0  ,0) [label=270: $v_5$] {};
	\node[vertex] (o) at (-1  ,1) [label=180: $v_1$] {};
	\node[vertex] (u) at (-1  ,0) [label=180: $v_4$] {};
	\node[vertex] (v) at (1  ,1) [label=0: $v_3$] {};
	\node[vertex] (x) at (1  ,0) [label=0: $v_6$] {};	
	\draw (s)--(i); \draw (s)--(o); \draw (s)--(v); \draw (o)--(u);
	\draw (u)--(i); \draw (x)--(i); \draw (v)--(x);	
	\node[vertex] (y1) at (-0.75  ,0.25) [label=270: $y_1$] {};
	\node[vertex] (yk) at (-0.25  ,0.75) [label=90: $y_k$]  {};
	\node[vertex] (zl) at (0.25  ,0.75)  [label=90: $z_1$] {};
	\node[vertex] (z1) at (0.75  ,0.25) [label=270: $z_\ell$] {};	
	\draw (o)--(y1); \draw (o)--(yk);  \draw (i)--(y1); \draw (i)--(yk);
	\draw (v)--(z1); \draw (v)--(zl); \draw (i)--(z1); \draw (i)--(zl);    	
	\draw (o) -- (i);
	\node[dot] (d) at (-0.6875,0.3125) {}; 
	\node[dot] (d) at (-0.625,0.375) {}; 
	\node[dot] (d) at (-0.5625,0.4375) {}; 
	
	\node[dot] (d) at (-0.4375,0.5625) {}; 
	\node[dot] (d) at (-0.375,0.625) {}; 
	\node[dot] (d) at (-0.3125,0.6875) {};
	
	\node[dot] (d) at (0.5,0.5) {};  
	\node[dot] (d) at (0.375,0.625) {}; 
	\node[dot] (d) at (0.625,0.375) {}; 
	\end{scope}
	\end{tikzpicture} 
	\medskip
	
	\noindent In particular, $\Gamma_{H\setminus F}(1,5)=\{y_1,\dots,y_k\}$ and $\Gamma_{H\setminus F}(3,5)=\{z_1,\dots,z_\ell\}$ and $k,\ell>0$. Now, if $k$ is even, then Lemma~\ref{lem:even_diamond} yields a hardness gadget of $H$. Finally, if $k$ is odd, then Lemma~\ref{lem:adj_odd_neighbours} yields a hardness gadget of $H$ --- note that Lemma~\ref{lem:adj_odd_neighbours} is applicable as $v_1$ and $v_5$ have precisely $k+2$ common neighbours, which is an odd number greater or equal than $3$ since $k$ is odd and positive. 
\end{proof}

\begin{lem}\label{lem:3x1}
	The following graph $J$ is a strong hardness gadget:
	
	\medskip
	
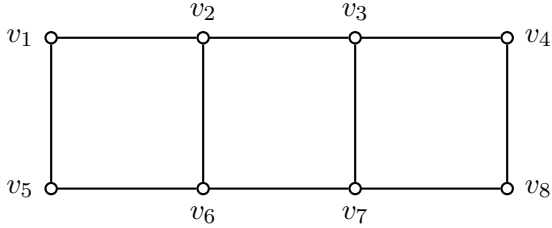
\begin{tikzpicture}[scale=2, node distance = 1.4cm, thick] 
	\tikzstyle{vertex}=[draw=black, circle, inner sep=1.5pt] 
	\begin{scope}[xshift=2cm] 
	\node[vertex] (s) at (0  ,1) [label=90: $v_2$] {};
	\node[vertex] (i) at (0  ,0) [label=270: $v_6$] {};
	\node[vertex] (o) at (-1  ,1) [label=180: $v_1$] {};
	\node[vertex] (u) at (-1  ,0) [label=180: $v_5$] {};
	\node[vertex] (v) at (1  ,1) [label=90: $v_3$] {};
	\node[vertex] (v4) at (2  ,1) [label=0: $v_4$] {};
	\node[vertex] (v8) at (2  ,0) [label=0: $v_8$] {};
	\draw (s)--(i); \draw (s)--(o); \draw (s)--(v); \draw (o)--(u);
	\draw (u)--(i); 
	\node[vertex] (x) at (1  ,0) [label=270: $v_7$] {};
	\draw (x)--(i); \draw (v)--(x); \draw (v)--(v4); \draw (v4)--(v8);
	\draw (v8)--(x);
	\end{scope}
	\end{tikzpicture}
	\medskip
\end{lem} 
\begin{proof}
	Let $H$ be a $K_4$-minor-free supergraph of $J$.

	\medskip
	\noindent \textbf{Claim A} \textit{If $J$ is not an induced subgraph of $H$ then $H$ has a $K_4$-minor or a hardness gadget.}
	\medskip
	\begin{claimproof}
		If $J$ is not an induced subgraph of $H$ then there is an edge $e=\{v_i,v_j\}\in E(H)\setminus E(J)$ for some $i\neq j\in[8]$.
		If $e$ is a diagonal of one of the three squares, such as $\{v_2,v_7\}$, then $H$ has a hardness gadget by Lemma~\ref{lem:diamond_square_gadget}. 
		
		If $e$ is not a diagonal of a square, then we obtain a $K_4$-minor; each case is similar to one of the following two:
		
		\medskip
		
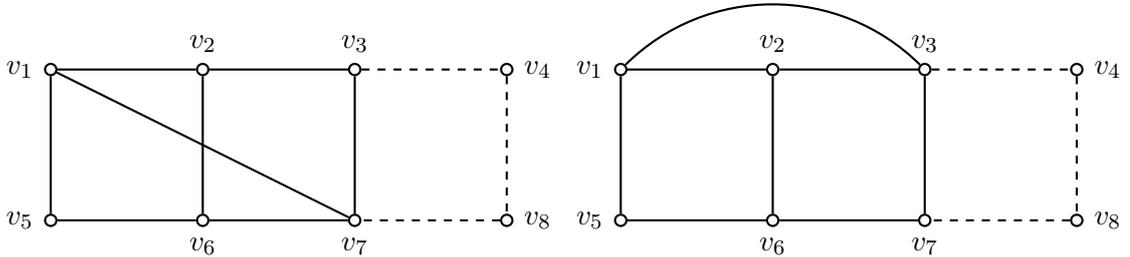
\begin{tikzpicture}[scale=2, node distance = 1.4cm, thick] 
		\tikzstyle{vertex}=[draw=black, circle, inner sep=1.5pt] 
		
		\begin{scope}[xshift=1.25cm]
		\node[vertex] (s) at (0  ,1) [label=90: $v_2$] {};
		\node[vertex] (i) at (0  ,0) [label=270: $v_6$] {};
		\node[vertex] (o) at (-1  ,1) [label=180: $v_1$] {};
		\node[vertex] (u) at (-1  ,0) [label=180: $v_5$] {};
		\node[vertex] (v) at (1  ,1) [label=90: $v_3$] {};
		\node[vertex] (v4) at (2  ,1) [label=0: $v_4$] {};
		\node[vertex] (v8) at (2  ,0) [label=0: $v_8$] {};
		\draw (s)--(i); \draw (s)--(o); \draw (s)--(v); \draw (o)--(u);
		\draw (u)--(i); 
		\node[vertex] (x) at (1  ,0) [label=270: $v_7$] {};
		\draw (x)--(i); \draw (v)--(x); \draw[dashed] (v)--(v4); \draw[dashed] (v4)--(v8);
		\draw[dashed] (v8)--(x); \draw (o)--(x);
		\end{scope}
		
		\begin{scope}[xshift=5cm] 
		\node[vertex] (s) at (0  ,1) [label=90: $v_2$] {};
		\node[vertex] (i) at (0  ,0) [label=270: $v_6$] {};
		\node[vertex] (o) at (-1  ,1) [label=180: $v_1$] {};
		\node[vertex] (u) at (-1  ,0) [label=180: $v_5$] {};
		\node[vertex] (v) at (1  ,1) [label=90: $v_3$] {};
		\node[vertex] (v4) at (2  ,1) [label=0: $v_4$] {};
		\node[vertex] (v8) at (2  ,0) [label=0: $v_8$] {};
		\draw (s)--(i); \draw (s)--(o); \draw (s)--(v); \draw (o)--(u);
		\draw (u)--(i); 
		\node[vertex] (x) at (1  ,0) [label=270: $v_7$] {};
		\draw (x)--(i); \draw (v)--(x); \draw[dashed] (v)--(v4); \draw[dashed] (v4)--(v8);
		\draw[dashed] (v8)--(x); \draw (o) to[out=45, in=135] (v);
		\end{scope}
		\end{tikzpicture}
		\medskip
	\end{claimproof}
	
	Thus assume for the remainder of the proof that $J$ is an induced subgraph of $H$.
	Note that $J$ has two subgraphs isomorphic to $F$. We are done unless both have type $\TV$ in $H$ by Lemma~\ref{lem:typeV_hard}. If both have type $\TV$, but Lemma~\ref{lem:even_diamond} is applicable, we are done as well. There is thus only one case (up to symmetry) remaining:  
	\begin{enumerate}[(a)]
		\item $\Gamma_H(v_1)\cap\Gamma_H(v_6) = \{v_2,v_5\}$,
		\item $\Gamma_H(v_3)\cap\Gamma_H(v_6) = \{v_2,v_7\}$,
		\item $\Gamma_H(v_3)\cap\Gamma_H(v_8) = \{v_4,v_7\}$,
		\item $|\Gamma_H(v_2)\cap\Gamma_H(v_5)|$ is odd,
		\item $|\Gamma_H(v_2)\cap\Gamma_H(v_7)|$ is odd, and
		\item $|\Gamma_H(v_4)\cap\Gamma_H(v_7)|$ is odd.
	\end{enumerate}
	We provide an illustration for convenience:
	
	\medskip
	
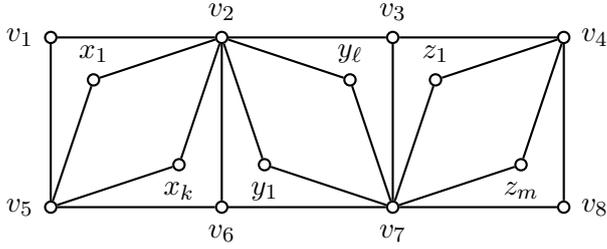
\begin{tikzpicture}[scale=2.25, node distance = 1.4cm, thick] 
	\tikzstyle{vertex}=[draw=black, circle, inner sep=1.5pt] 
	\begin{scope}[xshift=2cm] 
	\node[vertex] (s) at (0  ,1) [label=90: $v_2$] {};
	\node[vertex] (i) at (0  ,0) [label=270: $v_6$] {};
	\node[vertex] (o) at (-1  ,1) [label=180: $v_1$] {};
	\node[vertex] (u) at (-1  ,0) [label=180: $v_5$] {};
	\node[vertex] (v) at (1  ,1) [label=90: $v_3$] {};
	\node[vertex] (v4) at (2  ,1) [label=0: $v_4$] {};
	\node[vertex] (v8) at (2  ,0) [label=0: $v_8$] {};
	\draw (s)--(i); \draw (s)--(o); \draw (s)--(v); \draw (o)--(u);
	\draw (u)--(i); 
	\node[vertex] (x) at (1  ,0) [label=270: $v_7$] {};
	\draw (x)--(i); \draw (v)--(x); \draw (v)--(v4); \draw (v4)--(v8);
	\draw (v8)--(x);
	\node[vertex] (x1) at (-0.75  ,0.75) [label=90: $x_1$] {};
	\node[vertex] (xk) at (-0.25  ,0.25) [label=270: $x_k$] {};
	\node[vertex] (y1) at (0.25  ,0.25) [label=270: $y_1$] {};
	\node[vertex] (yl) at (0.75  ,0.75) [label=90: $y_\ell$] {};
	\node[vertex] (z1) at (1.25  ,0.75) [label=90: $z_1$] {};
	\node[vertex] (zm) at (1.75  ,0.25) [label=270: $z_m$] {};
	\draw (x1)--(u); \draw (x1)--(s); \draw (xk)--(u); \draw (xk)--(s);
	\draw (y1)--(x); \draw (y1)--(s); \draw (yl)--(x); \draw (yl)--(s);
	\draw (z1)--(x); \draw (z1)--(v4); \draw (zm)--(x); \draw (zm)--(v4);
	\end{scope}
	\end{tikzpicture}
	\medskip
	
	\noindent Note that $k$, $\ell$ and $m$ are odd. We construct a hardness gadget:
	\begin{itemize}
		\item $S=\{v_2\}$ and $I=\{v_7\}$.
		\item $J_1$ is the graph where $y$ is adjacent to a $v_1$-pin and a $v_6$-pin. Note that $\Omega_y=\{v_2,v_5\}$ by (a).
		\item $J_2$ is the graph where $z$ is adjacent to a $v_3$-pin and a $v_8$-pin. Note that $\Omega_z=\{v_7,v_4\}$ by (c).
		\item $J_3$ is a path of length $2$ from $y$ to $z$.
	\end{itemize}
	By (d), (e) and (f) we have that
	$|\Sigma_{v_5,v_2}|$, $|\Sigma_{v_2,v_7}|$ and $|\Sigma_{v_7,v_4}|$ are odd. Furthermore, we observe that $|\Sigma_{v_5,v_4}|=0$ as any path of length $2$ from $v_5$ to $v_4$ would create a $K_4$-minor. 
\end{proof}

\begin{lem}\label{lem:zigzag}
	The following graph $J$ is a strong hardness gadgets:
	
	
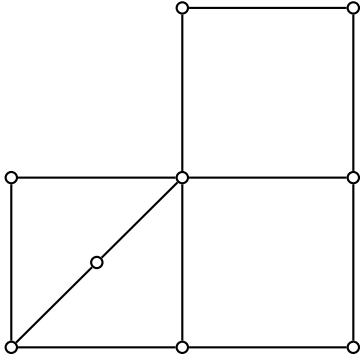
\medskip
	\begin{tikzpicture}[scale=2.25, node distance = 1.4cm, thick, yscale=-1] 
	\tikzstyle{vertex}=[draw=black, circle, inner sep=1.5pt] 
	\begin{scope}[xshift=2cm] 
	\node[vertex] (s) at (0  ,1)  {};
	\node[vertex] (i) at (0  ,0)  {};
	\node[vertex] (o) at (-1  ,1) {};
	\node[vertex] (u) at (-1  ,0) {};
	\node[vertex] (v) at (1  ,1)  {};
	\node[vertex] (v4) at (0  ,-1) {};
	\node[vertex] (v8) at (1  ,-1)  {};
	\draw (s)--(i); \draw (s)--(o); \draw (s)--(v); \draw (o)--(u);
	\draw (u)--(i); 
	\node[vertex] (x) at (1  ,0)  {};
	\draw (x)--(i); \draw (v)--(x); \draw (i)--(v4); \draw (v4)--(v8);
	\draw (v8)--(x);
	\node[vertex] (v9) at (-0.5  ,0.5)  {};
	\draw (v9)--(o); \draw (v9)--(i);
	\end{scope}
	
	\end{tikzpicture}
	\medskip
\end{lem} 
\begin{proof}
	Let $H$ be a $K_4$-minor-free supergraph of $J$. 
	Note that $J$ has two subgraphs isomorphic to $F$. By Lemma~\ref{lem:typeV_hard}, we obtain a hardness gadget of $H$, unless both of the subgraphs isomorphic to~$F$ have type $\TV$. If this is the case, however, we obtain the following subgraph $\hat{J}$ of~$H$:
	
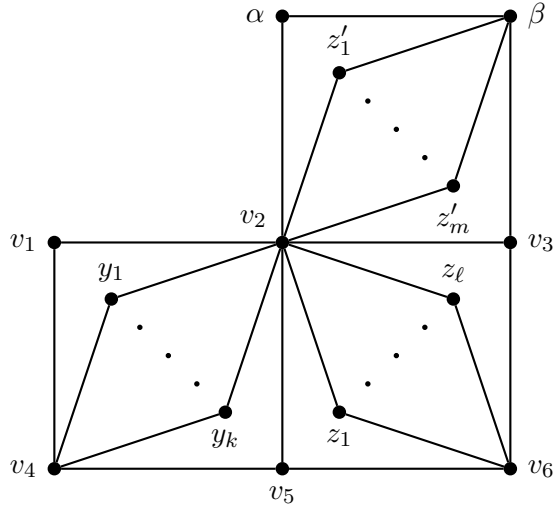
	\begin{center}
		\begin{tikzpicture}[scale=3, node distance = 1.4cm,thick]
		\tikzstyle{dot}   =[fill=black, draw=black, circle, inner sep=0.15mm]
		\tikzstyle{vertex}=[fill=black, draw=black, circle, inner sep=1.5pt]
		\tikzstyle{dist}  =[fill=white, draw=black, circle, inner sep=2pt]
		\tikzstyle{pinned}=[draw=black, minimum size=10mm, circle, inner sep=0pt]
		
		\node[vertex] (s) at (0  ,1) [label=135: $v_2$] {};
		\node[vertex] (i) at (0  ,0) [label=270: $v_5$] {};
		\node[vertex] (o) at (-1  ,1) [label=180: $v_1$] {};
		\node[vertex] (u) at (-1  ,0) [label=180: $v_4$] {};
		\node[vertex] (v) at (1  ,1) [label=0: $v_3$] {};
		\node[vertex] (x) at (1  ,0) [label=0: $v_6$] {};
		
		\draw (s)--(i); \draw (s)--(o); \draw (s)--(v); \draw (o)--(u);
		\draw (u)--(i); \draw (x)--(i); \draw (v)--(x);
		
		\node[vertex] (y1) at (-0.75  ,0.75) [label=90: $y_1$] {};
		\node[vertex] (yk) at (-0.25  ,0.25) [label=270: $y_k$]  {};
		\node[vertex] (zl) at (0.25  ,0.25)  [label=270: $z_1$] {};
		\node[vertex] (z1) at (0.75  ,0.75) [label=90: $z_\ell$] {};
		
		\draw (s)--(y1); \draw (s)--(yk);  \draw (u)--(y1); \draw (u)--(yk);
		\draw (x)--(z1); \draw (x)--(zl); \draw (s)--(z1); \draw (s)--(zl);    
		
		\node[dot] (d) at (-0.5,0.5) {};  
		\node[dot] (d) at (-0.625,0.625) {}; 
		\node[dot] (d) at (-0.375,0.375) {}; 
		\node[dot] (d) at (0.5,0.5) {};  
		\node[dot] (d) at (0.625,0.625) {}; 
		\node[dot] (d) at (0.375,0.375) {}; 
		\node[dot] (d) at (0.5,1.5) {};  
		\node[dot] (d) at (0.625,1.375) {}; 
		\node[dot] (d) at (0.375,1.625) {}; 
		
		\node[vertex] (a) at (0,  2) [label=180: $\alpha$] {};
		\node[vertex] (b) at (1  ,2) [label=0: $\beta$] {};
		\draw (s) -- (a); \draw (a) -- (b); \draw (b) -- (v); 
		
		\node[vertex] (zpm) at (0.75  ,1.25)  [label=270: $z'_m$] {};
		\node[vertex] (zp1) at (0.25  ,1.75) [label=90: $z'_1$] {};
		\draw (s) -- (zp1); \draw (s) -- (zpm); \draw (zpm) -- (b);
		\draw (zp1) -- (b);  
		
		\end{tikzpicture}
	\end{center}
	In $\hat{J}$, $k,\ell,m>0$ and all common neighbours in $H$ between the pairs $(v_2,v_4)$, $(v_2,v_6)$ and $(v_2,\beta)$ are depicted. By definition of type $\TV$, we also obtain that each of the pairs $(v_1,v_5)$, $(v_5,v_3)$ and $(v_3,\alpha)$ has only the two common neighbours in $H$ depicted. Note further, that $H$ has a hardness gadget if at least one of $k,\ell$ or $m$ is even by Lemma~\ref{lem:even_diamond}. Thus assume for the remainder of the proof that all three are odd. We will rely on the following claim, that we can assume $\hat{J}$ to be an \emph{induced} subgraph of $H$:
	
	\medskip
	\noindent \textbf{Claim A:} \textit{If $\hat{J}$ is not an induced subgraph of $H$ then $H$ has a $K_4$-minor or a hardness gadget.}
	\medskip
	
	\begin{claimproof}
		Let $e\in E(H)\setminus E(\hat{J})$ be an edge of $H$ that connects two vertices of $\hat{J}$. We first assume that $e$ connects two vertices in
		\[\{v_1,v_2,v_3,v_4,v_5,v_6,y_1,\dots,y_k,z_1,\dots,z_\ell\}\,.\]
		We show by case distinction that $e$ either yields a hardness gadget, or a $K_4$-minor:
		\begin{enumerate}[(I)]
			\item $x\in e$ for $x\in \{v_1,y_1,\dots,y_k\}$. Let $x'$ be the other endpoint of $e$ and note that $x'\notin\{v_4,v_2,x\}$ as we do not allow self-loops and multiple edges. 
			\begin{enumerate}[(i)]
				\item If $x'\in \{v_1,y_1,\dots,y_k,v_5\}$ then we obtain a $K_4$-minor induced by $x,x',v_2,v_4$ --- note that, as $k>0$, there exists a 2-path from $v_2$ to $v_4$ whose internal vertex is neither $x$ nor $x'$. 
				
				\item If $x' \in \{v_3,z_1,\dots,z_\ell\}$ then we obtain a $K_4$-minor induced by $x,v_2,x',v_5$ --- note that there is a 2-path from $v_5$ to $x$ via $v_4$, and a 2-path from $v_5$ to $x'$ via $v_6$.
				
				\item If $x'=v_6$, then we obtain a $K_4$-minor induced by $x,v_2,v_6,v_5$ --- note that there is a 2-path from $v_5$ to $x$ via $v_4$, and a 2-path from $v_6$ to $v_2$ via $v_3$.
			\end{enumerate}
			\item $x\in e$ for $x\in \{v_3,z_1,\dots,z_\ell\}$. Symmetric to the previous case (I).
			\item $v_4\in e$. Let $x'$ be the other endpoint of $e$ and note that $x'\notin\{v_4,v_1,y_1,\dots,y_k,v_5\}$ as we do not allow self-loops and multiple edges. 
			\begin{enumerate}[(i)]
				\item If $x'\in \{v_3,z_1,\dots,z_\ell\}$ then the case is symmetric to case (I)(iii).
				\item If $x'=v_6$ then we obtain a $K_4$-minor induced by $v_4,v_2,v_6,v_5$ --- note that there is a 2-path from $v_4$ to $v_2$ via $v_1$, and a 2-path from $v_2$ to $v_6$ via $v_3$.
				\item If $x'=v_2$, then $H$ has a hardness gadget by Lemma~\ref{lem:diamond_square_gadget}.
			\end{enumerate}
			\item $v_6\in e$. Symmetric to the previous case (III).
			\item $v_2\in e$. Let $x'$ be the other endpoint of $e$. Thus $x'\notin\{v_2,v_1,y_1,\dots,y_k,v_5,z_1,\dots,z_\ell,v_3\}$ as we do not allow self-loops and multiple edges. The only remaining candidates for $x'$ are thus $v_4$ and $v_6$. However, both of the latter candidates yield a hardness gadget by Lemma~\ref{lem:diamond_square_gadget}.
			\item $v_5\in e$. Let $x'$ be the other endpoint of $e$ and note that $x'\notin\{v_5,v_4,v_2,v_6\}$ as we do not allow self-loops and multiple edges. Similarly as in the previous case (V), all other candidates for $x'$ yield a hardness gadget by Lemma~\ref{lem:diamond_square_gadget}.
		\end{enumerate}
		This concludes the case distinction. Observe now, that a symmetric case analysis shows $H$ has a hardness gadget or a $K_4$-minor if $e$ connects two vertices in
		\[\{v_5,v_2,\alpha,v_6,v_3,\beta,z_1,\dots,z_\ell,z'_1,\dots,z'_m\}\,.\]
		The remaining possibility for $e$ is to have one endpoint in $\{v_4,v_1,y_1,\dots,y_k\}$ and the other endpoint in $\{\alpha,\beta,z'_1,\dots,z'_m\}$. However, in this case, we find a path from $v_5$ to $v_3$ whose vertices are disjoint from $\{v_2,z_1,\dots,z_\ell,v_6\}$. Consequently, we obtain a $K_4$-minor induced by $v_2,v_3,v_5,v_6$.
	\end{claimproof}
	
	\noindent We thus assume that $\hat{J}$ is an induced subgraph of $H$ in what follows. 
	Next, we perform a case distinction on the parity of the degree of $v_2$; in both cases, we construct a hardness gadget.
	\begin{enumerate}[(I)]
		\item $\deg_H(v_2)$ is even. We construct a hardness gadget:
		\begin{itemize}
			\item $I=\{v_4\}$ and $S=\{v_6\}$.
			\item $J_1$ is the graph where $y$ is adjacent to a $v_1$-pin and a $v_5$-pin so $\Omega_y = \{v_2,v_4\}$.
			\item $J_2$ is the graph where $z$ is adjacent to a $v_5$-pin and a $v_3$-pin so $\Omega_z = \{v_2,v_6\}$. 
			\item $J_3$ is a 2-path between $y$ and $z$. 
		\end{itemize}
		As the degree of $v_2$ is even,   $|\Sigma_{v_2,v_2}|$ is even. 
		As $k$ and $\ell$ are odd,   $|\Sigma_{v_2,v_6}|$ and $|\Sigma_{v_2,v_4}|$  are odd. Finally, we claim that 
		$|\Sigma_{v_6,v_4}|$ is odd: Otherwise there must be an additional 2-path from $v_6$ to $v_4$. As $\hat{J}$ is an induced subgraph of $H$, the internal vertex of this path, let us call it $x$, cannot be contained in $V(\hat{J})$; otherwise, $H$ would contain an edge between $x$ and a vertex $v$ of $\hat{J}$ while $x$ and $v$ are not adjacent in $\hat{J}$.
		
		This, however, yields a $K_4$-minor induced by the vertices $v_4,v_2,v_6$ and $v_5$ --- note that $v_4$ and $v_2$ are connected by the 2-path via $v_1$, $v_2$ and $v_6$ are connected by the 2-path via $v_3$, and $v_4$ and $v_6$ are connected by the 2-path via $x$.
		\item $\deg_H(v_2)$ is odd. We construct a hardness gadget:
		\begin{itemize}
			\item $I=S=\{v_2\}$.
			\item $J_1$ is the graph where $y$ is adjacent to a $v_1$-pin and a $v_5$-pin so $\Omega_y = \{v_2,v_4\}$.
			\item $J_2$ is the graph where $z$ is adjacent to an $\alpha$-pin and a $v_3$-pin so $\Omega_z = \{v_2,\beta\}$. 
			\item $J_3$ is a 2-path between $y$ and $z$. 
		\end{itemize}
		As the degree of $v_2$ is odd,   $|\Sigma_{v_2,v_2}|$ is odd. As $k$ and $\ell$ are odd, we have that 
		$|\Sigma_{v_4,v_2}|$ and $|\Sigma_{v_2,\beta}|$ are odd. Finally, we claim that $|\Sigma_{v_4,\beta}|$ is even: Assuming otherwise, there must be at least one 2-path in $H$ from $v_4$ to $\beta$; we show that there is none.
		
		As $\hat{J}$ is an induced subgraph of $H$, the internal vertex of this path, let us call it $x$, cannot be contained in $V(\hat{J})$; otherwise, $H$ would contain an edge between $x$ and a vertex $v$ of $\hat{J}$ while $x$ and $v$ are not adjacent in $\hat{J}$. 
		
		This, however, yields a $K_4$-minor induced by the vertices $v_2,\beta,v_5$ and $v_4$ --- note that $v_4$ and $v_2$ are connected by the 2-path via $v_1$, $v_2$ and $\beta$ are connected by the 2-path via $\alpha$, $v_4$ and $\beta$ are connected by the 2-path via $x$, and $v_5$ and $\beta$ are connected by the 3-path via $v_6$ and $v_3$. 
	\end{enumerate}
\end{proof}

\subsection{Chordal Bipartite Component Lemma}

\begin{defn}[(1,2)-supergraph]\label{def:12supergraph}
	Let $J$ be a connected graph.
	We say that a supergraph~$H$ of~$J$ is a \emph{(1,2)-supergraph} of $J$ 
	if every edge of~$H$ connecting vertices of~$J$ is also an edge of~$J$
	and every length-$2$ path of~$H$ connecting vertices of~$J$ is also a path of~$J$.  
\end{defn}

For what follows, recall that a chordal bipartite graph is a graph in which every induced cycle is a square.
The following notion captures the $K_4$-minor-free (biconnected) graphs that are obtained by gluing squares together without inducing $\parP$-hardness.

\begin{defn}[impasse, pair of connectors]\label{def:impasse}
	A $K_4$-minor-free biconnected graph $B$ is called an \emph{impasse} if there are odd positive integers $k$ and $\ell$ such that $B$ is a (1,2)-supergraph of the graph $S_{k,\ell}$. 
	Also, with the vertex labels from Definition~\ref{def:Skl}, all of the vertices in $\{v_1,y_1,\dots,y_k,v_3,z_1,\dots,z_\ell\}$ are required to have degree~$2$ in~$B$.
	The pair $(v_1, v_3)$ is called a \emph{pair of connectors} of the impasse $B$. (Note that a pair of connectors of $B$ is not unique as, for instance, $(v_1, z_1)$ is also a pair of connectors.)
\end{defn}
The graph in Figure~\ref{fig:typeV} is an example of an impasse. 

\begin{defn}[diamond]\label{def:diamond}
	A biconnected graph $B$ is a \emph{diamond} if,
	for an integer $k\geq 2$, 
	$V(B) =  \{s,t,x_1,\dots,x_k\}$   and
	$E(B) = \cup_{i\in [k]} \{ \{s,x_i\},\{x_i,t\} \}$.
\end{defn}

Note that a square is a diamond with $k=2$.
The following lemma classifies biconnected chordal bipartite graphs:

\begin{lem}[Chordal Bipartite Component Lemma]\label{lem:main_chordal_bipartite}\label{lem:X}
	Let $H$ be a $K_4$-minor-free graph and let $B$ be a biconnected component of $H$. If $B$ is chordal bipartite and not just a single edge, then at least one of the following is true:
	\begin{enumerate}[(a)]
		\item $B$ is a diamond.
		\item $H$ has a hardness gadget.
		\item $B$ is an impasse.
	\end{enumerate}
\end{lem}
\begin{proof}
	As $B$ is biconnected, chordal bipartite and not a single edge, there exists an induced square $C=(a,b,c,d,a)$ in $B$. Let us write $\Gamma_{H\setminus C}(a,c)$ for the set $\Gamma_H(a)\cap\Gamma_H(c)\setminus\{b,d\}$ and 
	$\Gamma_{H\setminus C}(b,d)$ for the set $\Gamma_H(b)\cap\Gamma_H(d)\setminus\{a,c\}$. 
	Since $B$ is a biconnected component of $H$, and $a,b,c,d \in B$, we actually have that
	$
	\Gamma_{H\setminus C}(a,c)
	= \Gamma_B(a)\cap\Gamma_B(c)\setminus\{b,d\}$
	and $
	\Gamma_{H\setminus C}(b,d) 		  = \Gamma_B(b)\cap\Gamma_B(d)\setminus\{a,c\}$.
	As $H$ is $K_4$-minor free, we observe that at least one of $\Gamma_{H\setminus C}(a,c)$ and $\Gamma_{H\setminus C}(b,d)$ is empty. Assume w.l.o.g., that $\Gamma_{H\setminus C}(b,d)$	  is empty. 
	Let $B'$ be the graph consisting of $C$  together with the edges from~$a$ and $c$ to $\Gamma_{H\setminus C}(a,c)$.
	If $B=B'$ 
	then $B$ is a diamond. Otherwise, as $B$ is biconnected, there is a shortest path $P$ in $B$ connecting two vertices of 
	$C\cup \Gamma_{H\setminus C}(a,c)$ whose internal vertices are not in $B'$.
	This path~$P$ has an internal vertex since $B$ has no triangle. 
	
	\medskip
	\noindent \textbf{Claim A:} \textit{$P$ has length $3$, one endpoint of $P$ is contained in $\Gamma_B(a)\cap \Gamma_B(c)$ and the other endpoint is contained in $\{a,c\}$. }
	\medskip
	
	\begin{claimproof}
		Assume first, for contradiction, that both endpoints of $P$, let us call them $s$ and $t$, are in $\Gamma_B(a)\cap \Gamma_B(c)$. The only possible length for $P$ under this assumption is $2$, as, otherwise, we obtain an induced cycle $(a,s,P,t,a)$ of length $\neq 4$. As $P$ must have length $2$, the endpoints of $P$ cannot be $b$ and $d$, as $\Gamma_{H\setminus C}(b,d)$ is empty. Thus we can assume w.l.o.g.\ that $s\neq b$ and $t\neq b$, which yields the following $K_4$-minor; $P$ is depicted dashed:
		
		\medskip
		
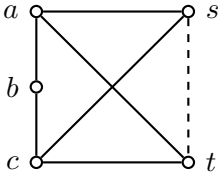
\begin{tikzpicture}[scale=2, node distance = 1.4cm, thick] 
		\tikzstyle{vertex}=[draw=black, circle, inner sep=1.5pt] 
		\begin{scope}[xshift=2cm] 
		\node[vertex] (s) at (1  ,1) [label=0: $s$] {};
		\node[vertex] (i) at (0  ,0) [label=180: $c$] {};
		\node[vertex] (o) at (0  ,1) [label=180: $a$] {};
		\node[vertex] (u) at (0  ,0.5) [label=180: $b$] {};
		\draw (s)--(i); \draw (s)--(o); \draw (o)--(u);
		\draw (u)--(i); 
		\node[vertex] (x) at (1  ,0) [label=0: $t$] {};
		\draw (x)--(i);  \draw (o) -- (x); \draw[dashed] (s)--(x); 
		
		\end{scope}
		\end{tikzpicture}
		\medskip
		
		\noindent This yields the desired contradiction.
		
		Next, if $P$ starts in $a$ and ends in $c$, then we obtain an induced cycle that is not a square, unless $P$ as length $2$. However, in the latter case, the internal vertex of $P$ is contained in $\Gamma_{H\setminus C}(a,c)$, contradicting the fact that $P$ is not fully contained in $C\cup \Gamma_{H\setminus C}(a,c)$.
		This shows that one endpoint of $P$ is in $\Gamma_B(a)\cap \Gamma_B(c)$ and the other endpoint is in $\{a,c\}$. 
		
		Recall that the length of~$P$ is greater than~$1$ (since it has internal vertices).  If $P$ has length $2$, then we obtain a triangle, contradicting the fact that $B$   is chordal bipartite. Finally, if $P$ has length at least $4$, we obtain an induced cycle of length at least $5$, also contradicting chordal-bipartiteness. Consequently, $P$ must have length $3$.
	\end{claimproof}
	
	\noindent Claim A yields that $B$ contains a subgraph isomorphic to the graph $F$ --- recall from Definition~\ref{def:graphF} that $F$ is just the graph containing two squares that share one edge. We use the vertex labels from Figure~\ref{fig:F}, i.e., the vertices are $\{v_1, \ldots, v_6\}$. Now assume that (b) is not true, i.e., that $H$ does not have a hardness gadget. Using the fact that $H$ is $K_4$-minor free, and invoking Lemma~\ref{lem:typeV_hard} we obtain that $F$ has to be of type $\TV$. So, without loss of generality (by renaming), we can assume that $\Gamma_B(v_4) \cap \Gamma_B(v_2) = \{v_1,y_1,\ldots,y_k,v_5\}$ 
	and $\Gamma_B(v_6) \cap \Gamma_B(v_2) = \{v_5,z_1,\ldots,z_\ell,v_3\}$ for some $k,\ell\ge 1$. Consequently, since $H$ is $K_4$-minor-free $\Gamma_B(v_1) \cap \Gamma_B(v_5) = \{v_2, v_4\}$ and $\Gamma_B(v_3) \cap \Gamma_B(v_5) = \{v_2, v_6\}$ have to hold. By Lemma~\ref{lem:even_diamond}, we obtain 
	that~$k$ and~$\ell$ have to be odd. 
	So we have shown that $B$ contains $S_{k,\ell}$ (Definition~\ref{def:Skl}) as a subgraph.


	\medskip
	\noindent \textbf{Claim B} \textit{$S_{k,\ell}$ is an induced subgraph of $B$.}
	\medskip
	
	\begin{claimproof}
		Assume that $S_{k,\ell}$ is not an induced subgraph. Then $B$ (equivalently, $H$) contains an edge $e\notin E(S_{k,\ell})$ between two vertices of $S_{k,\ell}$.
		We need to distinguish a variety of (simple) cases:
		\begin{itemize}
			\item $v_4 \in e$: The other endpoint of $e$ cannot be one of $v_4,v_5,y_1,\dots,y_k,v_1$ as we do not allow self-loops and multi-edges. Further, it cannot be $v_6$ or $v_2$, as this would create a triangle, contradicting the fact that $B$ is chordal-bipartite. Finally, if the other endpoint of $e$ is $x\in \{v_3,z_1,\dots,z_\ell\}$, then we obtain a $K_4$-minor induced by the vertices $v_4,v_5,v_2$ and $x$ --- note that there is a 2-path from $x$ to $v_5$ via $v_6$, and a 2-path from $v_2$ to $v_4$ via $v_1$.
			\item $v_6 \in e$: Symmetric to the previous case.
			\item $x\in e$ for some $x\in \{v_1,y_1,\dots,y_k\}$: The other endpoint of $e$ cannot be one of \[v_1,y_1,\dots,y_k,v_4,v_2,v_5,v_3,z_1,\dots,z_\ell\,,\] as each of those cases would yield a self-loop, a multi-edge, or a triangle (in $B$). The remaining candidate for the other endpoint is $v_6$, which was covered in the previous case.
			\item $x\in e$ for some $x\in \{v_3,z_1,\dots,z_\ell\}$: Symmetric to the previous case.
			\item $v_5\in e$: Any (additional) edge from $v_5$ to a vertex of $S_{k,\ell}$ would create either a multi-edge, a self-loop, or a triangle.
			\item $v_2\in e$: The other endpoint of $e$ cannot be $v_2$ as we this would create a self-loop. Consequently, one of the previous cases must be true for the other endpoint of $e$.
		\end{itemize}
	\end{claimproof}
	
	Recall that we want to show that (a)~$B$ is a diamond, (b)~$H$ has a hardness gadget, or (c)~$B$ is an impasse.
	For what follows, we distinguish two cases: 
	\begin{enumerate}[(I)]
		\item All vertices $v_1,y_1,\dots,y_k,z_1,\dots,z_\ell,v_3$ have degree $2$ in $B$.
		In this case we will show that $B$ is a (1,2)-supergraph of~$S_{k,\ell}$.
		This implies (see Definition~\ref{def:impasse}) that $B$ is an impasse, so we are finished.
		To see that $B$ is a (1,2)-supergraph of~$S_{k,\ell}$, recall (from Claim~B) that $S_{k,\ell}$ is an induced subgraph of~$B$.
		All neighbours of $v_1,y_1,\dots,y_k,z_1,\dots,z_\ell,v_3$	
		in~$B$ are included in~$S_{k,\ell}$.
		Thus, it suffices to show that $B$ has no 2-path connecting vertices in $\{v_4,v_5,v_6,v_2\}$ whose internal vertex~$x$, is outside of~$S_{k,\ell}$.
		We noted above that $\Gamma_B(v_4) \cap \Gamma_B(v_2) \subseteq V(S_{k,\ell})$ and
		$\Gamma_B(v_6) \cap \Gamma_B(v_2) \subseteq V(S_{k,\ell})$.
		There is no 2-path in~$B$ from $v_2$ to $v_5$ because that would yield a triangle in~$B$.
		Similarly, 2-paths from $v_5$ to~$v_4$ or~$v_6$ would yield triangles in~$B$,
		so the only possibility is a 2-path from~$v_4$ to~$v_6$ but this would yield the $K_4$-minor $\{v_4,v_5,v_6,v_2\}$ in~$B$, contradicting the fact that $H$ (hence $B$) has no $K_4$-minor.

		\item Otherwise, assume w.l.o.g.\ that $v_1$ has degree at least $3$ in $B$. As $B$ is biconnected, there exists a shortest path $P$ in the remainder of $B$ connecting $v_1$ with another vertex $w$ of $S_{k,\ell}$. We claim that the only candidates for $w$ are $v_4$ and $v_2$, which we will prove by case distinction:
		\begin{itemize}
			\item $w\in \{y_1,\dots,y_k,v_5\}$. Then we obtain a $K_4$-minor: $(v_4,w,v_2,v_1,v_4)$ is a square, $P$ connects $v_1$ and $w$ via vertices not contained in $S_{k,\ell}$, and $v_4$ and $v_2$ are connected by a 2-path via a vertex $x\in \{y_1,\dots,y_k,v_5\}\setminus w$ --- note that $x$ exists as $k\geq 1$.
			\item $w=v_6$. Then we obtain a $K_4$-minor induced by the vertices $v_5,v_6,v_1$ and $v_2$ --- note that $v_1$ is connected to $v_5$ by the 2-path via $v_4$, and that $v_2$ is connected to $v_6$ by the 2-path via $v_3$. 
			\item $w\in \{z_1,\dots,z_\ell,v_3\}$. Then we obtain a $K_4$-minor induced by the vertices $v_5,v_1,v_2$ and $w$ --- note that $v_1$ is connected to $v_5$ by the 2-path via $v_4$, and that $v_5$ is connected to $w$ by the 2-path via $v_6$. 
		\end{itemize}
		
		Consequently, $w$ must either be $v_4$ or $v_2$ as all other possibilities create a $K_4$-minor. As $B$ is chordal bipartite, $P$ must have length three. However, if $P$ connects $v_1$ and $v_2$, we obtain a strong hardness gadget by Lemma~\ref{lem:zigzag}, and if $P$ connects $v_1$ and $v_4$, we obtain a strong hardness gadget by Lemma~\ref{lem:3x1}. In both cases, $H$ therefore has  a hardness gadget.
	\end{enumerate}
\end{proof}

The following lemma shows that impasses already yield hardness if the vertex $v_2$ has even degree:
\begin{lem}\label{lem:hard_impasses}
	Let $H$ be a graph containing an impasse $B$ as biconnected component, that is, there are odd integers $k$ and $\ell$ such that $B$ is a (1,2)-supergraph of 
	the graph $S_{k,\ell}$ such that, using the vertex labels from Figure~\ref{fig:Skl}, all vertices $v_1,y_1,\dots,y_k,v_3,z_1,\dots,z_\ell$ have degree $2$ in $B$.
	If $\deg_H(v_2)$ is even, then $H$ has a hardness gadget.
\end{lem}
\begin{proof}
	We construct a hardness gadget:
	\begin{itemize}
		\item $I=\{v_4\}$ and $S=\{v_6\}$.
		\item $J_1$ is the graph where $y$ is adjacent to a $v_1$-pin and a $v_5$-pin so $\Omega_y = \{v_2,v_4\}$ as 
		$H$ has the impasse~$B$ as a biconnected component.
		\item $J_2$ is the graph where $z$ is adjacent to a $v_5$-pin and a $v_3$-pin so $\Omega_z = \{v_2,v_6\}$ as 
		$H$ has the impasse~$B$ as a biconnected component. 
		\item $J_3$ is a 2-path between $y$ and $z$. 
	\end{itemize}
	As the degree of $v_2$ is even,  $|\Sigma_{v_2,v_2}|$ is even. As $k$ and $\ell$ are odd, we have that 
	$|\Sigma_{v_2,v_4}|$ and $ |\Sigma_{v_2,v_6}|$ are odd. Finally, we also have $|\Sigma_{v_4,v_6}|=1$ as an additional 2-path from $v_4$ to $v_6$ would 
	contradict the fact that the biconnected component~$B$ of~$H$   is an impasse. 
\end{proof}

\section{Sequences of Chordal Bipartite Components}\label{sec:generalisingCaterpillar}
In 
Section~\ref{sec:chordalBipartiteComps}, we proved 
(Lemma~\ref{lem:main_chordal_bipartite})
that every chordal bipartite biconnected component $B$ of a $K_4$-minor-free graph $H$ is an edge, a diamond, or an impasse, or the graph $H$ has a hardness gadget. The goal of the current section is to establish a structural property (Lemma~\ref{lem:CaterpillarForBipChordalComponents}) of a graph $H$, informally stating that a path in a block-cut-tree of $H$ consisting only of edges, diamonds, and impasses either induces a hardness gadget of $H$, or has endpoints that satisfy a technical criterion necessary for our construction of global hardness gadgets in Section~\ref{sec:k4MinorFreeMain}.  

\begin{defn}[good start, good stop]\label{def:goodstartgoodstop}
	Let $H$ be a graph  and let $B$ be a subgraph of $H$. Let $y$ be a vertex in $B$ and let $L_B\subseteq \NH{y}\cap V(B)$.
	\begin{myitemize}
		\item We say that $(L_B,y)$ is a \emph{good start} in $B$ if there is a gadget $(J, z)$ such that $\{ v \in V(H) \mid \>\abs{\hom{(J, z)}{(H,v)}}\text{ is odd}\,\}=L_B \cup R_B$, where $\abs{L_B}$ is odd and $R_B=\NH{y}\setminus V(B)$.
		\item We say that $(L_B,y)$ is a \emph{good stop} in $B$ if it is a good start in~$B$ and $\abs{R_B}$ is odd. 
	\end{myitemize}
\end{defn}

For non-negative integers $k$ and $\ell$, we define some (classes of) graphs with a pair of distinguished vertices $a$ and $b$ each, see Figure~\ref{fig:bipbiconnchordalComponents} (The graph $S_{k,\ell}$ was already defined in Definition~\ref{def:impasse}, however, for the scope of this section it will be more convenient to work with the vertex labels as given in Figure~\ref{fig:bipbiconnchordalComponents}.).

\begin{figure}[ht]
	\centering
	\begin{minipage}{.24 \textwidth}
		\centering
		\begin{tikzpicture}[scale=2.5, node distance = 1.4cm,thick]
		\tikzstyle{dot}   =[fill=black, draw=black, circle, inner sep=0.15mm]
		\tikzstyle{vertex}=[draw=black, circle, inner sep=1.5pt]
		\tikzstyle{terminal}=[fill=black, draw=black, circle, inner sep=1.5pt]
		\tikzstyle{dist}  =[fill=white, draw=black, circle, inner sep=2pt]
		\tikzstyle{pinned}=[draw=black, minimum size=10mm, circle, inner sep=0pt]	
		\begin{scope}[xshift=2cm] 
		\node[vertex] (s) at (0  ,1) [label=0: $c$] {};
		\node[terminal] (i) at (0  ,0) [label=0: $b$] {};
		\node[vertex] (o) at (-1  ,1) [label=180: $d$] {};
		\node[terminal] (u) at (-1  ,0) [label=180: $a$] {};	
		\draw (s)--(i); \draw (s)--(o); \draw (o)--(u);
		\draw (u)--(i);	
		\node[vertex] (y1) at (-0.75  ,0.25) [label=270: $y_1$] {};
		\node[vertex] (yk) at (-0.25  ,0.75) [label=90: $y_k$]  {};	
		\draw (o)--(y1); \draw (o)--(yk);  \draw (i)--(y1); \draw (i)--(yk);   	
		\node[dot] (d) at (-0.5,0.5) {};  
		\node[dot] (d) at (-0.375,0.625) {}; 
		\node[dot] (d) at (-0.625,0.375) {}; 
		\end{scope}
		\end{tikzpicture} 
	\end{minipage}
	\centering
	\begin{minipage}{.24 \textwidth}
		\centering
		\begin{tikzpicture}[scale=2.5, node distance = 1.4cm,thick]
		\tikzstyle{dot}   =[fill=black, draw=black, circle, inner sep=0.15mm]
		\tikzstyle{vertex}=[draw=black, circle, inner sep=1.5pt]
		\tikzstyle{terminal}=[fill=black, draw=black, circle, inner sep=1.5pt]
		\tikzstyle{dist}  =[fill=white, draw=black, circle, inner sep=2pt]
		\tikzstyle{pinned}=[draw=black, minimum size=10mm, circle, inner sep=0pt]	
		\begin{scope}[xshift=2cm] 
		\node[vertex] (s) at (0  ,1) [label=180: $d$] {};
		\node[terminal] (i) at (0  ,0) [label=180: $a$] {};
		\node[vertex] (v) at (1  ,1) [label=0: $c$] {};
		\node[terminal] (x) at (1  ,0) [label=0: $b$] {};	
		\draw (s)--(i); \draw (s)--(v);
		\draw (x)--(i); \draw (v)--(x);	
		\node[vertex] (zl) at (0.25  ,0.75)  [label=90: $z_1$] {};
		\node[vertex] (z1) at (0.75  ,0.25) [label=270: $z_k$] {};	
		\draw (v)--(z1); \draw (v)--(zl); \draw (i)--(z1); \draw (i)--(zl);    	
		\node[dot] (d) at (0.5,0.5) {};  
		\node[dot] (d) at (0.375,0.625) {}; 
		\node[dot] (d) at (0.625,0.375) {}; 
		\end{scope}
		\end{tikzpicture} 
	\end{minipage}
	\centering
	\begin{minipage}{.48 \textwidth}
		\centering
		\begin{tikzpicture}[scale=2.5, node distance = 1.4cm,thick]
		\tikzstyle{dot}   =[fill=black, draw=black, circle, inner sep=0.15mm]
		\tikzstyle{vertex}=[draw=black, circle, inner sep=1.5pt]
		\tikzstyle{terminal}=[fill=black, draw=black, circle, inner sep=1.5pt]
		\tikzstyle{dist}  =[fill=white, draw=black, circle, inner sep=2pt]
		\tikzstyle{pinned}=[draw=black, minimum size=10mm, circle, inner sep=0pt]	
		\begin{scope}[xshift=2cm] 
		\node[vertex] (s) at (0  ,1) [label=90: $m_2$] {};
		\node[vertex] (i) at (0  ,0) [label=270: $m_1$] {};
		\node[vertex] (o) at (-1  ,1) [label=180: $d$] {};
		\node[terminal] (u) at (-1  ,0) [label=180: $a$] {};
		\node[vertex] (v) at (1  ,1) [label=0: $c$] {};
		\node[terminal] (x) at (1  ,0) [label=0: $b$] {};	
		\draw (s)--(i); \draw (s)--(o); \draw (s)--(v); \draw (o)--(u);
		\draw (u)--(i); \draw (x)--(i); \draw (v)--(x);	
		\node[vertex] (y1) at (-0.75  ,0.25) [label=270: $y_1$] {};
		\node[vertex] (yk) at (-0.25  ,0.75) [label=90: $y_k$]  {};
		\node[vertex] (zl) at (0.25  ,0.75)  [label=90: $z_1$] {};
		\node[vertex] (z1) at (0.75  ,0.25) [label=270: $z_\ell$] {};	
		\draw (o)--(y1); \draw (o)--(yk);  \draw (i)--(y1); \draw (i)--(yk);
		\draw (v)--(z1); \draw (v)--(zl); \draw (i)--(z1); \draw (i)--(zl);    	
		\node[dot] (d) at (-0.5,0.5) {};  
		\node[dot] (d) at (-0.375,0.625) {}; 
		\node[dot] (d) at (-0.625,0.375) {}; 
		\node[dot] (d) at (0.5,0.5) {};  
		\node[dot] (d) at (0.375,0.625) {}; 
		\node[dot] (d) at (0.625,0.375) {}; 
		\end{scope}
		\end{tikzpicture} 
	\end{minipage}
	\caption{The graphs $BD_k$ (for ``backward diamond''), $FD_k$ (for ``forward diamond'') and $S_{k,\ell}$ (from left to right).}
	\label{fig:bipbiconnchordalComponents}
\end{figure}
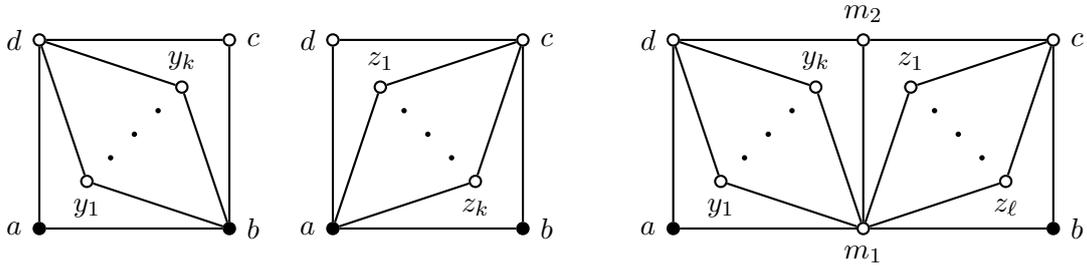

\subsection{Good Starts}

\begin{lem}\label{lem:EdgeGoodStart}
	Let $B$ be a biconnected component of a graph $H$, where $B$ is   an edge between vertices $a$ and $b$. 
	Then  $(\{a\},b)$ is a good start in $B$.
\end{lem}
\begin{proof}
	Clearly, $\{a\}$ has odd cardinality, and is contained in $\Gamma_H(b) \cap V(B)$.
	Let $(J_B, z_B)$ be the gadget where $z_B$ is adjacent to a $b$-pin and let 
	$R_B= \NH{b}\setminus\{a\}$.
	Then
	$\{ v \in V(H) \mid \>\abs{\hom{(J_B, z_B)}{(H,v)}}\text{ is odd}\,\}=\NH{b}=\{a\} \cup R_B$, as desired. 
\end{proof}

\begin{lem}\label{lem:FDGoodStart}
	Let $B$ be a biconnected component of a graph $H$ such that, for an \emph{even} non-negative integer $k$, $B$ is a graph of the form $FD_k$ and 
	the vertices~$a$ and~$b$ are as given in Figure~\ref{fig:bipbiconnchordalComponents}. Let $A$ be a subgraph of $H$ 
	such that $V(A)\cap V(B)=\{a\}$
	Suppose that $L_A\subseteq \NH{a}\cap V(A)$.
	If $(L_A, a)$ is a good start in~$A$ but not a good stop in~$A$ then $(\{a\},b)$ is a good start in~$B$.
\end{lem}
\begin{proof}
	By the definition of a good start,
	$|L_A|$ is odd and
	there is a gadget $(J_A, z_A)$ such that $\{ v \in V(H) \mid \>\abs{\hom{(J_A, z_A)}{(H,v)}}\text{ is odd}\,\}=L_A \cup R_A$ where 
	$R_A= \NH{a}\setminus V(A)$. Since $(L_A,a)$ is not a good stop in~$A$,   $\abs{R_A}$ is even. 
	
	Let $L_B = \{a\}$. 
	We now prove the lemma by showing that $(L_B,b)$ is a good start in $B$.
	Clearly, $|L_B|$ is odd.
	
	Let $(J_B, z_B)$ be the gadget where $z_B$ is adjacent to the vertex $z_A$ of the gadget $J_A$ and it is also adjacent to a $b$-pin. 
	In order to prove that $(L_B,b)$ is a good start we check that $\{ v \in V(H) \mid \>\abs{\hom{(J_B, z_B)}{(H,v)}}\text{ is odd}\,\}=L_B \cup R_B$, where $L_B=\{a\}$ and $R_B=\NH{b}\setminus V(B)$.  
	
	Since $z_B$ is adjacent to a $b$-pin we need only consider each $v\in \NH{b}$ and homomorphisms with $z_B\mapsto v$. Then $z_A$ is also adjacent to $z_B$ and can be mapped to every vertex in the set $\NH{v}\cap(L_A \cup R_A)$. We determine the cardinality of this set depending on $v$:
	\begin{myitemize}
		\item If $v=a$ then $\NH{v}\cap(L_A \cup R_A)=L_A \cup R_A$ and $\abs{L_A \cup R_A}$ is odd, as required.
		\item If $v=c$ (for $c$ as given in Figure~\ref{fig:bipbiconnchordalComponents}) then $v=c$ does not have any neighbours in $L_A$ since every path from $L_A$ to $B$ goes through $a$ because $B$ is a biconnected component of~$H$. Hence, $\NH{v}\cap(L_A \cup R_A)=\NH{v}\cap R_A$ and $\NH{v}\cap R_A= \NH{c} \cap (\NH{a}\setminus V(A))$ by definition of $R_A$. Finally, since $B$ is a biconnected component, the vertices of $B$ have no common neighbours outside of $B$. Thus $\NH{c} \cap (\NH{a}\setminus V(A))=\{d,b,z_1, \dots, z_k\}$, which has even cardinality, as required (since $k$ is even).
		\item If $v\in \NH{b}\setminus V(B)$ then   since $B$ is a biconnected component we have $\NH{v}\cap \NH{a} =\{b\}$. Consequently, $\NH{v}\cap(L_A \cup R_A)=\{b\}$, which is odd, as required.
	\end{myitemize} 
\end{proof}

\begin{lem}\label{lem:BDGoodStart}
	Let $B$ be a biconnected component of a graph $H$ such that,  for a non-negative integer $k$, $B$ is a graph of the form $BD_k$ 
	and the vertices~$a$ and~$b$ are as given in Figure~\ref{fig:bipbiconnchordalComponents}. Let $A$ be a subgraph of $H$ 
	such that   $V(A)\cap V(B)=\{a\}$.
	Suppose that  $L_A\subseteq \NH{a}\cap V(A)$.
	If $(L_A, a)$ is a good start in $A$ but not a good stop in $A$ then $(\{a\},b)$ is a good start in $B$.
\end{lem}
\begin{proof}
	The proof is analogous to that of Lemma~\ref{lem:FDGoodStart}. We define $L_B=\{a\}$ and 
	use the same gadget and again have to consider each $v\in \NH{b}$ and homomorphisms with $z_B\mapsto v$ and consequently determine the cardinality of the set $\NH{v}\cap(L_A \cup R_A)$ depending on $v$:
	\begin{myitemize}
		\item If $v=a$ then $\NH{v}\cap(L_A \cup R_A)=L_A \cup R_A$ and $\abs{L_A \cup R_A}$ is odd, as required.
		\item If $v\in \{c, y_1, \dots y_k\}$ (as given in Figure~\ref{fig:bipbiconnchordalComponents}) then $v$ does not have any neighbours in $L_A$ since every path from $L_A$ to $B$ goes through $a$. Hence, $\NH{v}\cap(L_A \cup R_A)=\NH{v}\cap R_A$ and $\NH{v}\cap R_A= \NH{v} \cap (\NH{a}\setminus V(A))$ by definition of $R_A$.  Finally, since $B$ is a biconnected component, the vertices of $B$ have no common neighbours outside of $B$, $\NH{v} \cap (\NH{a}\setminus V(A))=\{b,d\}$, which has even cardinality, as required.
		\item If $v\in \NH{b}\setminus V(B)$ then since $B$ is a biconnected component we have  $\NH{v}\cap \NH{a} =\{b\}$. Consequently, $\NH{v}\cap(L_A \cup R_A)=\{b\}$, which is odd, as required.
	\end{myitemize} 
\end{proof}

\begin{lem}\label{lem:SklGoodStart}
	Let $B$ be a biconnected component of a graph $H$, where $B$ is an impasse (Definition~\ref{def:impasse}). Let $(a,b)$ be a pair of connectors of $B$ and let $m_1$ be the unique common neighbour of $a$ and $b$ in $H$ (see Figure~\ref{fig:bipbiconnchordalComponents}). Suppose further that $\deg_H(m_1)$ is odd. Let $A$ be a subgraph of $H$  such that $V(A)\cap V(B)=\{a\}$. 
	Suppose that $L_A\subseteq \NH{a}\cap V(A)$.
	If $(L_A, a)$ is a good start in $A$ but not a good stop in $A$ then $(\{m_1\},b)$ is a good start in $B$.
\end{lem}
\begin{proof}
	By the definition of a good start, $|L_A|$ is odd and
	there is a gadget $(J_A, z_A)$ such that $\{ v \in V(H) \mid \>\abs{\hom{(J_A, z_A)}{(H,v)}}\text{ is odd}\,\}=L_A \cup R_A$, where   $R_A= \NH{a}\setminus V(A)$. Since $(L_A,a)$ is not a good stop, $\abs{R_A}$ is even. 
	
	Let $L_B = \{m_1\}$. We now prove the lemma by showing that $(L_B,b)$ is a good start in~$B$.
	Clearly, $|L_B|$ is odd.
	
	Let $(J_B, z_B)$ be the gadget that consists of the gadget $J_A$ joined with a path of length $2$ from the vertex $z_A$ to the vertex $z_B$, and a $b$-pin that is adjacent to $z_B$.
	In order to prove that $(L_B,b)$ is a good start we check that $\{ v \in V(H) \mid \>\abs{\hom{(J_B, z_B)}{(H,v)}}\text{ is odd}\,\}=L_B \cup R_B$, where $L_B=\{m_1\}$ and $R_B=\NH{b}\setminus V(B)$.  
	
	Since $z_B$ is adjacent to a $b$-pin we need only consider $v\in \NH{b}$ and homomorphisms with $z_B\mapsto v$. Then there is a path of length $2$ from $z_A$ to $z_B$ and therefore, for $v\in \NH{b}$,
	\[
	\abs{\hom{(J_B, z_B)}{(H,v)}} = \abs{\{u\in L_A \cup R_A \mid\enspace \abs{\NH{u} \cap \NH{v}} \text{ is odd.}\}}.
	\] 
	We determine $\abs{\{u\in L_A \cup R_A \mid \enspace \abs{\NH{u} \cap \NH{v}} \text{ is odd.}\}}$ depending on $v$ and using the vertex labels from Figure~\ref{fig:bipbiconnchordalComponents}.
	Note that $m_1$ and $c$ are the only neighbours of~$b$ in~$B$ since the degree of~$b$ is~$2$ in~$B$ (by the definition of an impasse).
	
	\begin{myitemize}
		\item Consider $v=m_1$. 
		\begin{itemize}
			\item If $u\in \NH{a}\setminus \{d,m_1\}$ then $u\notin V(B)$ since $\deg_B(a)=2$. As $B$ is a biconnected component, it follows that $a$ is the only common neighbour of $v=m_1$ and $u$.
			\item
			The vertices $v=m_1$ and $u=d$ have an odd number of common neighbours 
			in $S_{k,\ell}$ since $k$ is odd.  
			They have no further common neighbours in~$B$, since $B$ is an impasse, and no further common neighbours in~$H$
			since $B$ is a biconnected component of~$H$. 
			\item Finally, $v=m_1$ and $u=m_1$ have an odd number of
			common neighbours  since $\deg_H(m_1)$ is odd by assumption of the lemma.
		\end{itemize}
		Therefore, $\{u\in L_A \cup R_A \mid \enspace \abs{\NH{u} \cap \NH{v}} \text{ is odd}\}=L_A \cup R_A$ and $\abs{L_A \cup R_A}$ is odd, as required.
		
		\item Consider $v=c$. 
		\begin{itemize}
			\item If $u\in \NH{a}\setminus \{d,m_1\}$, then $u\notin V(B)$ and, as $B$ is a biconnected component, $v=c$ and $u$ have no common neighbours.
			\item The vertices $v=c$ and $u=d$ have one common neighbour in $S_{k,\ell}$
			(the vertex $m_2$) and no further common neighbours in~$H$ (by the same argument as we used for $v=m_1$), so 
			$v=c$ and $u=d$
			have an odd number of common neighbours in~$H$. 
			\item Finally,  $v=c$ and $u=m_1$ have an odd number of common neighbours (since $\ell$ is odd). 
		\end{itemize}
		Therefore, $\{u\in L_A \cup R_A \mid \enspace \abs{\NH{u} \cap \NH{v}} \text{ is odd}\}=\{d,m_1\}$ which has even cardinality, as required.
		\item Consider $v\in \NH{b}\setminus V(B)$.  
		\begin{itemize}
			\item
			If $u\in \NH{a}\setminus \{d,m_1\}$ then $u\notin V(B)$ (since $\deg_B(a)=2$) and, as $B$ is a biconnected component, $v$ and $u$ have no common neighbours.
			\item
			If $u=d$ then $\{u,b\}$ is not an edge of~$B$ (by the definition of impasse) so it is not an edge of~$H$ (since $B$ is a biconnected component). Hence $b$ is not a common neighbour of $u$ and $v$.
			Also, $v$ and $u$ have no other common neighbours since $v$ is not in the biconnected component containing $b$ and $d$.
			\item If $u=m_1$ then the only neighbour of $u$ and $v$ is $b$ since $v$ is not in the biconnected component containing $m_1$ and $b$. 
		\end{itemize}
		Since $L_A \cup R_A\subseteq \NH{a}$ and $m_1\in R_A$ it follows that $\{u\in L_A \cup R_A \mid \enspace \abs{\NH{u} \cap \NH{v}} \text{ is odd}\}=\{m_1\}$ which has odd cardinality, as required.
	\end{myitemize} 
	
\end{proof}

\subsection{Good Stops}

\begin{lem}\label{lem:FDGoodStop}
	Let $B$ be a biconnected component of a graph $H$.
	Suppose that, for an \emph{even} non-negative integer $k$, $B$ is a graph of the form $FD_k$ with vertices as given in Figure~\ref{fig:bipbiconnchordalComponents}. If $(\{a\},b)$ is a good stop in $B$ then  $H$ has a hardness gadget.
\end{lem}

\begin{proof}
	By the definition of a good stop, $R_B = \Gamma_H(b) \setminus \{a,c\}$ has odd cardinality
	and
	there is a gadget $(J_B, z_B)$ such that $\{ v \in V(H) \mid \>\abs{\hom{(J_B, z_B)}{(H,v)}}\text{ is odd}\,\}=\{a\} \cup R_B$.
	
	We give a hardness gadget $(I,S,(J_1, y), (J_2,z), (J_3,y, z))$ for $H$ as follows:
	\begin{myitemize}
		\item $I=\{a\}$ and $S=\{b\}$.
		\item $J_1$ is the gadget $J_B$ with $y=z_B$ so $\Omega_{y} = \{a\} \cup R_B$, which has even cardinality, as required.
		\item $J_2$ is the graph where $z$ is adjacent to an $a$-pin and a $c$-pin so $\Omega_{z} = \{b,d, z_1, \dots, z_k\}$, which has even cardinality, as required (since $k$ is even). 
		\item $J_3$ is an edge between $y$ and $z$. 
	\end{myitemize}
	Note that $a$ is adjacent to every vertex in $\Omega_{z}$, and $b$ is adjacent to every vertex in $\Omega_{y}$, as required. 
	Since $\Omega_y \setminus I = R_B = \Gamma_H(b) \setminus \{a,c\}$
	and  $\Omega_{z}\setminus S=\{d, z_1, \dots, z_k\}$
	and $B$ is a biconnected component,
	there is no edge from $\Omega_y \setminus I$ to $\Omega_z \setminus S$,  as required.
\end{proof}

\begin{lem}\label{lem:BDGoodStop}
	Let $B$ be a biconnected component of a graph $H$.
	Suppose that, for a non-negative integer $k$, $B$ is a graph of the form $BD_k$ with vertices as given in Figure~\ref{fig:bipbiconnchordalComponents}.  If $(\{a\},b)$ is a good stop in $B$ then  $H$ has a hardness gadget.
\end{lem}
\begin{proof}
	By the definition of a good stop, 
	$R_B= \NH{b}\setminus V(B)$ has odd cardinality and 
	there is a gadget $(J_B, z_B)$ such that $\{ v \in V(H) \mid \>\abs{\hom{(J_B, z_B)}{(H,v)}}\text{ is odd}\,\}= \{a\} \cup R_B$.
	We give a hardness gadget $(I,S,(J_1, y), (J_2,z), (J_3,y,z))$ for $H$ as follows:
	\begin{myitemize}
		\item $I=\{a\} $ and $S=\{b\}$.
		\item $J_1$ is the gadget $J_B$ with $y=z_B$ so $\Omega_{y} = \{a\} \cup R_B$, which has even cardinality, as required.
		\item $J_2$ is the graph where $z$ is adjacent to an $a$-pin and a $c$-pin so 
		$\Omega_z = \{b,d\}$, which has
		even cardinality, as required.
		\item $J_3$ is an edge between $y$ and $z$. 
	\end{myitemize}
	Note that $a$ is adjacent to every vertex in $\Omega_{z}$, and $b$ is adjacent to every vertex in $\Omega_{y}$, as required. Since $R_B= \NH{b}\setminus V(B)$ and $B$ is a biconnected component, note that there are no edges between $\Omega_{y}\setminus I=R_B$ and $\Omega_{z}\setminus S=\{d\}$, as required.
\end{proof}

\begin{lem}\label{lem:SklGoodStop}
	Let $B$ be a biconnected component of a graph $H$.
	Suppose that $B$ is an impasse (Definition~\ref{def:impasse}) and that $(a,b)$ is a pair of connectors of $B$. Let $m_1$ be the unique common neighbour of $a$ and $b$ in $H$ (see Figure~\ref{fig:bipbiconnchordalComponents}). Suppose further that $\deg_H(m_1)$ is odd. If $(\{m_1\},b)$ is a good stop in $B$ then  $H$ has a hardness gadget.
\end{lem}
\begin{proof}
	By the definition of a good stop, $R_B= \NH{b}\setminus V(B)$
	has odd cardinality and 
	there is a gadget $(J_B, z_B)$ such that $\{ v \in V(H) \mid \>\abs{\hom{(J_B, z_B)}{(H,v)}}\text{ is odd}\,\}=\{m_1\} \cup R_B$, 
	Using the vertex labels from Figure~\ref{fig:bipbiconnchordalComponents}, we give a hardness gadget $(I,S,(J_1, y), (J_2,z), (J_3,y, z))$ for $H$ as follows:
	\begin{myitemize}
		\item $I=\{m_1\}$ and $S=\{m_1\}$.
		\item $J_1$ is the gadget $J_B$ with $y=z_B$ so $\Omega_{y} = \{m_1\} \cup R_B$, which has even cardinality, as required.
		\item $J_2$ is the graph where $z$ is adjacent to an $a$-pin and an $m_2$-pin so $\Omega_{z} = \{m_1,d\}$, which has even cardinality, as required. 
		\item $J_3$ is a $2$-path between $y$ and $z$. 
	\end{myitemize}
	There are an odd number of $2$-walks from $m_1$ to itself since $\deg_H(m_1)$ is odd by assumption. 
	There are an odd number of $2$-walks from $m_1$ to $d$ since $k$ is odd and no pair of vertices of $S_{k, \ell}$ has common neighbours outside of $S_{k,\ell}$. Since $R_B= \NH{b}\setminus V(B)$ and $B$ is biconnected there is exactly one $2$-walk from $m_1$ to each vertex in $R_B$. Thus, for $s\in S=\{m_1\}$, $i\in I=\{m_1\}$, $o\in \Omega_{y}\setminus I=R_B$, $x\in \Omega_{z}\setminus S=\{d\}$, we have shown that $\abs{\Sigma_{i,s}}$, $\abs{\Sigma_{o,s}}$ and $\abs{\Sigma_{i,x}}$ are odd, as required.
	Finally, since $B$ is a biconnected component there are no $2$-walks from $d$ to a vertex in $R_B$ and therefore $\abs{\Sigma_{o,x}}$ is even, as required.
\end{proof}

\subsection{Hardness Results}

In this section we establish hardness results which are used to prove Lemma~\ref{lem:CaterpillarForBipChordalComponents} in Section~\ref{sec:chordalbipartitesequence}.

\begin{lem}\label{lem:oddFDTrivialHardness}
	Let $H$ be a graph and let $B$ be a biconnected component of $H$.
	Suppose that, for an \emph{odd} non-negative integer $k$, $B$ is a graph of the form $FD_k$ with vertex labels as given in Figure~\ref{fig:bipbiconnchordalComponents}. 
	If $\deg_H(a)$ is even then $H$ has a hardness gadget.
\end{lem}
\begin{proof}
	We give a hardness gadget $(I,S,(J_1, y), (J_2,z), (J_3,y,z))$ for $H$ as follows:
	\begin{myitemize}
		\item $I=\{a\}$ and $S=\{b,d, z_1, \dots z_k\}$.
		\item $J_1$ is the graph where $y$ is adjacent to a $b$-pin and a $d$-pin so $\Omega_{y} = \{a,c\}$, which has even cardinality, as required. 
		\item $J_2$ is the graph where $z$ is adjacent to an $a$-pin so $\Omega_{z} = \NH{a}$, which has even cardinality, as required. 
		\item $J_3$ is an edge between $y$ and $z$. 
	\end{myitemize} 
	Note that $a$ is adjacent to every vertex in $\Omega_{z}$, and each vertex of $S$ is adjacent to every vertex in $\Omega_{y}$, as required. Since $B$ is a biconnected component there are no edges between $\Omega_{y}\setminus I=\{c\}$ and $\Omega_{z}\setminus S=\NH{a}\setminus V(B)$, as required.
\end{proof}

\begin{lem}\label{lem:BDFDHardness}
	Let $H$ be a graph and let $A$ and $B$ be biconnected components of $H$.
	Suppose that, for \emph{odd} integers $k\ge 1$ and $\ell\ge 1$, there is an isomorphism $f$ from the graph $BD_k$ to $A$ and an isomorphism $g$ from the graph $FD_\ell$ to $B$. Suppose that there is a vertex $w=f(b)=g(a)$ such that  $\deg_H(w)$ is odd. Then $H$ has a hardness gadget.
\end{lem}
\begin{proof}
	We give a hardness gadget $(I,S,(J_1, y), (J_2,z), (J_3,y, z))$ for $H$ as follows:
	\begin{myitemize}
		\item $I=\{w\}$ and $S=\{w\}$.
		\item $J_1$ is the graph where $y$ is adjacent to an $f(a)$-pin and an $f(c)$-pin so $\Omega_{y} = \{f(d),f(b)\}=\{f(d),w\}$, which has even cardinality, as required. 
		\item $J_2$ is the graph where $z$ is adjacent to a $g(b)$-pin and an $g(d)$-pin so $\Omega_{z} = \{g(c),g(a)\}=\{g(c),w\}$, which has even cardinality, as required. 
		\item $J_3$ is a $2$-path between $y$ and $z$. 
	\end{myitemize} 
	By the fact that $A$ and $B$ are biconnected components, there are exactly $k+2$ walks of length $2$ from $f(d)$ to $w$, and there are exactly $\ell+2$ walks of length $2$ from $g(c)$ to $w$, where $k$ and $\ell$ are odd. Since $\deg_H(w)$ is odd, there is a an odd number of 
	length-$2$ walks  from $w$ to itself. Finally, there are no length-$2$ walks   from $f(d)$ to $g(c)$, as required.
\end{proof}

\begin{lem}\label{lem:BDList}
	Let $H$ be a graph and let $B$ be a biconnected component of $H$ that is of the form $BD_k$ for some integer $k\ge 0$. Using the vertex names from Figure~\ref{fig:bipbiconnchordalComponents}, there is a gadget $(J,z)$ such that $\{ v \in V(H) \mid \>\abs{\hom{(J, z)}{(H,v)}}\text{ is odd}\,\}=\NH{b}\setminus V(B)$.
\end{lem}
\begin{proof}
	The graph $J$ has three pinned vertices ---  an $a$-pin, a $b$-pin, and a $c$-pin.
	The $b$-pin is adjacent to the vertex~$z$ and the other two pins are attached to~$z$ by paths of length~$2$.
	
	We will now consider each $v\in V(H)$ to determine whether $\abs{\hom{(J, z)}{(H,v)}}$ is odd.
	Since $z$ is adjacent to a $b$-pin in~$J$, this can only be true for $v\in \NH{b}$.
	
	First, consider a vertex $v\in \NH{b}\cap V(B)$.
	\begin{itemize}
		\item If $v\in \{a,y_1,\ldots,y_k\}$ then $v$ has exactly two length-2 walks to~$c$, so  $\abs{\hom{(J, z)}{(H,v)}}$ is even.
		\item If $v=c$ then $v$ has exactly two length-2 walks to $a$ so  $\abs{\hom{(J, z)}{(H,v)}}$ is even.
	\end{itemize}
	
	Finally, consider a vertex $v\in \NH{b} \setminus V(B)$.
	There is exactly one $2$-walk to~$a$, and exactly one $2$-walk to~$c$, so  $\abs{\hom{(J, z)}{(H,v)}}$ is odd.
\end{proof}

The following lemma is essentially the same as Lemma~\ref{lem:BDList}.

\begin{lem}\label{lem:FDList}
	Let $H$ be a graph and let $B$ be a biconnected component of $H$ that is of the form $FD_k$ for some integer $k\ge 0$. Using the vertex names from Figure~\ref{fig:bipbiconnchordalComponents}, there is a gadget $(J,z)$ such that $\{ v \in V(H) \mid \>\abs{\hom{(J, z)}{(H,v)}}\text{ is odd}\,\}=\NH{a}\setminus V(B)$.
\end{lem}
\begin{proof} 
	The graph $J$ has three pinned vertices ---  an $a$-pin, a $b$-pin, and a $d$-pin.
	The $a$-pin is adjacent to the vertex~$z$ and the other two pins are attached to~$z$ by paths of length~$2$.

	We will now consider each $v\in V(H)$ to determine whether $\abs{\hom{(J, z)}{(H,v)}}$ is odd.
	Since $z$ is adjacent to an $a$-pin in~$J$, this can only be true for $v\in \NH{a}$.
	
	First, consider a vertex $v\in \NH{a}\cap V(B)$. 
	\begin{itemize}
		\item If $v\in \{d,z_1,\ldots,z_k\}$ then $v$ has exactly two length-2 walks to~$b$, so  $\abs{\hom{(J, z)}{(H,v)}}$ is even.
		\item If $v=b$ then $v$ has exactly two length-2 walks to $d$ so  $\abs{\hom{(J, z)}{(H,v)}}$ is even.
	\end{itemize}
	
	Finally, consider a vertex $v\in \NH{b} \setminus V(B)$.
	There is exactly one $2$-walk to~$b$, and exactly one $2$-walk to~$d$, so  $\abs{\hom{(J, z)}{(H,v)}}$ is odd.
\end{proof}

We obtain the following lemma, which is a generalisation of
\cite[Lemma 4.5]{squarefree}.

\begin{lem}\label{lem:SquarefreePathHardness}
	For an integer $q\ge 1$, let $P=v_0, \dots, v_q$ be a path in a graph $H$. Suppose that no edge of $P$ is part of a square in $H$ and that $\deg_H(v_j)$ is odd for all $j\in [q-1]$.
	Suppose that
	{
		\renewcommand{\theenumi}{1\,\alph{enumi}}
		\renewcommand{\labelenumi}{\theenumi)}
		\begin{myenumerate}
			\item $\deg_H(v_0)$ is even, or
			\item $\deg_H(v_0)$ is odd and there is a biconnected component $B_0$ that is isomorphic to $BD_k$ for some \emph{odd} integer $k\ge 1$, where the isomorphism maps $v_0$ to the vertex $b$ from Figure~\ref{fig:bipbiconnchordalComponents}.
		\end{myenumerate}
	}
	\noindent Suppose further that
	{
		\renewcommand{\theenumi}{2\,\alph{enumi}}
		\renewcommand{\labelenumi}{\theenumi)}
		\begin{myenumerate}
			\item $\deg_H(v_q)$ is even, or
			\item $\deg_H(v_q)$ is odd and there is a biconnected component $B_{q+1}$ that is isomorphic to $FD_k$ for some \emph{odd} integer $k\ge 1$, where the isomorphism maps $v_q$ to the vertex $a$ from Figure~\ref{fig:bipbiconnchordalComponents}.
		\end{myenumerate}
	}
	\noindent Then $H$ has a hardness gadget.
\end{lem}
\begin{proof}
	We give a hardness gadget $(I,S,(J_1, y), (J_2,z), (J_3,y, z))$ for $H$ as follows:
	\begin{myitemize}
		\item $I=\{v_{1}\}$ and $S=\{v_{q-1}\}$.
		\item If 1\,a) holds, then $J_1$ is the graph where $y$ is adjacent to a $v_0$-pin so $\Omega_{y} = \NH{v_0}$, which has even cardinality as required. If 1\,b) holds then $(J_1, y)$ is the gadget from Lemma~\ref{lem:BDList} and $\Omega_{y} = \NH{v_0}\setminus V(B_0)$, which has even cardinality as required.
		The vertex~$v_1$ is in $\Omega_y$ because the edge $\{v_0,v_1\}$ is not part of a square in~$H$.		
		
		\item If 2\,a) holds, then $J_2$ is the graph where $z$ is adjacent to a $v_q$-pin so $\Omega_{z} = \NH{v_q}$, which has even cardinality as required. If 2\,b) holds then $(J_2, z)$ is the gadget from Lemma~\ref{lem:FDList} and $\Omega_{z} = \NH{v_q}\setminus V(B_{q+1})$, which has even cardinality as required.
		The vertex~$v_{q-1}$ is in $\Omega_z$ because the edge $\{v_{q-1},v_q\}$ is not part of a square in~$H$.		
		\item $J_3$ is the path gadget $J_P$.
	\end{myitemize} 
	This is a hardness gadget by Lemma~\ref{lem:SquarefreePathPrehardness}.
\end{proof}

\subsection{Chordal Bipartite Sequence Lemma} \label{sec:chordalbipartitesequence}

\begin{lem}[Chordal Bipartite Sequence Lemma]\label{lem:CaterpillarForBipChordalComponents}
	For an integer $q\ge 1$, let $B_1, \dots, B_q$ be biconnected components of a graph $H$ and let $b_0, \dots, b_q$ be vertices such that, for all $i\in [q]$, $b_{i-1}$ and $b_i$ are distinct vertices of $B_i$, and $B_i$ satisfies one of the following:
	\begin{myitemize}
		\item $B_i$ is an edge from $b_{i-1}$ to $b_i$, 
		\item $B_i$ is a diamond in which $\{b_{i-1}, b_i\}$ is an edge, or
		\item $B_i$ is an impasse, where $(b_{i-1},b_i)$ is a pair of connectors of $B_i$. In this case, let $d_i$ be the unique
		common neighbour of~$b_{i-1}$ and~$b_i$ in~$H$. 
	\end{myitemize}
	If  $\abs{\NH{b_0}\setminus V(B_1)}$ is odd, then at least one of the following holds:
	\begin{myitemize}
		\item $B_q$ is an edge or a diamond and $(\{b_{q-1}\},b_q)$ is a good start in~$B_q$ but not a good stop in~$B_q$, 
		\item $B_q$ is an impasse 
		and $(\{d_q\}, b_q)$ is a good start in $B_q$ but not a good stop in~$B_q$, or
		\item $H$ has a hardness gadget.
	\end{myitemize}
\end{lem}
\begin{proof}

	We start by collecting some facts that we will need.   
	
	\noindent  {\bf Fact 1.} \textit{If $i\in [q]$ and  $B_i$ is a diamond, then at least one of the following holds:
		\begin{myitemize}
			\item for some non-negative integer~$k$ there is an isomorphism from $FD_k$ to~$B_i$, mapping the vertex~$a$ 
			from 
			Figure~\ref{fig:bipbiconnchordalComponents}
			to~$b_{i-1}$ and the vertex~$b$ to~$b_i$
			(we refer to this situation below by saying ``$B_i$ is of the form $FD_k$''), or
			\item for some non-negative integer~$k$ there is an isomorphism from $BD_k$ to~$B_i$, mapping the vertex~$a$
			from 
			Figure~\ref{fig:bipbiconnchordalComponents}
			to~$b_{i-1}$ and vertex~$b$ to~$b_i$. (We refer to this situation as ``$B_i$ is of the form $BD_k$'').
	\end{myitemize}}
	
	\noindent {\bf Fact 2.} \textit{If $B_1$ is an edge 
		or a biconnected component of the form $FD_k$ for an odd integer~$k$
		then
		$\NH{b_0}$ is even.}
	(This is because $b_0$ has and odd number of neighbours in~$B_1$ and an odd number outside of~$B_1$, by assumption.)
	\medskip

	Let $L_0=\NH{b_0}\setminus V(B_1)$ and let $B_0$ be the subgraph of~$H$ induced by the vertices in~$L_0\cup \{b_0\}$. 
	For every $i\in[q]$ such that $B_i$ is an edge or a diamond, let $L_i = \{b_{i-1}\}$. For every $i\in[q]$ such that $B_i$
	is an impasse, let $L_i = \{d_i\}$.
	For every $i\in \{0,\ldots,q\}$, let $R_i=\NH{b_i}\setminus V(B_i)$.
	
	We start by considering $i\in\{0,1\}$ and showing  that $(L_i,b_i)$ is a good start in~$B_i$ or that $H$ has a hardness gadget. 
	We first deal with the easy case $i=0$. $\abs{L_0}$ is odd by assumption 
	so, to show that $(L_0,b_0)$ is a good start in~$B_0$,
	it suffices to use the gadget $(J,z)$ in which $z$ is adjacent to a $b_0$-pin.
	Note that   $R_0= \NH{b_0}\cap V(B_1)$.
	We next deal with $i=1$ by considering four cases depending on~$B_1$:
	\begin{myitemize}
		\item If $B_1$ is an edge from $b_0$ to $b_1$ then, by Lemma~\ref{lem:EdgeGoodStart}, $(L_1, b_1)$ is a good start in $B_1$.
		
		\item If $B_1$ is a diamond of the form $FD_k$ for an  {odd} integer $k\ge0$ 
		then    $H$ has a hardness gadget by Fact~2 and Lemma~\ref{lem:oddFDTrivialHardness}.
		
		\item If $B_1$ is a diamond of the form $FD_k$ for an  {even} integer $k\ge0$ 
		then    $R_0$ has even cardinality. 
		Therefore $(L_0, b_0)$ is not a good stop in $B_0$ and we can apply Lemma~\ref{lem:FDGoodStart} 
		to show that  $(L_1, b_1)$ is a good start in $B_1$.
		
		\item If $B_1$ is a diamond   of the form $BD_k$ 
		then    $R_0 $ has even cardinality. 
		Therefore $(L_0, b_0)$ is not a good stop in $B_0$ and we can apply Lemma~\ref{lem:BDGoodStart} 
		to show that  $(L_1, b_1)$ is a good start in $B_1$.

		\item If $B_1$ is an impasse where $(b_{0},b_1)$ is a pair of connectors. Then $b_0$ has $2$ neighbours in $B_1$ and hence  $R_0 $ has even cardinality. Therefore $(L_0, b_0)$ is not a good stop in $B_0$. 
		Recall that
		$d_1$ be the unique
		common neighbour of~$b_{0}$ and~$b_1$ in~$H$. 
		If $\deg_H(d_1)$ is even then $H$ has a hardness gadget by Lemma~\ref{lem:hard_impasses}. Otherwise  Lemma~\ref{lem:SklGoodStart}   shows that $(L_1, b_1)$ is a good start in $B_1$.
	\end{myitemize}
	
	For the rest of the proof, let $j$ be the smallest index in $[q]$ that satisfies one of the following properties:   
	\begin{enumerate}[(P1)]
		\item \label{prop1}\label{equ:CaterpillarBipChordal1}	
		$|R_j|$ is odd and
		there is no odd integer~$k$ such that
		$B_j$ is   of the form $FD_k$.
		\item \label{prop2}\label{equ:CaterpillarBipChordal2}
		There is an odd integer~$k$ such that
		$B_j$ is of the form $FD_k$.
		\item \label{prop3} \label{equ:CaterpillarBipChordal3}
		$j=q$ and
		$|R_j|$ is even and there is no odd integer $k$ such that $B_j$ is of the form $FD_k$.
	\end{enumerate}
	
	We will use the following claims.

	\medskip
	\noindent \textbf{Claim A} \textit{Suppose that $j$ does not satisfy (P\ref{equ:CaterpillarBipChordal2}). 
		Then $H$ has a hardness gadget or, for all $\ell \in  [j]$, the following are satisfied.
		\begin{myitemize}
			\item
			$(L_\ell, b_\ell)$ is a good start in $B_\ell$, and
			\item If $\ell>1$ then
			$(L_{\ell-1},b_{\ell-1})$ is not a good stop in $B_{\ell-1}$.
		\end{myitemize}
	}
	\medskip
	
	\begin{claimproof}
		The proof of  Claim~A is by induction on~$\ell$. We have already established the base case $\ell=1$.
		Now  fix $\ell\in \{2,\ldots,j\}$ and suppose (from the inductive hypothesis) that $(L_{\ell-1}, b_{\ell-1})$ is a good start in $B_{\ell-1}$. By the minimality of $j$, 
		$B_{\ell-1}$ is not of the form $FD_k$ for an odd integer $k$ (otherwise $\ell-1$ would satisfy~(P\ref{equ:CaterpillarBipChordal2})).
		Again, by minimality of~$j$, $|R_{\ell-1}|$  is even
		(otherwise $\ell-1$ would satisfy~(P\ref{equ:CaterpillarBipChordal1})).
		By the definition of good stop,  $(L_{\ell-1}, b_{\ell-1})$ is not a good stop in $B_{\ell-1}$. 
		Since $j$ does not satisfy~(P\ref{equ:CaterpillarBipChordal2}),
		$B_\ell$ is not of the form $FD_k$ for an {odd} integer $k$. Thus, we can apply one of Lemmas~\ref{lem:EdgeGoodStart},~\ref{lem:FDGoodStart},~\ref{lem:BDGoodStart} or~\ref{lem:SklGoodStart} depending on the form of $B_\ell$ to show that $(L_\ell, b_\ell)$ is a good start in $B_\ell$. This completes the proof of Claim~A.\end{claimproof}

	\medskip 
	\noindent \textbf{Claim B} \textit{Suppose that 
		$B_j$ satisfies one of the following.
		\begin{enumerate}[(B1)]
			\item \label{B1} $B_j$ is the edge $\{b_{j-1},b_j\}$ and $\deg_H(b_j)$ is even, or
			\item \label{B2} there is an odd integer~$k$ such that $B_j$ is of the form $FD_k$.
		\end{enumerate}
		Suppose that there is an integer $\ell$ in the range $1\leq \ell\leq j$ such that, 
		for $i\in \{\ell,\ldots,j-1\}$, $B_i$ is the edge $\{b_{i-1},b_i\}$
		and for $i\in \{\ell,\ldots,j\}$, 
		$\NH{b_{i-1}}$   and
		$\Gamma_H(b_{i-1}) \cap V(B_{i-1})$ have odd cardinality.  
		Then $H$ has a hardness gadget or
		there is an integer $p$ in the range $1\leq p \leq \ell-1$ and an odd integer~$k'$ 
		such that 
		$B_{p}$ is  of the form $BD_{k'}$. Also, for $i\in\{p+1,\ldots,j-1\}$, $B_i$ is  the edge  
		$\{b_{i-1},b_i\}$ 
		where  $\NH{b_{i-1}}$   has odd cardinality.}
	\medskip
	
	\begin{claimproof}
		The proof of Claim~B is by induction on~$\ell$. The base case $\ell=1$ is vacuous ---   taking $i=1$, the precondition of the claim ensures that $|\Gamma_H(b_0)|$  is odd, contrary to Fact~2.
		So  consider some $\ell>1$ for which we wish to prove the claim.  
		Since taking $i=\ell$ guarantees that
		$|\Gamma_H(b_{\ell-1}) \cap V(B_{\ell-1})|$ is odd, 
		$B_{\ell-1}$ is either  
		an edge or it is of the form $BD_{k'}$ for an odd integer $k'$. We consider each case.
		
		\begin{myitemize}
			\item $B_{\ell-1}$ is an edge:
			If $\deg_H(b_{\ell-2})$ is even  
			$H$ has a hardness gadget by Lemma~\ref{lem:SquarefreePathHardness}
			(take   $v_0,\ldots,v_q = b_{\ell-2},\ldots, b_j$ in Case (B\ref{B1}) and 
			$v_0,\ldots,v_q = b_{\ell-2},\ldots, b_{j-1}$ in Case (B\ref{B2}). 
			Thus,   assume that $\deg_H(b_{\ell-2})$ is odd.
			By Fact~2, $\ell-2\ge 1$ and consequently (by Claim~A),  $(L_{\ell-2},b_{\ell-2})$ is a good start in $B_{\ell-2}$ that is not 
			a good stop in~$B_{\ell-2}$. This implies that   $|R_{\ell-2}|$  is even, which together with the fact that $\deg_H(b_{\ell-2})$ is odd implies that $b_{\ell-2}$ has an odd number of neighbours in $B_{\ell-2}$.  
			So the preconditions of the claim are met with $i=\ell-1$ and we can finish by induction.

			\item $B_{\ell-1}$ is of the form $BD_{k'}$ for an odd integer $k'$. 
			The claim follows  by taking $p=\ell-1$.
			
		\end{myitemize} 
		This concludes the proof of Claim~B.\end{claimproof}

	We now make a case distinction, depending on which property~$j$ satisfies.
	\begin{description}

		\item[Case~(P\ref{equ:CaterpillarBipChordal1}).] 
		We will show that $H$ has a hardness gadget.
		
		By Claim~A, either  $H$ has a hardness gadget (in which case we are finished) or $(L_{j}, b_{j})$ is a good start in $B_{j}$. Since $|R_j|$
		is odd, $(L_{j}, b_{j})$ is a good stop in $B_{j}$. 
		We now distinguish several cases, depending on the form of~$B_j$. 
		\begin{itemize}
			\item
			If $B_j$ is of the form $FD_k$ for even $k$, or of the form $BD_k$, then  $H$ has a hardness gadget by Lemmas~\ref{lem:FDGoodStop} or~\ref{lem:BDGoodStop}, respectively. 
			
			\item If $B_j$ is an impasse, then depending on the degree of $d_j$,  $H$ has a hardness gadget either by Lemma~\ref{lem:hard_impasses} (if the degree of $d_j$ is even) or by Lemma~\ref{lem:SklGoodStop} (if the degree of $d_j$ is odd).
			
			\item  Finally, suppose that $B_j$ is an edge. 
			We will use Claim~B with $\ell=j$ to show that $H$ has a hardness gadget.
			The first step is to show that (unless $H$ has a hardness gadget) the preconditions of the
			claim are met  --- that is  
			$\deg_H(b_j)$ is even,
			$\deg_H(b_{j-1})$ is odd,  and
			$|\Gamma_H(b_{j-1}) \cap V(B_{j-1})|$  is odd. 
			
			By (P\ref{prop1}),  $|R_j| $ is odd. Since $b_j$ has only one neighbour in $B_j$,
			$\deg_H(b_j)$ is even. If $\deg_H(b_{j-1})$ is even,  $H$ has a  hardness gadget by Lemma~\ref{lem:SquarefreePathHardness}
			(taking $q=1$, $v_0 = b_{j-1}$ and $v_1 = b_j$). 
			From now on, we   assume that $\deg_H(b_{j-1})$ is odd. 
			By Fact~2, $j-1\geq 1$.
			By the minimality of $j$, $|R_{j-1}| $  is even, 
			which implies that $b_{j-1}$ has an odd number of neighbours in $B_{j-1}$.

			Applying Claim~B with $\ell=j$, 
			either $H$ has a hardness gadget. or
			there is an integer $p$ in the range $1\leq p \leq j-1$ and an odd integer~$k'$ 
			such that 
			$B_{p}$ is  of the form $BD_{k'}$. 
			Also, for $i\in\{p+1,\ldots,j-1\}$, $B_i$ is  the edge  
			$\{b_{i-1},b_i\}$ 
			where  $\NH{b_{i-1}}$   has odd cardinality.

			Now we apply Lemma~\ref{lem:SquarefreePathHardness}  
			with the path $v_0,\ldots,v_q$ equal to $b_p,\ldots,b_j$.
			The degrees of $v_0,\ldots,v_{q-1}$ are odd  and the degree of $v_q$ is even.
			$v_0$ is in the biconnected component $B_p$.	This shows that $H$ has a hardness gadget.
			
		\end{itemize}

		\item[Case~(P\ref{equ:CaterpillarBipChordal2}).] 
		We will use Claim~B with $\ell=j$ to show that $H$ has a hardness gadget.
		 
		The first step is to show that (unless $H$ has a hardness gadget) the preconditions of the
		claim are met  --- that is   
		$\deg_H(b_{j-1})$ is odd,  and
		$|\Gamma_H(b_{j-1}) \cap V(B_{j-1})|$  is odd.

		If $\deg_H(b_{j-1})$ is even then  $H$ has a hardness gadget by Lemma~\ref{lem:oddFDTrivialHardness}. 
		From now on, we   assume that $\deg_H(b_{j-1})$ is odd. 
		By Fact~2, $j-1\geq 1$.
		By the minimality of $j$, $|R_{j-1}| $  is even, 
		which implies that $b_{j-1}$ has an odd number of neighbours in $B_{j-1}$.
		
		Applying Claim~B with $\ell=j$, 
		either $H$ has a hardness gadget. or
		there is an integer $p$ in the range $1\leq p \leq j-1$ and an odd integer~$k'$ 
		such that 
		$B_{p}$ is  of the form $BD_{k'}$. 
		Also, for $i\in\{p+1,\ldots,j-1\}$, $B_i$ is  the edge  
		$\{b_{i-1},b_i\}$ 
		where  $\NH{b_{i-1}}$   has odd cardinality.

		If $p=j-1$  then $H$ has a hardness gadget by Lemma~\ref{lem:BDFDHardness}.
		Otherwise, we apply Lemma~\ref{lem:SquarefreePathHardness}  
		with the path $v_0,\ldots,v_q$ equal to $b_{p},\ldots,b_{j-1}$.
		The degrees of $v_0,\ldots,v_{q}$ are odd.
		$v_0$ is in the biconnected component $B_p$
		and $b_q$ is in the biconnected component $B_j$.
		This shows that $H$ has a hardness gadget.

		\item[Case~(P\ref{equ:CaterpillarBipChordal3})] 
		By Claim~A,  $H$ has a hardness gadget or $(L_q, b_q)$ is a good start in $B_q$.
		In the latter case, since $|R_q|$ is even, $(L_q,b_q)$ is not a good stop in~$B_q$.

	\end{description}

\end{proof}

\section{$K_4$-minor-free Components}\label{sec:K4freeGraphs}
In this section, we establish a structural classification for biconnected $K_4$-minor-free graphs. Recall that such a classification has already been achieved for biconnected \emph{chordal bipartite} graphs in Section~\ref{sec:chordalBipartiteComps} (Lemma~\ref{lem:X}). For this reason, we focus in what follows on biconnected $K_4$-minor-free graphs that are not chordal bipartite, which is equivalent to 
focusing on biconnected $K_4$-minor-free graphs 
which have  an induced cycle  that is not a square. The main result of this section is presented in Lemma~\ref{lem:main_biconnected}. Informally, it states that every biconnected $K_4$-minor-free graph $B$ either induces a hardness gadget for every $K_4$-minor-free graph containing $B$ as a biconnected component, or $B$ is an edge, a diamond, an impasse, or a so-called obstruction (Definition~\ref{def:obstruction}). Together with our insights from the previous section, obstructions will be the final building block in our construction of global hardness gadgets in Section~\ref{sec:k4MinorFreeMain}. 

\begin{defn}[separation, separator]	\label{def:separator}
	Let $G$ be a graph and let $A$ and $B$ be subsets of $V(G)$.
	The pair $(A,B)$   is a \emph{separation} of   $G$ if $V(G)=A \cup B$ and  $G$ has no edges between $A\setminus B$ and $B \setminus A$. 
	The set $A \cap B$ is called the \emph{separator} of this separation.
\end{defn}

\subsection{Induced Cycles}\label{sec:induced_cycles}

Recall Definition~\ref{def:walkneighbourset}, which defines for a closed walk $W = (w_0,\ldots,w_{q-1},w_0)$  
in a graph~$H$ the walk-neighbour-set $N_{W,H}(w_i)=\Gamma_H(w_{i-1}) \cap \Gamma_H(w_{i+1})$, where the indices are taken modulo~$q$. In this section we will use this notion mainly for cycles.		

\begin{lem}\label{lem:disjoint}
	Let $H$ be a biconnected $K_4$-minor-free graph containing an induced cycle $C=(c_0,\dots,c_{q-1},c_0)$ for some $q\neq4$. Then the walk-neighbour-sets $N_{C,H}(c_0), \dots, N_{C,H}(c_{q-1})$ are pairwise disjoint.
\end{lem}
\begin{proof}
	If $q=3$ then the fact that we do not allow self-loops in $H$ together with the fact that $H$ does not contain $K_4$ as a subgraph ensures that the $N_{C,H}(c_i)$ are pairwise disjoint.
	
	Suppose $q>4$. Assume for contradiction that there exists a vertex $w\in N_{C,H}(c_i) \cap N_{C,H}(c_j)$ for some $i\neq j$. If $w$ is part of the cycle $C$, then we obtain a chord (note that $q>4$), contradicting the fact that $C$ is induced. If $w$ is not part of the cycle $C$, then $w$ is adjacent to at least $3$ vertices of the cycle, yielding a $K_4$-minor.
\end{proof}

\begin{lem}\label{lem:separators} 
	Let $H$ be a biconnected $K_4$-minor-free graph containing an induced cycle $C=(c_0,\dots,c_{q-1},c_0)$.
	If $q>4$ and $|N_{C,H}(c_i)|>1$ for some $i\in \{0, \dots, q-1\}$ then there exists a separation $(A,B)$ of $H$ such that $C\setminus\{c_i\}  \subseteq A$, $N_{C,H}(c_i)\subseteq B$ and $A\cap B =\{c_{i-1},c_{i+1}\}$. Furthermore, $H$ is a (1,2)-supergraph of $H[B]$.
\end{lem}
\begin{proof}
	Let $S_1,\dots, S_k$ be the connected components of the graph obtained from $H$ by deleting $c_{i-1}$, $c_{i+1}$, and all edges incident to $c_{i-1}$ and $c_{i+1}$. Then w.l.o.g.\ we can assume that $c_j\in V(S_1)$ for all $j\notin \{i-1,i,i+1\}$. Set $A= V(S_1)\cup\{c_{i-1},c_{i+1}\}$ and $B= V(S_2)\cup\dots\cup V(S_k)\cup \{c_{i-1},c_{i+1}\}$. By construction, we have $A\cup B=V(H)$ and $A\cap B=\{c_{i-1},c_{i+1}\}$, and there are no edges between $A\setminus B$ and $B\setminus A$.
	We claim that $N_{C,H}(c_i) \cap V(S_1) =\emptyset$ as, otherwise, we obtain the following $K_4$-minor; recall that $|N_{C,H}(c_i)|>1$ and $q>4$:
	
	\medskip
	
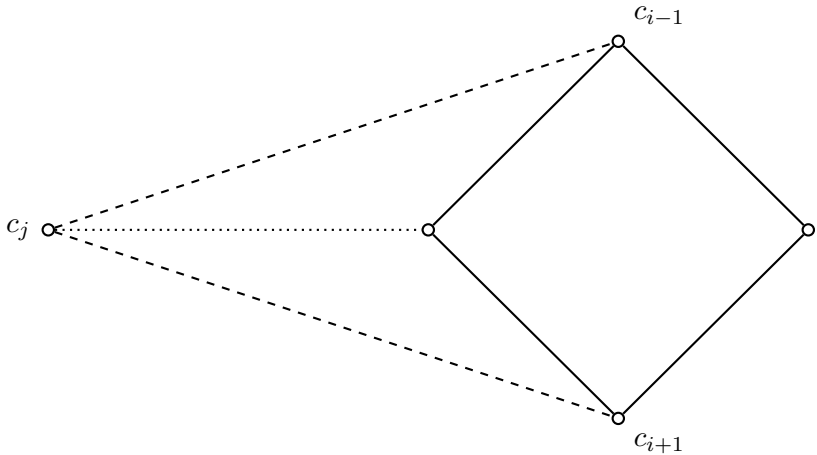
\begin{tikzpicture}[scale=2.5, node distance = 1.4cm,thick]
	\tikzstyle{dot}   =[fill=black, draw=black, circle, inner sep=0.15mm]
	\tikzstyle{vertex}=[  draw=black, circle, inner sep=1.5pt]
	\tikzstyle{dist}  =[fill=white, draw=black, circle, inner sep=2pt]
	\tikzstyle{pinned}=[draw=black, minimum size=10mm, circle, inner sep=0pt]	
	\begin{scope}[xshift=2cm] 
	\node[vertex] (cim1) at (3  ,1) [label=45: $c_{i-1}$] {};
	\node[vertex] (cip1) at (3  ,-1) [label=-45: $c_{i+1}$] {};
	\node[vertex] (ci1) at (2  ,0) {};
	\node[vertex] (ci2) at (4  ,0) {};
	\node[vertex] (v) at (0  ,0) [label=180: $c_{j}$]{};
	\draw[dotted] (v) -- (ci1); \draw[dashed] (v) -- (cim1);
	\draw[dashed] (v) -- (cip1); \draw (cip1) -- (ci1);
	\draw (cip1) -- (ci2); \draw (cim1) -- (ci1); \draw (cim1) -- (ci2);
	\end{scope}
	\end{tikzpicture} 
	\medskip
	
	\noindent Here, the dashed lines depict the path $c_{i-1},\dots,c_j,\dots,c_{i+1}$ which is $C\setminus\{c_i\}$. Further, $c_j$ is a vertex of $C$ satisfying that there exists a (shortest) path $P$ from a vertex in $N_{C,H}(c_i)$ to $c_j$ such that the internal vertices of $P$ are disjoint from $C\cup N_{C,H}(c_i)$. Note that $c_j$ exists if $N_{C,H}(c_i) \cap V(S_1)$ is not empty. We depict $P$ in the above picture with a dotted line. In particular, $P$ has length at least one, i.e., $c_j\notin N_{C,H}(c_i)$ by Lemma~\ref{lem:disjoint}. Hence we obtain indeed a $K_4$-minor.
	Consequently, no vertex of $N_{C,H}(c_i)$ is contained in $A$, and thus $N_{C,H}(c_i)\subseteq B$.

	It remains to show that $H$ is a (1,2)-supergraph of $H[B]$: It is immediate that an edge $e$ between two vertices in $B$ is present in $H$ if and only if it is present in $H[B]$. By the definition of $B$, $H$ and $H[B]$ cannot have a different number of $2$-paths between two different vertices $b_1$ and $b_2$ in $B$, unless $\{b_1,b_2\}=\{c_{i-1},c_{i+1}\}$. However, regarding the latter case, all common neighbours of $c_{i-1}$ and $c_{i+1}$ are contained in $N_{C,H}(c_i) \subseteq B$ and thus the claim also holds for those two vertices. 
\end{proof}

\begin{cor}\label{cor:diamonds_disjoint}
	Let $H$ be a biconnected $K_4$-minor-free graph containing an induced cycle $C=(c_0,\dots,c_{q-1},c_0)$.
	If $q>4$ then, for all $i\in \{0, \dots, q-1\}$, we have that at least one of $N_{C,H}(c_i)$ and $N_{C,H}(c_{i+1})$ has cardinality~$1$.
\end{cor}
\begin{proof}
	Assume for contradiction that for some $i$, both, $N_{C,H}(c_i)$ and $N_{C,H}(c_{i+1})$, have cardinality greater than $1$. We invoke Lemma~\ref{lem:separators} for $C$ and $i$, which yields a separation $(A,B)$ of $H$ such that $C\setminus\{c_i\}\subseteq A$, $N_{C,H}(c_i)\subseteq B$, and $A\cap B=\{c_{i-1},c_{i+1}\}$. However, by assumption, there exists $c'\in N_{C,H}(c_{i+1})\setminus\{c_{i+1}\}$. Note further, that $c'\neq c_{i-1}$ as $q>4$ and $C$ is induced. Thus there is a path connecting $c_i\in B$ and $c_{i+2}\in A$ which does not pass through either one of $c_{i-1}$ and $c_{i+1}$ contradicting the assumption that $(A,B)$ is a separation with $A\cap B=\{c_{i-1},c_{i+1}\}$.
\end{proof}

\begin{cor}\label{cor:diamonds_disjoint*}
	Let $H$ be a biconnected $K_4$-minor-free graph containing an induced cycle $C=(c_0,\dots,c_{q-1},c_0)$.
	If $q\neq 4$ and $H$ does not have a hardness gadget then, for all $i\in \{0, \dots, q-1\}$, we have that at least one of $N_{C,H}(c_i)$ and $N_{C,H}(c_{i+1})$ has cardinality~$1$.	
\end{cor}
\begin{proof}
	If $q>4$ the statement follows from Corollary~\ref{cor:diamonds_disjoint}.
	If $q=3$ then $C$ is a triangle and the
	statement follows directly from Lemma~\ref{lem:good_triangle_gadget}.
\end{proof}

\subsection{Pre-Hardness Gadgets and Obstructions}

\begin{defn}[obstruction]\label{def:obstruction}
	Let $B$ be a $K_4$-minor-free biconnected graph
	and let $C$ be an induced cycle of~$B$ whose length is not~$4$.
	We say that $B$ is an \emph{obstruction with cycle~$C$}
	if  every even-cardinality walk-neighbour-set of~$C$ in~$B$ 
	only contains vertices whose degree in~$B$ is~$2$.
	We say that $B$ is an obstruction if, for some~$C$, it is an obstruction with cycle~$C$.
	We use $\cycles(B)$ to denote $\{C \mid \mbox{$B$ is an obstruction with cycle~$C$}\}$.
\end{defn}

\begin{defn}[pre-hardness gadget]
	\label{def:prehardness}
	Let $J$ be a connected graph.
	We say that $J$ is a \emph{pre-hardness gadget} if, for every (1,2)-supergraph~$H$ of~$J$ without $K_4$-minors,
	$H$ has a hardness gadget. 
\end{defn}

Note that if $J$ is a biconnected graph that is a pre-hardness gadget, then 
every $K_4$-minor-free graph~$H$ which contains $J$ as a biconnected component  
has a hardness gadget.

It will be convenient to establish the following special case of an obstruction.

\begin{lem}\label{lem:triangles_general}
	Let $J$ be a $K_4$-minor-free  biconnected graph such that the largest induced cycle of~$J$ is a square. If $J$ contains a triangle then $J$ is either a pre-hardness gadget or an obstruction.
\end{lem}
\begin{proof}
	Let $(a,b,c,a)$ be a triangle of $J$, let $a_1,\dots,a_k$ be the common neighbours of $b$ and~$c$ with $a=a_1$, let $b_1,\dots,b_\ell$ be the common neighbours of $a$ and~$c$ with $b=b_1$, and let $c_1,\dots,c_m$ be the common neighbours of $a$ and $b$ with $c=c_1$. If at least two of $k,\ell$, and $m$ are at least $2$, then $J$ is a strong hardness gadget by Lemma~\ref{lem:good_triangle_gadget}. In particular, every strong hardness gadget is also a pre-hardness gadget. If $k=\ell=m=1$ then $J$ is a pre-hardness gadget by Corollary~\ref{cor:CycleOfSquares}, as follows. Let $H$ be a (1,2)-supergraph of $J$, let $q=3$, and let $C=(a,b,c,a)$. Then $\abs{N_{C,H}(a)}=\abs{N_{C,H}(b)}=\abs{N_{C,H}(c)}=1$ since $k=\ell=m=1$. Also, suppose for contradiction that there exists a walk $D=(d_a,d_b,d_c,d_a)$ with $d_a\in \NH{a}\setminus\{b,c\}$, $d_b\in \NH{b}\setminus\{a,c\}$ and $d_c\in \NH{c}\setminus\{a,b\}$. Consequently, as we do not allow self-loops in $H$, $d_a\neq a$, $d_b\neq b$ and $d_c\neq c$. Then the vertices $d_a, a, b, c$ induce a $K_4$-minor (contract the edges $\{b,d_b\}$ and $\{c, d_c\}$ to obtain a $K_4$).
	
	Hence assume w.l.o.g.\ that $k>1$ and $\ell=m=1$. If $k$ is odd then $J$ is a pre-hardness gadget by Lemma~\ref{lem:adj_odd_neighbours}. If $k$ is even and all $a_j$ have degree $2$ then $J$ is an obstruction. Otherwise, 
	for some 	$j\in\{1,\dots,k\}$,	let $a_j$ have degree at least $3$. As~$J$ is biconnected, there exists a shortest (induced) path $P$ of length at least $2$ from $a_j$ to one of the vertices $b,c$ or 
	to some $a_i$ with $i\in[k]\setminus\{j\}$. The internal vertices of $P$ are disjoint from $b,c$ and 
	$\{ a_i \mid i\in[k]\}$. If the endpoint of $P$ is one of the other $a_i$, we obtain a $K_4$-minor, hence the endpoint must be $b$ or $c$; suppose w.l.o.g.\ that it is $c$. As the largest induced cycle of $J$ is a square, $P$ has either length $2$ or $3$. In the former case, we obtain a strong (and thus also a pre-) hardness gadget by Lemma~\ref{lem:good_triangle_gadget}. In the latter case, $J$ is a strong (and thus also a pre-) hardness gadget by Lemma~\ref{lem:diamond_square_gadget}.
\end{proof}

\begin{lem}\label{lem:main_biconnected_c5}
	Let $H$ be a biconnected $K_4$-minor-free graph. If $H$ contains an induced cycle of length at least $5$ then $H$ is either an obstruction or a pre-hardness gadget.
\end{lem}
\begin{proof}
	We perform induction on $|V(H)|$: Let $C=(c_0,\dots,c_{q-1},c_0)$ be an induced cycle of length $q\geq 5$.  If $H$ is not an obstruction then, by Definition~\ref{def:obstruction} there exists $i$ such that $N_{C,H}(c_i)$ has even cardinality and contains a vertex of degree not equal to $2$. Assume w.l.o.g.\ that $i=1$. So we can assume that $N_{C,H}(c_1) = \{c_1^1,\dots,c_1^k\}$ where $k>0$ is even and $\deg_H(c_1^1)\neq 2$.
	
	We invoke Lemma~\ref{lem:separators} and obtain a separation $(A,B)$ of $H$ such that $C\setminus\{c_1\}  \subseteq A$, $N_{C,H}(c_1)\subseteq B$ and $A\cap B =\{c_0,c_2\}$. Furthermore, $H$ is a (1,2)-supergraph of $H[B]$. Now consider the neighbours of $c_1^1$: We have that $c_0\in \Gamma_H(c_1^1)$ and $c_2\in \Gamma_H(c_1^1)$ by the definition of $N_{C,H}(c_1)$. As $\deg_H(c_1^1)\neq 2$, there exists another neighbour $w\in \Gamma_H(c_1^1)$. By the properties of the separation $(A,B)$, for any $w\in \Gamma_H(c_1^1)\setminus \{c_0, c_2\}$, $w\in B$.
	
	\medskip
	\noindent \textbf{Claim A:} \textit{
		There is a vertex $w$ in $\Gamma_H(c_1^1)\setminus \{c_0, c_2\}$ and an induced path $P$ in $H[B]$ from $w$ to either $c_0$ or $c_2$ such that all internal vertices of $P$ are contained in $B\setminus (N_{C,H}(c_1) \cup \{c_0,c_2\})$. Furthermore, no internal vertex of $P$ is a neighbour of $c_1^1$.
	}
	\begin{claimproof}
		Let $w'\in \Gamma_H(c_1^1)\setminus \{c_0, c_2\}$.
		As $H$ is biconnected, the vertex $c_1^1$ is not an articulation point. Consequently, there exists a path $P'$ from $w'$ to $c_0$ not containing $c_1^1$ as internal vertex. We can assume $P'$ to be induced by taking possible ``shortcuts''. W.l.o.g. we have that $P'$ does not visit $c_2$ as internal vertex as, otherwise, we can just continue with $c_2$ instead of $c_0$. 
		
		Assume first that $P'$ contains a vertex in $A\setminus B$. As $(A,B)$ is a separation and $w'\in B$, we have that $P'$ is of the form 
		\[w'\stackrel{P_1}{\rightarrow}x\stackrel{P_2}{\rightarrow}c_0\,,\]
		such that $P_1$ is contained in $H[B]$ and $x\in A\cap B=\{c_0,c_2\}$. However, as $P'$ is a path that does not contain $c_2$ as internal vertex, we obtain that $P_2=\emptyset$ and $x=c_0$, contradicting the assumption.
		
		Next assume that $P'$ contains an internal vertex~$z$ in~$N_{C,H}(c_1)\setminus\{c_1^1\}$; we obtain the contradiction by identifying a $K_4$-minor in $H$ as depicted in Figure~\ref{fig:k_4-minorC5}. We have now shown that there is an induced path $P'$ in $H[B]$ from $w'$ to $c_0$ or $c_2$ such that all
		internal vertices of $P'$ are contained in $B\setminus  (N_{C,H}(c_1) \cup \{c_0,c_2\})$. Now choose $w$ to be the first neighbour of $c_1^1$ along $P'$ from $c_0$ or $c_2$, respectively, and let $P$ be the sub-path of $P'$
		going from $c_0$ or $c_2$, respectively, to $w$.
		
		\begin{figure}
			\centering
			\begin{tikzpicture}[scale=2, node distance = 1.4cm,thick]
			\tikzstyle{dot}   =[fill=black, draw=black, circle, inner sep=0.15mm]
			\tikzstyle{vertex}=[  draw=black, circle, inner sep=1.5pt]
			\tikzstyle{dist}  =[fill=white, draw=black, circle, inner sep=2pt]
			\tikzstyle{pinned}=[draw=black, minimum size=10mm, circle, inner sep=0pt]	
			\begin{scope}[xshift=2cm] 
			\node[vertex] (c0) at (3  ,1) [label=90: $c_{0}$] {};
			\node[vertex] (c2) at (3  ,-1) [label=-90: $c_{2}$] {};
			\node[vertex] (c1p) at (2  ,0) [label=180: $z$] {};
			\node[vertex] (c1) at (4  ,0) [label=0: $c_1^1$] {}; 
			\node[vertex] (w) at (3.5  ,0) [label=90: $w'$] {}; 
			
			\draw (c2) -- (c1);\draw (c2) -- (c1p);\draw (c0) -- (c1);\draw (c0) -- (c1p);
			\draw[dashed] (c0) -- (c2); \draw (w) -- (c1); \draw[dotted] (c1p) -- (w);
			\end{scope}
			\end{tikzpicture} 
			\caption{\label{fig:k_4-minorC5}The $K_4$-minor used in the proof of Claim A in Lemma~\ref{lem:main_biconnected_c5}. The dashed line depicts the remainder of the cycle, and the dotted line depicts $P'$.}
		\end{figure}
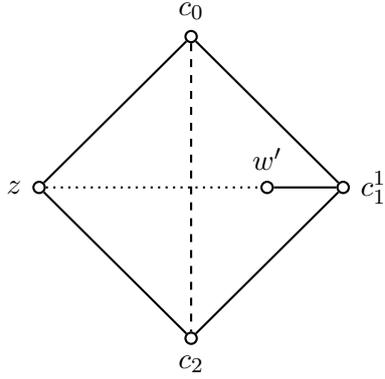
	\end{claimproof}
	
	We assume in the remainder of the proof that the Claim A holds for $c_0$; the case of $c_2$ is completely symmetric (by substituting every subsequent appearance of $c_0$ by $c_2$ and vice versa). For convenience, we also provide an illustration of our current situation in Figure~\ref{fig:cycles_biconnected}. 
	\begin{figure}
		\centering
		\begin{tikzpicture}[scale=2.5, node distance = 1.4cm,thick]
		\tikzstyle{dot}   =[fill=black, draw=black, circle, inner sep=0.15mm]
		\tikzstyle{vertex}=[  draw=black, circle, inner sep=1.5pt]
		\tikzstyle{dist}  =[fill=white, draw=black, circle, inner sep=2pt]
		\tikzstyle{pinned}=[draw=black, minimum size=10mm, circle, inner sep=0pt]	
		\begin{scope}[xshift=2cm] 
		\node[vertex] (c0) at (3  ,1) [label=90: $c_{0}$] {};
		\node[vertex] (c2) at (3  ,-1) [label=-90: $c_{2}$] {};
		\node[vertex] (c1k) at (2  ,0) [label=180: $c_1^k$] {};
		\node[vertex] (c11) at (4  ,0) [label=0: $c^1_1$] {}; 
		\node[vertex] (c12) at (3.5  ,0) [label=0: $c_1^2$] {};
		\node[vertex] (c1km1) at (2.5  ,0) [label=180: $c^{k-1}_1$] {}; 
		\node[vertex] (w) at (5  ,1) [label=90: $w$] {}; 
		
		\node[dot] (d) at (3  ,0) {}; 
		\node[dot] (d) at (2.75  ,0) {}; 
		\node[dot] (d) at (3.25  ,0) {}; 
		
		\draw (c2) -- (c12); \draw (c2) -- (c1km1); \draw (c0) -- (c12); \draw (c0) -- (c1km1);
		\draw (c2) -- (c11);\draw (c2) -- (c1k);\draw (c0) -- (c11);\draw (c0) -- (c1k);
		\draw[dashed] (c0) to[out=180,in=180,distance=2.5cm] (c2); \draw (w) -- (c11); \draw[dotted] (c0) -- (w);
		\end{scope}
		\end{tikzpicture} 
		\medskip
		\caption{\label{fig:cycles_biconnected} Illustration of cycle $C$ consisting of the dashed line and one of the vertices $c_1^i$, and path $P$ (dotted) in the proof of Lemma~\ref{lem:main_biconnected_c5}.}
	\end{figure}
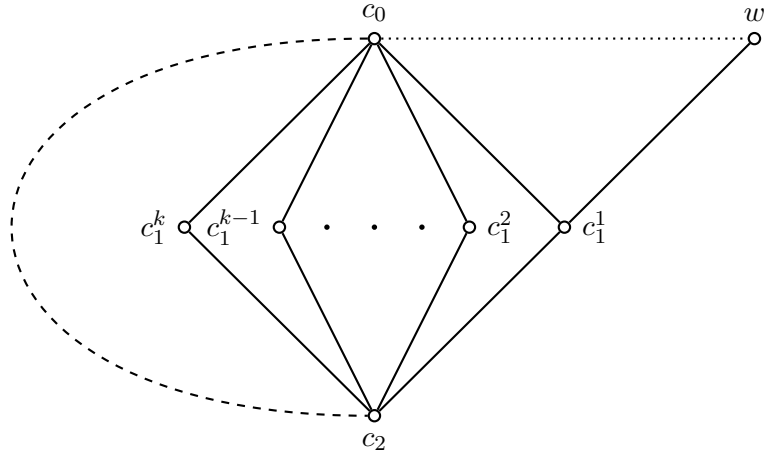
	For the remainder of the proof, we need the following observation:
	
	\medskip
	\noindent \textbf{Claim B:} \textit{
		$H[B]$ is biconnected.
	}
	\begin{claimproof}\label{clm:biconnected}
		By Menger's Theorem, we have to show that there are two internally vertex-disjoint paths (in $H[B]$) between every pair of different vertices $x$ and $y$ in $B$. As $H$ is biconnected, there are two such paths $Q_1$ and $Q_2$ connecting $x$ and $y$ in $H$. If $\{x,y\}=\{c_0,c_2\}$, then the claim follows immediately as $N_{C,H}(c_1)\subseteq B$ and $|N_{C,H}(c_i)|\geq 2$.
		
		Hence we can assume that $\{x,y\} \neq \{c_0,c_2\}$. The next step is to show that 
		at least one of $Q_1$ and $Q_2$ is fully contained in $H[B]$.
		If $\{x,y\}$ intersects $\{c_0,c_2\}$ (for example, if $y=c_2$) then this is clear because $c_0$ can only be on one of $Q_1, Q_2$.
		Otherwise, suppose that one of the paths, say $Q_2$, from $x$ to $y$, leaves $H[B]$.
		It leaves by one of the vertices in the separator $\{c_0,c_2\}$ and returns by the other. So $Q_1$ stays within $H[B]$. If $Q_2$ is also fully contained in $H[B]$ we are done. 
		
		Otherwise we have that w.l.o.g.\ (otherwise switch $c_0$ and $c_2$ and proceed symmetrically):
		\[Q_2= x \stackrel{Q^1_2}{\rightarrow} c_0 \stackrel{Q^2_2}{\rightarrow} c_2 \stackrel{Q^3_2}{\rightarrow} y \,,\]
		where $Q^1_2$ and $Q^3_2$ are in $H[B]$ and $Q^2_2$ is non-empty in $H[A\setminus B]$. Next we claim that $Q_1$ contains at most one vertex in $N_{C,H}(c_1)$ as internal vertex. Assuming otherwise, we have 
		\[Q_1= x \stackrel{Q^1_1}{\rightarrow} c_1^1 \stackrel{Q^2_1}{\rightarrow} c^2_1 \stackrel{Q^3_1}{\rightarrow} y \,,\]
		where $c_1^1\neq c_1^2 \in N_{C,H}(c_1)$. As $Q_1$ is fully contained in $H[B]$, we obtain a $K_4$-minor, unless $Q^2_1$ contains $c_0$ or $c_2$ as internal vertices: The $K_4$-minor is induced by $c_0,c_1^1,c_1^2,c_2$ --- note that $c_0$ is connected to $c_2$ by $C\setminus\{c_1\}$, and $c_1^1$ is connected to $c_1^2$ by $Q^2_1$.
		
		Thus we can assume that $Q^2_1$ contains $c_0$ or $c_2$ as internal vertices.
		\begin{itemize}
			\item If $c_2$ is an internal vertex of $Q^2_1$, then $y\neq c_2$. In this case, however, $Q_1$ and $Q_2$ share $c_2$ as internal vertex, which leads to a contradiction.
			\item If $c_0$ is an internal vertex of $Q^2_1$, then $x\neq c_0$. In this case, however, $Q_1$ and $Q_2$ share $c_0$ as internal vertex, which leads to a contradiction.
		\end{itemize}
		Consequently, $Q_1$ contains at most one vertex in $N_{C,H}(c_1)$. As $N_{C,H}(c_1)$ is of even positive cardinality, there exists hence a vertex $z \in N_{C,H}(c_1)$ which is not part of $Q_1$. Finally, this enables us to modify $Q_2$ be substituting $Q_2^2$ by the path $c_0,z,c_2$. The resulting path is fully contained in $H[B]$ and, by the previous analysis, internally vertex-disjoint from $Q_1$. This concludes the proof of Claim~B.
	\end{claimproof}
	We proceed with the following claim.
	
	\medskip
	\noindent \textbf{Claim C:} \textit{
		If $H[B]$ is an obstruction, then so is $H$.
	}
	\begin{claimproof}
		If $H[B]$ is an obstruction then it contains an induced cycle $D$ satisfying the requirements of Definition~\ref{def:obstruction}. For the sake of readability, we state those requirements explicitly:
		The graph~$H[B]$ contains an induced cycle $D=(d_0,\dots,d_{r-1},d_0)$ for some $r\neq 4$. Furthermore, we have that for all $i$, every vertex in $N_{D,H[B]}(d_i)$ has degree~$2$ in~$H[B]$, unless $|N_{D,H[B]}(d_i)|$ is odd.
		
		We claim that $H$ is an obstruction with cycle $D$: Observe that
		\[N_{D,H[B]}(d_i) = \Gamma_{H[B]}(d_{i-1})\cap \Gamma_{H[B]}(d_{i+1}) = \Gamma_H(d_{i-1})\cap \Gamma_H(d_{i+1}) = N_{D,H}(d_i)\,,\]
		where the second equality is true as $H$ is a (1,2)-supergraph of $H[B]$. Consequently, it remains to show that for all $i$ with $|N_{D,H}(d_i)|$ even, every vertex in $N_{D,H}(d_i)$ has degree $2$ in $H$. For the sake of contradiction, we assume w.l.o.g.\ that $N_{D,H}(d_1)$ is of even cardinality and contains a vertex $\hat{d}_1$ such that $\hat{d}_1$ has degree $2$ in $H[B]$, but degree at least $3$ in $H$. As the separator of $(A,B)$ is $\{c_0,c_2\}$, the only possibility for this to happen is $\hat{d}_1 = c_0$ or $\hat{d}_1 = c_2$. However, $\hat{d}_1 = c_0$ is impossible, as $c_0$ has at least three neighbours already in $H[B]$: $c_0$ is adjacent to every vertex in $N_{C,H}(c_1)$, which is of positive even cardinality (i.e., of size at least $2$), and $c_0$ is adjacent to the first vertex in the path $P$ from $c_0$ to $w$ (see Figure~\ref{fig:cycles_biconnected}).
		
		Hence the remaining possibility is $\hat{d}_1 = c_2$. Recall that $\hat{d}_1 = c_2$ has neighbours $c_1^1,\dots,c_1^k$ in $B$.
		So if there are only two of them, then $k=2$ and $|N_{C,H}(c_1)|=2$. However, as $\hat{d}_1 \in N_{D,H}(d_1)$ has degree $2$ in $H[B]$, and $c^1_1$ and $c^2_1$ are adjacent to $\hat{d}_1=c_2$ in $H[B]$, we obtain that $\{c_1^1,c^2_1\} = \{d_0,d_2\}$ --- recall that $\hat{d}_1\in \Gamma_{H[B]}(d_0)\cap\Gamma_{H[B]}(d_2)$. Finally, $N_{D,H}(d_1)$ has positive, even cardinality. Thus there exists a vertex $d'_1\neq \hat{d}_1$ in $N_{D,H}(d_1)$ which is also adjacent to $d_0$ and $d_2$. This yields the following $K_4$-minor of $H$; note that $N_{D,H}(d_1) \subseteq B$ and the dashed line is $C\setminus c_1$, which is in $A$.
		
		
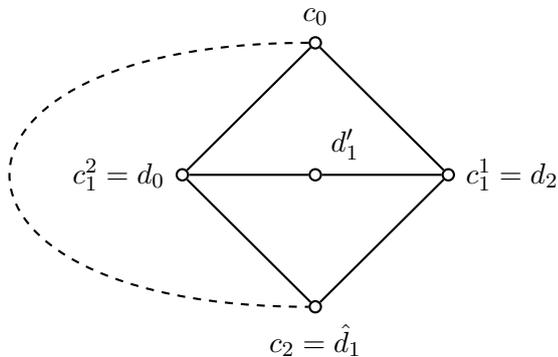
\medskip
		\begin{tikzpicture}[scale=1.75, node distance = 1.4cm,thick]
		\tikzstyle{dot}   =[fill=black, draw=black, circle, inner sep=0.15mm]
		\tikzstyle{vertex}=[  draw=black, circle, inner sep=1.5pt]
		\tikzstyle{dist}  =[fill=white, draw=black, circle, inner sep=2pt]
		\tikzstyle{pinned}=[draw=black, minimum size=10mm, circle, inner sep=0pt]	
		\begin{scope}[xshift=2cm] 
		\node[vertex] (c0) at (3  ,1) [label=90: $c_{0}$] {};
		\node[vertex] (c2) at (3  ,-1) [label=-90: ${c_{2}=\hat{d}_1}$] {};
		\node[vertex] (c12) at (2  ,0) [label=180: ${c_1^2 = d_0}$] {};
		\node[vertex] (c11) at (4  ,0) [label=0: ${c^1_1 = d_2}$] {}; 
		\node[vertex] (dhp) at (3  ,0) [label=45: $d'_1$] {}; 
		
		\draw (c2) -- (c12);  \draw (c0) -- (c12);
		\draw (c2) -- (c11);;\draw (c0) -- (c11);
		\draw[dashed] (c0) to[out=180,in=180,distance=3cm] (c2); \draw (dhp) -- (c11); 
		\draw (dhp) -- (c12); 
		\end{scope}
		\end{tikzpicture} 
		\medskip
		
		\noindent This concludes the proof of Claim~C.
	\end{claimproof}
	
	In what follows, we perform a case distinction along the length $L$ of the largest induced cycle in $H[B]$.
	\begin{enumerate}[(I)]
		\item $L\geq 5$. This allows us to invoke the induction hypothesis to the graph $H[B]$; note that $|V(H[B])|=|B|$ is indeed strictly smaller than $|V(H)|$ as the cycle $C$ has length at least~$5$, and thus $A$ is not empty. Furthermore, $H[B]$ is biconnected by Claim B. If $H[B]$ is a pre-hardness gadget, then so is $H$, as $H$ is a (1,2)-supergraph of~$H[B]$. If $H[B]$ is an obstruction, then so is $H$ by Claim~C.
		\item $L\leq 4$. Consider again the path $P$ in Figure~\ref{fig:cycles_biconnected}. By the assumption of this case, $P$ is either an edge or a 2-path. If $P$ is an edge, then $H[B]$ satisfies all conditions of Lemma~\ref{lem:triangles_general}. Consequently, $H[B]$ is a pre-hardness gadget or an obstruction. In the former case, we are done as $H$ is a (1,2)-supergraph of $H[B]$ and thus also a pre-hardness gadget. In the latter case, we obtain that $H$ is an obstruction as well by invoking Claim~C.
		
		Finally, assume that $P$ is a 2-path. We claim that $H$ is a pre-hardness gadget. To this end, let $H'$ be a $K_4$-minor-free (1,2)-supergraph of $H$. Then $H'$ contains the following subgraph: 
		
		\medskip
		
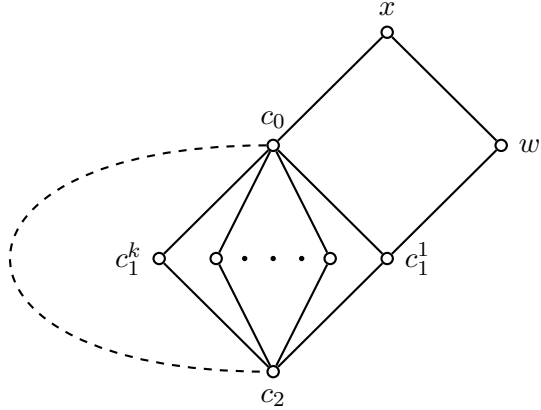
\begin{tikzpicture}[scale=1.5, node distance = 1.4cm,thick]
		\tikzstyle{dot}   =[fill=black, draw=black, circle, inner sep=0.15mm]
		\tikzstyle{vertex}=[  draw=black, circle, inner sep=1.5pt]
		\tikzstyle{dist}  =[fill=white, draw=black, circle, inner sep=2pt]
		\tikzstyle{pinned}=[draw=black, minimum size=10mm, circle, inner sep=0pt]	
		\begin{scope}[xshift=2cm] 
		\node[vertex] (c0) at (3  ,1) [label=90: $c_{0}$] {};
		\node[vertex] (c2) at (3  ,-1) [label=-90: $c_{2}$] {};
		\node[vertex] (c1k) at (2  ,0) [label=180: $c_1^k$] {};
		\node[vertex] (c11) at (4  ,0) [label=0: $c^1_1$] {}; 
		\node[vertex] (c12) at (3.5  ,0)  {};
		\node[vertex] (c1km1) at (2.5  ,0)  {}; 
		\node[vertex] (w) at (5  ,1) [label=0: $w$] {}; 
		\node[vertex] (x) at (4,2) [label=90: $x$] {};
		
		\node[dot] (d) at (3  ,0) {}; 
		\node[dot] (d) at (2.75  ,0) {}; 
		\node[dot] (d) at (3.25  ,0) {}; 
		
		\draw (c2) -- (c12); \draw (c2) -- (c1km1); \draw (c0) -- (c12); \draw (c0) -- (c1km1);
		\draw (c2) -- (c11);\draw (c2) -- (c1k);\draw (c0) -- (c11);\draw (c0) -- (c1k);
		\draw[dashed] (c0) to[out=180,in=180,distance=3cm] (c2); \draw (w) -- (c11); 
		\draw (c0) -- (x); \draw (w) -- (x);
		\end{scope}
		\end{tikzpicture}
		\medskip
		
		\noindent In particular, we have that $c_1^1$ and $c_1^k$ have no common neighbours in $H'$ apart from $c_0$ and $c_2$: This is due to the fact that they have no further common neighbours in $H$, as otherwise $H$  has a $K_4$-minor similarly as in the proof of Claim~A, and as $H'$ is a (1,2)-supergraph, it cannot add common neighbours to vertices. Furthermore, we have that $k>0$ is even and that $c_1^1,\dots,c_1^k$ are all common neighbours of $c_0$ and $c_2$ in $H$ and thus in $H'$. We apply Lemma~\ref{lem:typeV_hard} to the subgraph of $H'$ induced by the vertices $c_1^k,c_0,x,w,c_1^1,c_2$ and obtain a hardness gadget in $H'$, unless this subgraph, call it $F$, is of type $\TV$. By the previous argument, the only possibility for $F$ being of type $\TV$ is $k$ being strictly greater than~$2$. However, as $k$ is even, we obtain a hardness gadget in $H'$ in this case as well: We found an instance of Lemma~\ref{lem:even_diamond}. 
	\end{enumerate}
\end{proof}

\subsection{$K_4$-minor-free Component Lemma}

\begin{lem}[$K_4$-minor-free Component Lemma]\label{lem:main_biconnected}
	Let $B$ be a biconnected $K_4$-minor-free graph. If $B$ is not an edge then at least one of the following is true:
	\begin{enumerate}[(a)]
		\item $B$ is a diamond.
		\item $B$ is an obstruction.
		\item $B$ is an impasse.
		\item For every $K_4$-minor-free graph $H$ containing $B$ as a biconnected component, $H$ has a hardness gadget.
	\end{enumerate}
\end{lem}
\begin{proof}
	Let $L$ be the size of the largest induced cycle of $B$. Note that $L\geq 3$ is well-defined as $B$ is biconnected, but not an edge. If $L\geq 5$ we obtain by Lemma~\ref{lem:main_biconnected_c5} that $B$ is either a pre-hardness gadget or an obstruction. In the latter case, (b) holds. In the former case, (d) holds, as every $K_4$-minor-free graph $H$ containing $B$ as a biconnected component is a (1,2)-supergraph of $B$.
	
	If $L\leq 4$ and $B$ contains a triangle, then $B$ is either a pre-hardness gadget or an obstruction by Lemma~\ref{lem:triangles_general}. Similarly as before, (b) or (d) hold.
	
	In the remaining case, $B$ is chordal bipartite and we can invoke Lemma~\ref{lem:main_chordal_bipartite}, yielding that either (a), (c) or (d) hold. 
\end{proof}

\section{$K_4$-minor-free Graphs}\label{sec:k4MinorFreeMain}
\subsection{Suitable Connectors}

\begin{defn}[suitable connector]\label{def:suitableconnector}
	Let $H$ be a graph, let $B$ be a biconnected component of $H$, and let $A\subseteq V(B)$ be a set of articulation points of $H$. We say that $(B,A)$ is a \emph{suitable connector} in $H$ if one of the following cases holds:
	\begin{myitemize}
		\item $B$ is an edge $\{a,b\}$ and $A=\{a,b\}$, or
		\item $B$ is a diamond (Definition~\ref{def:diamond}) that contains an edge $\{a,b\}$ such that $A=\{a,b\}$, or
		\item $B$ is an impasse (Definition~\ref{def:impasse}) that has a pair of connectors $(a,b)$ such that $A=\{a, b\}$, or
		\item $B$ is an obstruction (Definition~\ref{def:obstruction}).
		In this case $(B,A)$ is a suitable connector in~$H$
		if there is a cycle $C\in \cycles(B)$
		such that  
		$A=\{c \in C \mid \text{the cardinality of } N_{C,H}(c)  \text{ is even}\}$. 
		Note that $A$ could be the empty set.
		If $(B,A)$ is a suitable connector in~$H$
		then we fix a particular cycle $C(B,A)\in \cycles(B)$
		such that 
		$$A=\{c \in C(B,A) \mid \text{the cardinality of } N_{C(B,A),H}(c)  \text{ is even}\}.$$ 
		(It does not matter if there are multiple possibilities for $C(B,A)$ in $\cycles(B)$ --- we just fix one, for example, the lexicographically least one.)
	\end{myitemize}
\end{defn}

\begin{lem}\label{lem:diamondsuitableconnector}
	Let $B$ be a biconnected component in an involution-free graph $H$. 
	If $B$ is a diamond then there exists a set $A\subseteq V(B)$ of articulation points of $H$ such that $(B,A)$ is a suitable connector in $H$.
\end{lem}
\begin{proof}
	If $B$ is a diamond with vertices as given in Definition~\ref{def:diamond}, then as $H$ is involution-free there exist articulation points $a\in \{s,t\}$ and $b\in \{x_1, \dots, x_k\}$. Hence, for $A=\{a,b\}$, $(B,A)$ is a suitable connector in $H$.
\end{proof}

\begin{lem}\label{lem:impassesuitableconnector}
	Let $B$ be a biconnected component in an involution-free graph $H$. If $B$ is an impasse then there exists a set $A\subseteq V(B)$ of articulation points of $H$ such that $(B,A)$ is a suitable connector in $H$. 
\end{lem}
\begin{proof}
	If $B$ is an impasse with vertices as given in Definition~\ref{def:impasse}, then as $H$ is involution-free there exist articulation points $a\in\{v_1, y_1, \dots, y_k\}$ and $b\in \{v_3, z_1, \dots, z_\ell\}$. Note that $(a,b)$ is a pair of connectors (cf.~Definition~\ref{def:impasse}) and hence, for $A=\{a,b\}$, $(B,A)$ is a suitable connector in $H$.
\end{proof}

\begin{lem}\label{lem:obstructionsuitableconnector}
	Let $B$ be a biconnected component in an involution-free graph $H$. If $B$ is an obstruction then there exists a set $A\subseteq V(B)$ of articulation points of $H$ such that $(B,A)$ is a suitable connector in $H$. 
\end{lem}
\begin{proof}
	If $B$ is an obstruction then there exists a cycle $C$ with $C\in \cycles(B)$. Let $c\in C$ such that $\abs{N_{C,H}(c)}$ is even. By definition of an obstruction, every vertex in $\abs{N_{C,H}(c)}$ has degree $2$ in $B$. Since $c\in N_{C,H}(c)$, $\abs{N_{C,H}(c)}\ge 2$. Therefore, as $H$ is involution-free, at least one vertex in $N_{C,H}(c)$ is an articulation point of $H$. By renaming vertices, we can assume without loss of generality that $c$ is an articulation point. Hence, for $A=\{c \in C \mid \text{the cardinality of } N_{C,H}(c)  \text{ is even}\}$, $(B,A)$ is a suitable connector in $H$, where $C(B,A)=C$.
\end{proof}

\subsection{Finding a Suitable Subtree}
In this section we will use the notion of rooted trees. Given a tree $T$ and a vertex $r$ in $T$, $(T,r)$ is a \emph{rooted tree} and the \emph{tree-order} $<_r$ induced by $r$ (on $T$) is the partial order of the vertices of $T$, where for vertices $u$ and $v$ of $T$ we have $u <_r v$ if and only if the unique path from $r$ (the root) to $v$ passes through $u$.
Such a partial order gives rise to the standard notion of child, parent, ancestor and descendant. In order to clarify which tree-order we are referring to we speak of an $r$-child, $r$-parent, $r$-ancestor and $r$-descendant when we mean child, parent, ancestor and descendant with respect to $<_r$.

For a connected graph $H$, recall the definition of the block-cut tree $\bc(H)$ from Definition~\ref{def:blockcuttree}.

\begin{defn}[$R$-open, $R$-closed]\label{def:Rclosed}
	Let $H$ be a connected graph, let $a$ be a cut vertex in $\bc(H)$, and let $R$ be a block in $\bc(H)$. If $a$ has exactly one descendant with respect to $<_R$ in $\bc(H)$ and this descendant is a block in $\bc(H)$ that is an edge, then $a$ is \emph{$R$-closed} (in $\bc(H)$).
	Otherwise, $a$ is \emph{$R$-open} (in $\bc(H)$).
\end{defn}

\begin{defn}[suitable subtree, closed]\label{def:suitabletree}
	Let $H$ be a connected graph. Let $T$ be a subtree of $\bc(H)$. 
	We say that $T$ is \emph{suitable} if it has the following properties:
	\begin{myenumerate}
		\item \label{item:suitable2} For every block $B$ in $T$, $(B,\Nof{T}(B))$ is a suitable connector in $H$ (Definition~\ref{def:suitableconnector}).
		\item \label{item:suitable3} Every cut vertex of $T$ has degree at most $2$ in $T$.	
	\end{myenumerate}
	
	A suitable subtree $T$ is \emph{closed} if there exists a block $R$ in $T$ such that every cut vertex that is a leaf in $T$ is $R$-closed in $\bc(H)$.	
\end{defn}

\begin{lem}\label{lem:choiceofRdoesnotmatter}
	Let $H$ be a connected graph and let $T$ be a suitable subtree of $\bc(H)$. Let $R$ and $R'$ be distinct blocks in $T$ and let $a$ be a cut vertex that is a leaf in $T$. If $a$ is $R$-closed in $\bc(H)$ then it is $R'$-closed in $\bc(H)$.
\end{lem}
\begin{proof}
	Let $B$ be a block of $\bc(H)$. We show that $B$ is an $R$-descendant of $a$ in $\bc(H)$ if and only if it is an $R'$-descendant. From this it follows immediately that if $a$ is $R$-closed in $\bc(H)$ then it is $R'$-closed in $\bc(H)$.
	Let $B$ be an $R$-descendant of $a$. Since $a$ is a leaf of $T$ and $R$ is in $T$ it follows that $B$ is not in $T$. Since $R$, $R'$ and $a$ are all in $T$, there is a path in $T$ from $R'$ to $a$ and consequently this path does not contain $B$. Hence the unique path from $R'$ to $B$ goes through $a$, which means that $B$ is an $R'$-descendant of $a$ in $\bc(H)$. It is analogous to show that if $B$ is an $R'$-descendant of $a$ it is also an $R$-descendant.
\end{proof}

The following lemma gives the initialisation for finding a closed suitable subtree (which is then done in Lemma~\ref{lem:findsuitablesubtree}).
\begin{lem}\label{lem:initsuitablesubtree}
	Let $H$ be an involution-free, connected graph such that every biconnected component of $H$ is an edge, a diamond, an impasse or an obstruction. Then there exists a biconnected component $B_0$ and a set of articulation points $A_0\subseteq V(B_0)$ such that $(B_0, A_0)$ is a suitable connector in $H$ and hence $T(B_0)=(\{B_0\}\cup A_0, \{\{B_0, a\}\mid a\in A_0\})$ is a suitable subtree of $\bc(H)$.
\end{lem}
\begin{proof}
	First note that if all biconnected components of $H$ are edges, then there is at least one edge between articulation points as $H$ is involution-free and therefore $H$ is not a star. Therefore, $H$ contains a biconnected component $R$ that is one of the following: a diamond, an impasse, an obstruction, or an edge for which both endpoints are articulation points of $H$.
	In the first three cases we can use Lemmas~\ref{lem:diamondsuitableconnector},~\ref{lem:impassesuitableconnector} or~\ref{lem:obstructionsuitableconnector}, respectively, to obtain a suitable connector. If $B_0$ is an edge $\{a,b\}$ where both end points are articulation points, then $(B_0, \{a,b\})$ is a suitable connector. Then it is immediate that $T(B_0)$ is a suitable subtree of $\bc(H)$.
\end{proof}

\begin{lem}\label{lem:findsuitablesubtree}
	Let $H$ be an involution-free, connected graph such that every biconnected component of $H$ is an edge, a diamond, an impasse or an obstruction. Then there exists a closed suitable subtree of $\bc(H)$.
\end{lem}
\begin{proof}
	Let $B_0, A_0, T(B_0)$ be as given by Lemma~\ref{lem:initsuitablesubtree}.
	Algorithm~\ref{alg:findsuitabletree} keeps track of a suitable subtree $T$ of $\bc(H)$, a block $R$ of $T$, and the set $A(T)$ of leaves of $T$ that are cut vertices (i.e., that are articulation points of $H$). 
	
	\begin{algorithm}[H]
		\caption{}\label{alg:findsuitabletree}
		\begin{algorithmic}
			\State $T \leftarrow T(B_0)$
			\State $R \leftarrow B_0$
			\State $A(T) \leftarrow A_0$
			\While{$A(T)$ contains an $R$-open cut vertex $a^*$}
			\State // Invariant: All elements of $A(T)$ are $B_0$-descendants of $R$.
			\If {there is a suitable connector $(B,A)$ in $H$ such that $B$ is an $R$-child of $a^*$ and $a^*\in A$}
			\State // By the invariant, every element of $A\setminus\{a^*\}$ is a $B_0$-descendant of $a^*$.
			\State $V \leftarrow V(T) \cup \{B\} \cup A$
			\State $E \leftarrow E(T) \cup \{\{B, a\}\mid a\in A\}$
			\State $T \leftarrow (V, E)$
			\State $A(T) \leftarrow (A(T) \cup A)\setminus \{a^*\}$
			\Else 
			\State Choose a suitable connector $(B,A)$ in $H$ such that $B$ is an $R$-child of $a^*$ in $\bc(H)$.
			\State // By the invariant, every element of $A$ is a $B_0$-descendant of $a^*$.
			\State $V \leftarrow \{B\} \cup A$
			\State $E \leftarrow \{\{B, a\}\mid a\in A\}$
			\State $T \leftarrow (V, E)$
			\State $R\leftarrow B$
			\State $A(T) \leftarrow A$
			\EndIf
			\EndWhile
		\end{algorithmic}
	\end{algorithm}
	
	We now show that Algorithm~\ref{alg:findsuitabletree} is well-defined and finds a closed suitable subtree.\footnote{Since the graph $H$ is fixed, the running time of Algorithm~\ref{alg:findsuitabletree} is not important for us. What is important is that the algorithm gives us a (constructive) proof that such a closed suitable subtree exists.} 
	In order to show that the algorithm is well-defined note that any $R$-open cut vertex $a^*$ is an articulation point of $H$ and therefore is adjacent to at least two blocks of $\bc(H)$. At most one of these blocks can be an $R$-parent. Therefore $a^*$ has an $R$-child in $\bc(H)$. If there is such an $R$-child $B$ that is a diamond, an impasse, or an obstruction, then by Lemmas~\ref{lem:diamondsuitableconnector},~\ref{lem:impassesuitableconnector} or~\ref{lem:obstructionsuitableconnector}, respectively, there exists a suitable connector of the form $(B,A)$. If otherwise all $R$-children of $a^*$ are edges then $a^*$ has at least one such $R$-child $B=\{a^*,b\}$ for which $b$ is an articulation point (as $a^*$ is $R$-open and $H$ is involution-free). Therefore $(B, \{a^*,b\})$ is a suitable connector. Thus, the algorithm is well-defined as we can always choose a suitable connector $(B,A)$ where $B$ is an $R$-child of $a^*$.

	We next show that at any point during the algorithm, $T$ is a suitable subtree of $\bc(H)$, $R$ is a block in $T$, and $A(T)$ is the set of leaves of $T$ that are cut vertices of $\bc(H)$. First note that in the initialisation this clearly holds by Lemma~\ref{lem:initsuitablesubtree}. 
	We show that after each update these properties still hold. 
	Note that if we update $T$, $R$, and $A(T)$ as part of the else-block then $R=B$ is the only block in $T$, $\Nof{T}(B)=A$, and $(B,A)$ is a suitable connector. Thus, $T$ is a suitable subtree. Furthermore, the cut vertex leaves of $T$ are precisely the elements of $A$ and we have $A(T)=A$, as required.  
	
	If otherwise we update $T$ and $A(T)$ as part of the if-block then
	\begin{myenumerate}
		\item The block $R$ continues to be a vertex of $T$. 
		\item We add precisely one block $B$ together with the articulation points $A$ and the edges $\{\{B, a\}\mid a\in A\}$, which ensures that $\Nof{T}(B)=A$ and hence $(B,\Nof{T}(B))$ is a suitable connector.
		\item All cut vertices in $A\setminus \{a^*\}$ are leaves in $T$ and since $a^*$ was a leaf before the update, it now has degree $2$ in $T$.
	\end{myenumerate}
	Consequently, $T$ is a suitable subtree after the update. Furthermore, we remove $a^*$ from $A(T)$ as it now has degree $2$ in $T$, and we add the cut vertices $A\setminus \{a^*\}$ to $A(T)$ since they are leaves in $T$.
	
	We have established that at any point during the algorithm, $T$ is a suitable subtree of $\bc(H)$, $R$ is a block in $T$, and $A(T)$ is the set of leaves of $T$ that are cut vertices. It remains to show that Algorithm~\ref{alg:findsuitabletree} terminates (in which case it is immediate that $T$ is a closed suitable subtree). Note that with each iteration we remove a vertex $a^*$ from $A(T)$. With each iteration we may also add some vertices to $A(T)$. As noted in Algorithm~\ref{alg:findsuitabletree}, the vertices that  are added in each iteration are always $B_0$-descendants in $\bc(H)$ of the vertex $a^*$ that is deleted.   It follows immediately that Algorithm~\ref{alg:findsuitabletree} terminates as we only consider finite graphs.
\end{proof}

\subsection{Suitable Subtrees without Obstructions}

\begin{lem}\label{lem:suitableTreeHardnessNoObstruction}
	Let $H$ be a connected graph and let $T$ be a closed suitable subtree of $\bc(H)$. If no block of $T$ is an obstruction then $H$ has a hardness gadget.
\end{lem}
\begin{proof}
	As $T$ does not contain an obstruction, the degree of every block in $T$ is $2$. Together with the fact that every cut vertex has degree at most $2$, this implies that, for a non-negative integer $q$, $T$ is a path of the form $(b_0, B_1, b_1, B_2, \dots, B_q, b_q)$, where $B_1, \dots , B_q$ are blocks, i.e.~biconnected components of $H$, and $b_0, \dots, b_q$ are cut vertices, i.e.~articulation points of $H$. Since $T$ is closed it contains at least one block $R$ and therefore $q\ge 1$. Furthermore, for each $i\in [q]$, $(B_i, \{b_{i-1}, b_i\})$ is a suitable connector. And since $B_i$ is no obstruction, one of the following holds:
	\begin{myitemize}
		\item $B_i$ is an edge $\{b_{i-1}, b_i\}$, or
		\item $B_i$ is a diamond that contains the edge $\{b_{i-1},b_i\}$, or
		\item $B_i$ is an impasse such that $(b_{i-1},b_i)$ is a pair of connectors.
	\end{myitemize}	
	Since $T$ is closed, there is a block $R$ among $B_1, \dots, B_q$ such that both $b_0$ and $b_q$ are $R$-closed. By Lemma~\ref{lem:choiceofRdoesnotmatter}, $b_0$ is $B_1$-closed and $b_q$ is $B_q$-closed. It follows that $\abs{\NH{b_0} \setminus V(B_1)}=1$ and $\abs{\NH{b_q} \setminus V(B_q)}=1$.
	
	Thus, we can apply Lemma~\ref{lem:CaterpillarForBipChordalComponents} to obtain that $H$ has a hardness gadget or otherwise there exists $L_q\subseteq \Nof{B_q}(b_q)$ such that $(L_q, b_q)$ is a good start in $B_q$. Since $\NH{b_q}\setminus V(B_q)$ has odd cardinality, this means that $(L_q, b_q)$ is a good stop in $B_q$. Then Lemma~\ref{lem:CaterpillarForBipChordalComponents} ensures that $H$ has a hardness gadget in this case as well. 
\end{proof}

\subsection{Suitable Subtrees with Obstructions}

The goal of this section is to prove Lemma~\ref{lem:suitableTreeHardnessWithObstruction}, which gives a hardness gadget in a connected $K_4$-minor-free graph using a closed suitable subtree that contains an obstruction. In order to find this hardness gadget we use Lemma~\ref{lem:CycleHardness+}, which derives a hardness gadget based on the generalised cycle gadget from Definition~\ref{def:generalisedcyclegadget}. The sets of vertices $\calC_0, \dots, \calC_{q-1}$ from Lemma~\ref{lem:CycleHardness+} will correspond to the walk-neighbour-sets of a specific closed walk $W$. With Algorithms~\ref{alg:closedwalksub} and~\ref{alg:closedwalk}  we define this walk $W$ --- it is the output of Algorithm~\ref{alg:closedwalk}. In Lemmas~\ref{lem:feasibleinput},~\ref{lem:closedwalksub} and~\ref{lem:closedwalk} we establish that the algorithms are well-defined and give as output
a closed walk in~$H$ whose length is at least~$3$, and not equal to~$4$.  
In Figure~\ref{fig:closedwalk} we give an example that illustrates how $W$ is derived. In Lemmas~\ref{lem:properWalkIntersection} and~\ref{lem:noDwalk} we then show that the walk-neighbour-sets of $W$ satisfy certain properties required to apply Lemma~\ref{lem:CycleHardness+}. In the proof of Lemma~\ref{lem:suitableTreeHardnessWithObstruction} we put all the pieces together and establish the remaining necessary properties of $W$.

\begin{defn}[obstruction-free path, proper]\label{def:obstructionfreePath}
	Let $H$ be a connected graph and let $T$ be a closed suitable subtree of $\bc(H)$. A path in $T$ is \emph{obstruction-free} if it does not contain a block that is an obstruction. An obstruction-free path is \emph{proper} if its endpoints are cut vertices of $\bc(H)$. Note that it is possible that a proper obstruction-free path has length $0$. Then it is of the form $(v)$ where $v$ is a cut vertex of $\bc(H)$.
\end{defn}

\begin{defn}[$\SP{a,b}$]\label{def:pathroutine}
	Let $H$ be a graph and let $a$ and $b$ be vertices of $H$. If $a=b$ then $\SP{a,b}=(a)$.
	If $a\neq b$ then $\SP{a, b}$ is a shortest path from $a$ to $b$ in $H$.
\end{defn}

In Definition~\ref{def:pathroutine}, it is of course possible that $H$ might have multiple shortest paths from~$a$ to~$b$.
In this case, it doesn't matter which of these is chosen to be $\SP{a,b}$ --- for concreteness, the reader may assume that
$\SP{a,b}$ is the lexicographically least of these. (In fact, when we use the definition, this shortest path will turn out to be unique.)

\begin{lem}\label{lem:obstructionfreepath}
	Let $H$ be a connected graph and let $T$ be a closed suitable subtree of $\bc(H)$. For a non-negative integer $q$, let $P=(b_0, B_1, b_1, B_2, \dots, B_q, b_q)$ be a proper obstruction-free path in $T$.
	Then $\SP{b_0, b_q}$  is the unique 
	shortest path from~$a$ to~$b$ in~$H$. It
	passes through $b_0, b_1, \dots, b_q$ in order. For $i\in [q]$, the subpath of $\SP{b_0, b_q}$ that connects $b_{i-1}$ and $b_i$ is either an edge or 
	it is of the form $(b_{i-1},v,b_i)$, where $v$ is the unique common neighbour of $b_{i-1}$ and $b_i$ in $H$.
\end{lem}
\begin{proof}
	Since $P$ is a proper obstruction-free path, $b_0, b_1 \dots, b_q$ are cut vertices and $B_1, \dots, B_q$ are blocks. Since $T$ is a suitable subtree and $P$ is obstruction free, for each $i\in [q]$, $(B_i, \{b_{i-1}, b_i\})$ is a suitable connector, where $B_i$ is an edge, diamond or impasse. Since $B_1, \dots, B_q$ are biconnected components, every path from $b_0$ to $b_q$ traverses $b_0, b_1 \dots, b_q$ in order. The shortest path from $b_0$ to $b_q$ is unique, if for each $i\in [k]$, the shortest path from $b_{i-1}$ to $b_i$ is unique. If $B_i$ is an edge or diamond, this is clearly the case since then the shortest path from $b_{i-1}$ to $b_i$ is an edge. If $B_i$ is an impasse then $(b_{i-1}, b_i)$ is a pair of connectors of $B_i$ and by Definition~\ref{def:impasse} there is no edge between $b_{i-1}$ and $b_i$, but there is a unique common neighbour $v$ of $b_{i-1}$ and $b_i$ in $H$ and consequently the unique shortest path from $b_{i-1}$ to $b_i$ is of the form $(b_{i-1},v,b_i)$, as required.
\end{proof}

\begin{defn}[attachment point, exit, destination]\label{def:exit}
	Let $H$ be a connected graph and let $T$ be a closed suitable subtree of $\bc(H)$.	
	Let $a$ be a cut vertex that has
	an obstruction~$B$ as a neighbour in~$T$. Then,
	since every cut vertex of $T$ has degree at most $2$, there is  a unique maximal-length proper obstruction-free path $P^*$ in $T$ starting at $a$. 
	Let $b$ be the other endpoint of $P^*$ (possibly $P^*=(a)$
	in which case $b=a$). 
	The vertex $a$ is an \emph{attachment point of $(T,B)$} if $b$ is a leaf in $T$. 
	Otherwise, $a$ is an \emph{exit of $(T,B)$}. In this case, $b$ is adjacent to a block $B'\neq B$ which is an obstruction. We say that $(b, B')$ is the \emph{destination} of $a$ in $T$.
\end{defn}

At the beginning of this section we outlined our plan to define a particular closed walk $W$. We chose the names in Definition~\ref{def:exit} since $W$ will \emph{exit} an obstruction when it encounters an exit, and it will then proceed towards the destination of that exit. The walk $W$ will not exit an obstruction when it encounters an attachment point. However, $W$ will be designed so that every even-cardinality walk-neighbour-set of $W$ contains an attachment point, and the structure that is \emph{attached} to such a point will allow us to  construct a hardness gadget.

\begin{defn}[concatenation ``$+$'']\label{def:plus}	
	Let $W=(w_0, \dots, w_k)$ and $W'=(w'_0, \dots, w'_{\ell})$ be two walks such that $w_k=w'_0$.  
	If $k=0$ then the \emph{concatenation} $W+W'$ of $W$ with $W'$ is equal to~$W'$.
	Similarly, if $\ell=0$, it is equal to~$W$.
	If both~$k$ and~$\ell$ are positive then  $W+W'=(w_0, \dots, w_{k-1},w_k,w'_1 \dots, w'_{\ell})$.	
\end{defn}

\begin{defn}[$D(C)$, $\cseg_C(a)$, $\cseg_C(a,b)$]\label{def:DofC}
	For an integer $q\ge 3$, let $C=(c_0, \dots, c_{q-1}, c_0)$ be a cycle in a graph $H$. Then $D(C)$ is the cyclic order induced by the order in which the walk $C$ traverses the vertices $\{c_0, \dots, c_{q-1}\}$. For $a\in C$, $\cseg_C(a)$ is the walk from $a$ to itself following all of the vertices of~$C$  in the order given by~$D(C)$.
	For $a,b\in C$, $\cseg_C(a,b)$ is the walk from $a$ to $b$ along~$C$ in the order given by~$D(C)$.
\end{defn}

\begin{algorithm}[H]
	\caption{{$\subwalk{T,a^*,B,\ell,a_0}$}}\label{alg:closedwalksub}
	\begin{algorithmic}
		\Require{A closed suitable subtree $T$ of $\bc(H)$ of a connected graph $H$, a cut vertex $a^*$ in $T$, an obstruction $B$ that is a block in $T$ such that $\dist_T(a^*,B)=\ell$, and an exit $a_0$ of $(T,B)$}
		\State $C \leftarrow C(B,\Nof{T}(B))$
		\State $\{a_0, \dots, a_k\} \leftarrow$ The exits of $(T,B)$ in the order of $D(C)$, starting from $a_0$
		\If {$k=0$}
		\State $W\leftarrow \cseg_C(a_0)$.
		\Else
		\State $\{(b_1, B_1), \dots, (b_k, B_k)\} \leftarrow$ The destinations of $a_1, \dots, a_k$, respectively	
		\State $W \leftarrow \cseg_C(a_0, a_1)$
		\For {$i=1, \dots, k$}
		\State $r_i \leftarrow \dist_T(B,b_i)$
		\State  \mbox{$W\leftarrow W + \SP{a_i, b_i} + \subwalk{T,a^*, B_i, \ell+r_i+1, b_i} + \SP{b_i, a_i} + \cseg_C(a_i, a_{i+1 \mod k+1})$}
		\EndFor
		\EndIf
		\Ensure{$W$}
	\end{algorithmic}
\end{algorithm}

\begin{algorithm}[H]
	\caption{{$\walk{T,B'}$}}\label{alg:closedwalk}
	\begin{algorithmic}
		\Require{A closed suitable subtree $T$ of $\bc(H)$ of a connected graph $H$, an obstruction $B'$ that is a block in $T$}
		\If{there is an exit $a^*$ of $(T,B')$}
		\State $(b^*, B^*) \leftarrow$ The destination of $a^*$	
		\State $r^* \leftarrow \dist_T(a^*, b^*)$
		\State \mbox{$W \leftarrow 
			\subwalk{T,a^*, B',1,a^*} +		
			\SP{a^*, b^*} + \subwalk{T, a^*, B^*, r^*+1, b^*} + \SP{b^*, a^*}
			$}
		\Else
		\State $W\leftarrow C(B',\Nof{T}(B'))$
		\EndIf
		\Ensure{$W$}
	\end{algorithmic}
\end{algorithm}

In Figure~\ref{fig:closedwalk} we provide some illustrations of 
a graph~$H$, a closed suitable subtree $T\in \bc(H)$, and the walk~$W$ returned by Algorithm~\ref{alg:closedwalk}.
In order to gain intuition, it is probably useful to simulate 
$\walk{T,O_1}$.
The exit~$a^*$ can be chosen to be~$a_2$ with destination $(b^*,B^*)=(b_2,O_2)$.
The variable~$r^*$ is set to~$0$. So the first part of $W$ is the output of the 
call   $ \subwalk{T,a_2, O_1,1,a_2}$.

Let's start by considering that call.
$\Gamma_T(O_1) = \{t_1,a_1,a_2\}$
and $C$ is the cycle around~$O_1$ shown in red.
The exits are  $\{a_2,a_1\}$  so 
the output $W$ of this call starts by following the red cycle clockwise from $a_2$ to~$a_1$.
In the else-clause we have $k=1$
and the destination of $a_1$ is $(b_1,O_3)$. 
The walk next takes the unique shortest path from~$a_1$ to~$b_1$.
Then there is a call to 
$\subwalk{T,a_2, O_3, \ell, b_1}$, for some value of~$\ell$ (the value of $\ell$ doesn't matter --- it is just for accounting).
The only exit of $(T, O_3)$ is~$b_1$, so this call  returns a walk around the red cycle in~$O_3$
from~$b_1$ to itself. 
Finally, the call to $ \subwalk{T,a_2, O_1,1,a_2}$
takes the unique shortest path back from~$b_1$ to~$a_1$ and finishes the red cycle in~$O_1$ clock-wise, back to~$a_2$.
Thus, the output of $ \subwalk{T,a_2, O_1,1,a_2}$
is a closed walk from~$a_2$ to itself that covers all of the red edges in the picture, apart from the triangle in~$O_2$.

This is concatenated with
$P_H(a_2,b_2) = (a_2)$ (which does nothing).
Then it is concatenated with the output of  a call to 
$\subwalk{T,a_2, O_2,1,a_2}$.
Now $(O_2,\{a_2\})$ 
is a suitable connector in~$H$ with $C(O_2,\{a_2\})$ equal to the red triangle in~$O_2$, so $C$ is assigned to be this triangle.
The cut-vertex~$a_2$ is the only exit of~$C$,
so this call returns the walk  from~$a_2$ to itself around~$C$.
Concatenating $P_H(a_2,b_2)$ with this does not change the output.
The entire walk is coloured in red.

\begin{figure}
	\vspace{-1cm}
	\centering
	\begin{tikzpicture}[xscale=1, yscale=0.75, node distance = 1.4cm,thick]
	\tikzstyle{dot}   =[fill=black, draw=black, circle, inner sep=0.15mm]
	\tikzstyle{vertex}=[  draw=black, circle, inner sep=1.5pt]
	\tikzstyle{apoint}=[  draw=black, circle, fill, inner sep=1.5pt]
	\tikzstyle{dist}  =[fill=white, draw=black, circle, inner sep=2pt]
	\tikzstyle{pinned}=[draw=black, minimum size=10mm, circle, inner sep=0pt]	
	\node[vertex] (a) at (0  ,2) {};
	\node[vertex] (b) at (1  ,3) {};
	\node[apoint] (c) at (2  ,3.5) [label=270:\tiny $c_1$] {};
	\node[vertex] (d) at (2  ,2.5) {};
	\node[vertex] (e) at (3  ,3) {};
	\node[vertex] (f) at (4  ,2) {};
	\node[apoint] (g) at (4.5  ,1) [label=180:\tiny $c_2$] {};
	\node[vertex] (h) at (3.5  ,1) {};
	\node[vertex] (i) at (4  ,0) {};
	\node[vertex] (j) at (3  ,-1) {};
	\node[vertex] (k) at (2  ,-0.5) {};
	\node[apoint] (l) at (2  ,-1.5) [label=90:\tiny $c_3$] {};
	\node[vertex] (m) at (1  ,-1) {};
	\node[vertex] (n) at (0  ,0) {};
	
	\draw (a) -- (b); \draw (c) -- (b); \draw (d) -- (b);
	\draw (c) -- (e); \draw (d) -- (e); \draw (e) -- (f);
	\draw (f) -- (g); \draw (f) -- (h); \draw (h) -- (i);
	\draw (g) -- (i); \draw (i) -- (j); \draw (j) -- (k);
	\draw (j) -- (l); \draw (k) -- (m); \draw (l) -- (m);
	\draw (m) -- (n); \draw (a) -- (n); 
	
	\node[apoint] (o) at (5.5  ,1) [label=225:\tiny $c_4$] {};
	\node[apoint] (p) at (5.5  ,2)  {};
	\node[vertex] (q) at (4.5  ,2) {};
	\node[vertex] (r) at (5  ,1.5) {};
	\node[apoint] (s) at (6.5  ,1)  [label=-45:\tiny $c_5$] {};
	\node[vertex] (t) at (7.5  ,1) {};
	\node[apoint] (u) at (8.5  ,1) [label=225:\tiny $c_{6}$] {};
	\node[vertex] (v) at (8.5  ,2) {};
	\node[vertex] (w) at (7.5  ,2) {};
	\node[vertex] (x) at (6.5  ,2) {};
	\node[vertex] (y) at (7  ,1.5) {};
	\node[vertex] (z) at (8  ,1.5) {};
	
	\draw (g) -- (o); \draw (o) -- (p); \draw (p) -- (q);
	\draw (q) -- (r); \draw (q) -- (g); \draw (r) -- (o);
	\draw (u) -- (t); \draw (t) -- (s); \draw (o) -- (s);
	\draw (u) -- (v); \draw (v) -- (w); \draw (w) -- (x);
	\draw (t) -- (y); \draw (y) -- (x); \draw (s) -- (x);
	\draw (t) -- (z); \draw (z) -- (v); \draw (t) -- (w);
	
	\node[apoint] (1) at (9.5  ,1) [label=0:\tiny $c_{7}$] {};
	\node[vertex] (2) at (10  ,2) {};
	\node[apoint] (3) at (11  ,3) {};
	\node[apoint] (4) at (13  ,3)  {};
	\node[apoint] (5) at (12  ,2) {};
	\node[vertex] (6) at (13  ,2) {};
	\node[apoint] (7) at (14  ,2)  {};
	\node[vertex] (8) at (13  ,1) {};
	\node[vertex] (9) at (13  ,0) {};
	\node[vertex] (10) at (12  ,0.5) {};
	\node[apoint] (11) at (12  ,-0.5) [label=90:\tiny $c_{8}$] {};
	\node[vertex] (12) at (11  ,0) {};
	\node[vertex] (13) at (10  ,0) {};
	\node[vertex] (14) at (10.5  ,1) {};
	
	\draw (u) -- (1); \draw (1) -- (2); \draw (2) -- (3);
	\draw (6) -- (4); \draw (5) -- (4); \draw (4) -- (3);
	\draw (4) -- (7); \draw (5) -- (8); \draw (6) -- (8);
	\draw (7) -- (8); \draw (8) -- (9); \draw (9) -- (10);
	\draw (12) -- (11); \draw (12) -- (10); \draw (11) -- (9);
	\draw (12) -- (13); \draw (13) -- (14); \draw (14) -- (2);
	\draw (13) -- (1);

	\node[vertex] (c1) at (2  ,4.5) {};
	\node[vertex] (l1) at (1  ,-2) {};
	\node[apoint] (l2) at (3  ,-2)   {};
	\node[vertex] (l3) at (2  ,-2.5) {};
	\node[vertex] (111) at (12.666  ,-0.5) {};
	\node[vertex] (112) at (12.666  ,-1.5) {};
	\node[apoint] (113) at (12  ,-1.5) [label=135:\tiny $c_{9}$] {};
	\node[apoint] (114) at (12  ,-2.5) [label=225:\tiny $c_{10}$] {};
	\node[apoint] (115) at (12.666  ,-2.5) {};
	\node[vertex] (116) at (12.333  ,-3) {};
	\node[vertex] (117) at (12.666  ,-3.5) {};
	\node[apoint] (118) at (12  ,-3.5) [label=180:\tiny $c_{11}$] {};
	\node[vertex] (119) at (12  ,-4.5) {};
	
	\draw (c) -- (c1); \draw (l1) -- (l); \draw (l2) -- (l);
	\draw (11) -- (111); \draw (l2) -- (l3); \draw (l1) -- (l3);
	\draw (11) -- (113); \draw (111) -- (112); \draw (112) -- (113);
	\draw (116) -- (114); \draw (115) -- (114); \draw (114) -- (113);
	\draw (114) -- (118); \draw (118) -- (117); \draw (117) -- (116);
	\draw (119) -- (118); \draw (115) -- (117); \draw (l2) -- (l1);

	\node[vertex] (l4) at (3  ,-3) {};
	\node[apoint] (l5) at (4  ,-3) {};
	\node[apoint] (l6) at (4  ,-2) {};
	\node[vertex] (l7) at (5  ,-2) {};
	\node[vertex] (l8) at (3.5  ,-2.5) {};

	\draw (l2) -- (l4); \draw (l2) -- (l6); \draw (l4) -- (l8);
	\draw (l5) -- (l6); \draw (l7) -- (l6); \draw (l6) -- (l8);
	\draw (l4) -- (l5); 
	
	\node[vertex] (p1) at (5.5  ,3) {};
	\draw (p) -- (p1); 
	
	\node[vertex] (n1) at (1  ,0) {};
	\node[vertex] (a1) at (1  ,2) {};
	\draw (n) -- (n1); \draw (a) -- (a1);
	\draw (a1) -- (n1);

	\node[vertex] (3a) at (11  ,4) {};
	\node[vertex] (4a) at (13  ,4) {};
	\node[vertex] (7a) at (14  ,3) {};
	\node[vertex] (5b) at (11  ,2) {};
	\node[apoint] (5a) at (11.5  ,2) {};
	\draw (3) -- (3a); \draw (4) -- (4a); \draw (5) -- (5a);
	\draw (5a) -- (5b); \draw (7) -- (7a);
	
	\node[vertex] (l9) at (5  ,-3) {};
	\node[vertex] (1111) at (14  ,-2.5) {};
	\node[apoint] (1110) at (13.333  ,-2.5) {};
	\draw (l5) -- (l9); \draw (115) -- (1110); \draw (1110) -- (1111);
	
	\node[vertex] (no2) at (5.5  ,0) {};
	\node[vertex] (ns2) at (6.5  ,0) {};
	\node[vertex] (nu2) at (8.5  ,0) {};
	\node[vertex] (1132) at (11  ,-1.5) {};
	\node[vertex] (1142) at (11  ,-2.5) {};
	\draw (o) -- (no2); \draw (ns2) -- (s); \draw (nu2) -- (u);
	\draw (113) -- (1132); \draw (114) -- (1142);
	
	\node (o1) at (2,1) {$O_1$};
	\node (e2) at (1.5,4) {$E_1$};
	\node (o2) at (2,-3) {$O_2$};
	\node (d1) at (5,2.5) {$D_1$};
	\node (i1) at (7.5,2.5) {$I_1$};
	\node (e2) at (6,0.5) {$E_2$};
	\node (e3) at (9,0.5) {$E_3$};
	\node (o3) at (11.5,1) {$O_3$};
	\node (d2) at (13,-1) {$D_2$};
	\node (e4) at (11.5,-2) {$E_4$};
	\node (e4) at (13,-3) {$D_3$};
	\node (e5) at (11.5,-4) {$E_5$};
	
	\begin{scope}[yshift=-7cm, xshift=-1cm,xscale=0.75]
	\node[dist] (to1) at (2,0) {$O_1$};
	\node[apoint] (tc1) at (3  ,1.5) [label=90:\tiny $c_{1}$] {};
	\node[apoint] (tc2) at (3  ,0) [label=90:\tiny $c_{2}$] {};
	\node[dist] (td1) at (4,0) {$D_1$};
	\node[apoint] (tc3) at (3  ,-1.5) [label=270:\tiny $c_{3}$] {};
	\node[dist] (to2) at (4,-1.5) {$O_2$};
	\node[apoint] (tc4) at (5  ,0) [label=90:\tiny $c_{4}$] {};
	\node[dist] (te2) at (6,0) {$E_2$};
	\node[apoint] (tc5) at (7  ,0) [label=90:\tiny $c_{5}$] {};
	\node[dist] (ti1) at (8,0) {$I_1$};
	\node[apoint] (tc6) at (9  ,0) [label=90:\tiny $c_{6}$] {};
	\node[dist] (te3) at (10,0) {$E_3$};
	\node[apoint] (tc7) at (11  ,0) [label=90:\tiny $c_{7}$] {};
	\node[dist] (to3) at (12,0) {$O_3$};
	
	\node[apoint] (tc8) at (13  ,0) [label=90:\tiny $c_{8}$] {};
	\node[dist] (td2) at (14,0) {$D_2$};
	\node[apoint] (tc9) at (15  ,0) [label=90:\tiny $c_{9}$] {};
	\node[dist] (te4) at (16,0) {$E_4$};
	\node[apoint] (tc10) at (17  ,0) [label=90:\tiny $c_{10}$] {};
	\node[dist] (td3) at (18,0) {$D_3$};
	\node[apoint] (tc11) at (19  ,0) [label=90:\tiny $c_{11}$] {};

	\draw (to1) -- (tc1); 
	\draw (to1) -- (tc2); \draw (tc2) -- (td1);
	\draw (to1) -- (tc3); \draw (tc3) -- (to2);
	
	\draw (td1) -- (tc4); \draw (tc4) -- (te2);
	\draw (te2) -- (tc5); \draw (tc5) -- (ti1);
	\draw (ti1) -- (tc6); \draw (tc6) -- (te3);
	\draw (te3) -- (tc7); \draw (tc7) -- (to3);
	\draw (to3) -- (tc8); \draw (tc8) -- (td2);
	\draw (td2) -- (tc9); \draw (tc9) -- (te4);
	\draw (te4) -- (tc10); \draw (tc10) -- (td3);
	\draw (td3) -- (tc11); 
	\end{scope}
	
	\begin{scope}[yshift=-15cm]
	\node[vertex] (a) at (0  ,2) {};
	\node[vertex] (b) at (1  ,3) {};
	\node[apoint] (c) at (2  ,3.5) [label=270:\tiny $t_1$] {};
	\node[vertex] (d) at (2  ,2.5) {};
	\node[vertex] (e) at (3  ,3) {};
	\node[vertex] (f) at (4  ,2) {};
	\node[apoint] (g) at (4.5  ,1) [label=180:\tiny $a_1$] {};
	\node[vertex] (h) at (3.5  ,1) {};
	\node[vertex] (i) at (4  ,0) {};
	\node[vertex] (j) at (3  ,-1) {};
	\node[vertex] (k) at (2  ,-0.5) {};
	\node[apoint] (l) at (2  ,-1.5) [label=90:\tiny ${a_2=b_2}$] {};
	\node[vertex] (m) at (1  ,-1) {};
	\node[vertex] (n) at (0  ,0) {};
	
	\draw[thick,red] (a) -- (b); \draw[thick,red] (c) -- (b); \draw (d) -- (b);
	\draw[thick,red] (c) -- (e); \draw (d) -- (e); \draw[thick,red] (e) -- (f);
	\draw[thick,red] (f) -- (g); \draw (f) -- (h); \draw (h) -- (i);
	\draw[thick,red] (g) -- (i); \draw[thick,red] (i) -- (j); \draw (j) -- (k);
	\draw[thick,red] (j) -- (l); \draw (k) -- (m); \draw[thick,red] (l) -- (m);
	\draw[thick,red] (m) -- (n); \draw[thick,red] (a) -- (n); 
	
	\node[apoint] (o) at (5.5  ,1)  {};
	\node[apoint] (p) at (5.5  ,2)  {};
	\node[vertex] (q) at (4.5  ,2) {};
	\node[vertex] (r) at (5  ,1.5) {};
	\node[apoint] (s) at (6.5  ,1)   {};
	\node[vertex] (t) at (7.5  ,1) {};
	\node[apoint] (u) at (8.5  ,1) {};
	\node[vertex] (v) at (8.5  ,2) {};
	\node[vertex] (w) at (7.5  ,2) {};
	\node[vertex] (x) at (6.5  ,2) {};
	\node[vertex] (y) at (7  ,1.5) {};
	\node[vertex] (z) at (8  ,1.5) {};
	
	\draw[thick,red] (g) -- (o); \draw (o) -- (p); \draw (p) -- (q);
	\draw (q) -- (r); \draw (q) -- (g); \draw (r) -- (o);
	\draw[thick,red] (u) -- (t); \draw[thick,red] (t) -- (s); 
	\draw[thick,red] (o) -- (s);
	\draw (u) -- (v); \draw (v) -- (w); \draw (w) -- (x);
	\draw (t) -- (y); \draw (y) -- (x); \draw (s) -- (x);
	\draw (t) -- (z); \draw (z) -- (v); \draw (t) -- (w);
	
	\node[apoint] (1) at (9.5  ,1) [label=0:\tiny $b_1$] {};
	\node[vertex] (2) at (10  ,2) {};
	\node[apoint] (3) at (11  ,3) {};
	\node[apoint] (4) at (13  ,3)  {};
	\node[apoint] (5) at (12  ,2) {};
	\node[vertex] (6) at (13  ,2) {};
	\node[apoint] (7) at (14  ,2)  {};
	\node[vertex] (8) at (13  ,1) {};
	\node[vertex] (9) at (13  ,0) {};
	\node[vertex] (10) at (12  ,0.5) {};
	\node[apoint] (11) at (12  ,-0.5) [label=90:\tiny $t_2$] {};
	\node[vertex] (12) at (11  ,0) {};
	\node[vertex] (13) at (10  ,0) {};
	\node[vertex] (14) at (10.5  ,1) {};
	
	\draw[thick,red] (u) -- (1); \draw[thick,red] (1) -- (2); 
	\draw[thick,red] (2) -- (3);
	\draw[thick,red] (6) -- (4); \draw (5) -- (4); \draw[thick,red] (4) -- (3);
	\draw (4) -- (7); \draw (5) -- (8); \draw[thick,red] (6) -- (8);
	\draw (7) -- (8); \draw[thick,red] (8) -- (9); \draw (9) -- (10);
	\draw[thick,red] (12) -- (11); \draw (12) -- (10); \draw[thick,red] (11) -- (9);
	\draw[thick,red] (12) -- (13); \draw (13) -- (14); \draw (14) -- (2);
	\draw[thick,red] (13) -- (1);

	\node[vertex] (c1) at (2  ,4.5) {};
	\node[vertex] (l1) at (1  ,-2) {};
	\node[apoint] (l2) at (3  ,-2)   {};
	\node[vertex] (l3) at (2  ,-2.5) {};
	\node[vertex] (111) at (12.666  ,-0.5) {};
	\node[vertex] (112) at (12.666  ,-1.5) {};
	\node[apoint] (113) at (12  ,-1.5)  {};
	\node[apoint] (114) at (12  ,-2.5)  {};
	\node[apoint] (115) at (12.666  ,-2.5) {};
	\node[vertex] (116) at (12.333  ,-3) {};
	\node[vertex] (117) at (12.666  ,-3.5) {};
	\node[apoint] (118) at (12  ,-3.5)  {};
	\node[vertex] (119) at (12  ,-4.5) {};
	
	\draw[dashed] (c) -- (c1); \draw[thick,red] (l1) -- (l); \draw[thick,red] (l2) -- (l);
	\draw (11) -- (111); \draw (l2) -- (l3); \draw (l1) -- (l3);
	\draw (11) -- (113); \draw (111) -- (112); \draw (112) -- (113);
	\draw (116) -- (114); \draw (115) -- (114); \draw (114) -- (113);
	\draw (114) -- (118); \draw (118) -- (117); \draw (117) -- (116);
	\draw[dashed] (119) -- (118); \draw (115) -- (117); \draw[thick,red] (l2) -- (l1);

	\node[vertex] (l4) at (3  ,-3) {};
	\node[apoint] (l5) at (4  ,-3) {};
	\node[apoint] (l6) at (4  ,-2) {};
	\node[vertex] (l7) at (5  ,-2) {};
	\node[vertex] (l8) at (3.5  ,-2.5) {};
	
	\draw[dashed] (l2) -- (l4); \draw[dashed] (l2) -- (l6); \draw[dashed] (l4) -- (l8);
	\draw[dashed] (l5) -- (l6); \draw[dashed] (l7) -- (l6); \draw[dashed] (l6) -- (l8);
	\draw[dashed] (l4) -- (l5); 
	
	\node[vertex] (p1) at (5.5  ,3) {};
	\draw[dashed] (p) -- (p1); 
	
	\node[vertex] (n1) at (1  ,0) {};
	\node[vertex] (a1) at (1  ,2) {};
	\draw (n) -- (n1); \draw (a) -- (a1);
	\draw (a1) -- (n1);

	\node[vertex] (3a) at (11  ,4) {};
	\node[vertex] (4a) at (13  ,4) {};
	\node[vertex] (7a) at (14  ,3) {};
	\node[vertex] (5b) at (11  ,2) {};
	\node[apoint] (5a) at (11.5  ,2) {};
	\draw[dashed] (3) -- (3a); \draw[dashed] (4) -- (4a); \draw[dashed] (5) -- (5a);
	\draw[dashed] (5a) -- (5b); \draw[dashed] (7) -- (7a);
	\node (o1) at (2,1) {$O_1$};
	\node (o2) at (2,-3) {$O_2$};
	\node (o3) at (11.5,1) {$O_3$};
	
	\node[vertex] (l9) at (5  ,-3) {};
	\node[vertex] (1111) at (14  ,-2.5) {};
	\node[apoint] (1110) at (13.333  ,-2.5) {};
	\draw[dashed] (l5) -- (l9); \draw[dashed] (115) -- (1110); \draw[dashed] (1110) -- (1111);
	
	\node[vertex] (no2) at (5.5  ,0) {};
	\node[vertex] (ns2) at (6.5  ,0) {};
	\node[vertex] (nu2) at (8.5  ,0) {};
	\node[vertex] (1132) at (11  ,-1.5) {};
	\node[vertex] (1142) at (11  ,-2.5) {};
	\draw[dashed] (o) -- (no2); \draw[dashed] (ns2) -- (s); \draw[dashed] (nu2) -- (u);
	\draw[dashed] (113) -- (1132); \draw[dashed] (114) -- (1142);
	\end{scope}
	\end{tikzpicture} 
	\caption[caption]{
		\textit{A graph $H$, a closed suitable subtree $T$ of $\bc(H)$ with a block $O_1$ that is an obstruction, and $\walk{T,O_1}$.}\\
		\emph{(Top)} An involution-free and $K_4$-minor-free graph $H$ such that every biconnected component is an edge, a diamond, an impasse or an obstruction. Articulation points are depicted as filled vertices.\\
		\emph{(Center)} A \emph{closed} and \emph{suitable} subtree $T$ of the block-cut tree of $H$, rooted at $O_1$. Note that every cut vertex of $T$ that is a leaf (i.e., $c_1$ or $c_{11}$) is $O_1$-closed in $\bc(H)$.\\
		\emph{(Bottom)} Solid lines are contained in the subgraph of $H$ induced by $V(T)$, while dashed lines are not. The red closed walk is the output of $\walk{T,O_1}$. Observe that $a_1$ and $a_2$ are exits of $(T,O_1)$ with destinations $(b_1,O_3)$ and $(b_2,O_2)$, respectively, and that $t_1$ and $t_2$ are attachment points of $(T,O_1)$ and $(T,O_3)$, respectively.}
	\label{fig:closedwalk}
\end{figure}

We now proceed to establish the correctness of Algorithms~\ref{alg:closedwalksub} and~\ref{alg:closedwalk}, and to 
prove some properties of the walks that they output.

\begin{lem}\label{lem:feasibleinput}
	All calls to $\subwalk{\cdot}$ in Algorithms~\ref{alg:closedwalksub} and~\ref{alg:closedwalk} have arguments that are  feasible inputs to Algorithm~\ref{alg:closedwalksub}.
\end{lem}
\begin{proof}
	First, consider Algorithm~\ref{alg:closedwalksub}, where for $i\in [k]$ we make a call $\subwalk{T,a^*, B_i, \ell+r_i+1, b_i}$. Observe that $b_i$ is an exit of $(T,B_i)$ by the definition of a destination (Definition~\ref{def:exit}). It remains to check that $\dist_T(a^*, B_i)=\ell+r_i+1$. This is true since $\ell=\dist_T(a^*,B)$ and $r_i = \dist_T(B,b_i)$ using the fact that the (unique) path from $a^*$ to $B_i$ in the tree $T$ goes from $a^*$ to $B$ then from $B$ to $b_i$ and then from $b_i$ to $B_i$, where $B_i$ is adjacent to $b_i$.
	
	Second, consider Algorithm~\ref{alg:closedwalk}, where the if-block makes two calls to $\subwalk{\cdot}$ --- one is 
	$\subwalk{T,a^*, B',1,a^*}$ and the other is
	$\subwalk{T, a^*, B^*, r^*+1, b^*}$.
	Observe that 
	$a^*$ is an exit of $(T,B')$ by the condition of the if-block and
	$b^*$ is an exit of $(T,B^*)$ by the definition of a destination.
	It remains to check that $\dist_T(a^*,B')=1$ and $\dist_T(a^*, B^*)=r^*+1$. The former is immediate since $a^*$ is adjacent to $B'$ in $T$. The latter is true since $r^*=\dist_T(a^*,b^*)$ and $B^*$ is adjacent to $b^*$ in $T$ where the (unique) path from $a^*$ to $B^*$ goes via $b^*$.
\end{proof}

\begin{lem}\label{lem:closedwalksub}
	$\subwalk{T,a^*,B,\ell,a_0}$ (Algorithm~\ref{alg:closedwalksub}) terminates, is well-defined, and returns a closed walk 
	in~$H$ of length at least $3$ from $a_0$ to itself.
\end{lem}
\begin{proof}
	First consider the case $k=0$. Clearly, Algorithm~\ref{alg:closedwalksub} terminates and is well-defined. It returns $\cseg_C(a_0)$, which is a cycle from $a_0$ to itself of length at least $3$.
	Now consider the case where $k\ge 1$.
	Note that with each recursive call of $\subwalk{\cdot}$ the value of the parameter $\ell$ increases. Since $\ell$ corresponds to the distance between $a^*$ and $B$  in the finite graph~$T$,
	Algorithm~\ref{alg:closedwalksub} terminates.
	
	We now show that Algorithm~\ref{alg:closedwalksub} returns a closed walk of length at least $3$ from $a_0$ to itself.
	If, for $i\in [k]$, $\subwalk{T,a^*, B_i, \ell+r_i+1, b_i}$ returns a closed walk from $b_i$ to itself of length at least $3$ then  
	$\SP{a_i, b_i} + \subwalk{T, a^*,B_i, \ell+r_i+1,b_i} + \SP{b_i, a_i} + \cseg_C(a_i, a_{i+1 \mod k+1})$ is a walk from $a_i$ to $a_{i+1 \mod {k+1}}$ of length at least $3$. Thus, $\subwalk{T,a^*,B,\ell,a_0}$ returns a closed walk from $a_0$ to itself of length at least $3$. Since Algorithm~\ref{alg:closedwalksub} terminates, it reaches the base of the recursion, i.e.,~the case $k=0$, at some point, and we have already verified that the base case returns a closed walk of length at least $3$, as required.
	
	Finally, we show that Algorithm~\ref{alg:closedwalksub} is well-defined. By Lemma~\ref{lem:feasibleinput} all subroutine calls have feasible inputs. Also observe that all concatenation operations are well-defined since, for each $i\in [k]$, $\subwalk{T,a^*, B_i, \ell+r_i+1, b_i}$ returns a closed walk from $b_i$ to itself.
\end{proof}

\begin{lem}\label{lem:closedwalk}
	$\walk{T,B'}$	(Algorithm~\ref{alg:closedwalk}) terminates, is well-defined, and returns a closed walk in $H$ of length $q$, where $q\ge 3$ and $q\neq 4$.
\end{lem}
\begin{proof}
	Since Algorithm~\ref{alg:closedwalksub} terminates (Lemma~\ref{lem:closedwalksub}), it is immediate that Algorithm~\ref{alg:closedwalk} terminates.
	If there is no exit of $(T,B')$ then   $\walk{T,B'}$   returns $C(B',\Nof{T}(B'))$ which is a cycle in $\cycles(B')$ 
	by Definition~\ref{def:suitableconnector} 
	and hence has length at least~$3$, but 
	not~$4$, by Definition~\ref{def:obstruction}.
	If there is an exit $a^*$ of $(T,B')$ then, by Lemma~\ref{lem:feasibleinput} all subroutine calls have feasible inputs. 
	By Lemma~\ref{lem:closedwalksub}, 
	$\subwalk{T,a^*, B',1,a^*}$ returns a closed walk from $a^*$ to itself,  of length at least $3$, and
	$\subwalk{T, a^*, B^*, r^*+1, b^*}$ returns a closed walk from $b^*$ to itself, also of length at least $3$.
	It follows that the concatenations in the if-block are well-defined and therefore that Algorithm~\ref{alg:closedwalk} is well-defined. Furthermore,
	$  \subwalk{T,a^*, B',1,a^*} +		
	\SP{a^*, b^*} + \subwalk{T, a^*, B^*, r^*+1, b^*} + \SP{b^*, a^*}$
	is a closed walk from $a^*$ to itself of length $q\ge6$.
\end{proof}

\begin{obs}\label{obs:alternatingwalk} 
	$\walk{T,B'}$
	(Algorithm~\ref{alg:closedwalk}) outputs a closed walk $W$.
	If $(T,B')$ has no exit then $W=C$ for a   cycle   $C\in \cycles(B')$. Otherwise,
	the following holds.
	For a positive integer~$j$, there are obstructions $B'_0, \dots, B'_j$ such that 
	$W$ is of the form  
	$W=Q_0+P_0+ Q_1+ P_1 \dots+Q_j+ P_j$  
	where $Q_i$ and $P_i$ satisfy the following properties for all $i\in \{0,\ldots,j\}$.	 	\begin{myitemize}
		\item  Let $C_i=C(B'_i,\Nof{T}(B'_i))$.
		Then there are vertices $a$ and $a'$ in~$B'_i$ such that
		$Q_i$ is of the form $\cseg_{C_i}(a)$
		or $\cseg_{C_i}(a,a')$. Either way,  all vertices of $Q_i$ are in $B'_i$. Furthermore, only the endpoints of $Q_i$ are exits of $(T,B'_i)$. 
		\item  There is an exit~$a$ of $(T,B'_i)$ 
		with a destination $(b,B'_{i+1 \mod (j+1)})$
		such that
		$P_i$ is a path of the form $\SP{a,b}$. Hence, by Definition~\ref{def:exit} and Lemma~\ref{lem:obstructionfreepath},
		the endpoints 
		of~$P_i$
		are the only vertices of $P_i$ that are part of an obstruction.
		\item The obstruction $B'_{i}$ is distinct from the obstruction  $B'_{i+1 \mod (j+1)}$.
	\end{myitemize}
\end{obs}

\noindent{\it Explanation of Observation~\ref{obs:alternatingwalk}.}
	If $(T,B')$ has no exit then 
	the result follows directly from  the definition of $C(B,A)$ for a suitable connector $(B,A)$ of an obstruction~$B$ (Definition~\ref{def:suitableconnector}).
	
	For the remaining case, we will prove the following about Algorithm~\ref{alg:closedwalksub}.
	$\subwalk{T,a^*,B,\ell,a_0}$  outputs a closed walk $W'$.
	For a positive integer~$q$, there are obstructions $B''_0, \dots, B''_q$
	such that $W'$ is of the form  
	$W'=Q'_0+\sum_{i=1}^q (P'_i + Q'_i ) $  
	where $Q'_i$ and $P'_i$ satisfy the following properties for all $i\in \{0,\ldots,q\}$.	 
	\begin{myitemize}
		\item  Let $C_i=C(B''_i,\Nof{T}(B''_i))$.
		Then there are vertices $a$ and $a'$ in~$B''_i$ such that
		$Q'_i$ is of the form $\cseg_{C_i}(a)$
		or $\cseg_{C_i}(a,a')$. Either way,  all vertices of $Q'_i$ are in $B''_i$. Furthermore, only the endpoints of $Q'_i$ are exits of $(T,B''_i)$. 
		\item  There is an exit~$a$ of $(T,B''_i)$ 
		with a destination $(b,B''_{i+1 \mod (j+1)})$
		such that
		$P'_i$ is a path of the form $\SP{a,b}$. Hence, by Definition~\ref{def:exit},
		the  endpoints 
		of~$P'_i$
		are the only vertices of $P'_i$ that are part of an obstruction.
		\item The obstruction $B''_{i}$ is distinct from the obstruction  $B''_{i+1 \mod (q+1)}$.
	\end{myitemize}
	The proof is by induction on the recursion depth.
	If $\subwalk{T,a^*,B,\ell,a_0}$ makes no recursive calls then
	the variable ``$k$'' is equal to~$0$ and we also set $q=0$. In this case, $B''_0$ is equal to $B$, and the first property follows easily (the others are vacuous).
	Otherwise, $k$ is positive and $(T,B)$ has exits $\{a_0,\ldots,a_k\}$ where $a_1,\ldots,a_k$
	have destinations $(b_1,B_1),\ldots,(b_k,B_k)$. 
	Note that $B_1,\ldots,B_k$ are disjoint from~$B$ and from each other.
	Once again,  $B''_0$ is $B$.
	$Q'_0$ is $W_C(a_0,a_1)$, as defined in the algorithm. Then $B''_1$ is $B_1$, and $P'_1$ is $P_H(a_1,b_1)$ as defined in the algorithm.
	The rest follows by induction, and examination of the algorithm, using
	the fact that the block-cut tree is a tree. 
	
	Given this fact for  Algorithm~\ref{alg:closedwalksub}, we obtain the conclusion for 
	Algorithm~\ref{alg:closedwalk} by putting together the pieces in the output~$W$.
	This completes our explanation of Observation~\ref{obs:alternatingwalk}.

\begin{lem}\label{lem:distinctwiandwi+2} 	Let $H$ be a connected graph. Let $T$ be a closed suitable subtree of $\bc(H)$. Let $B'$ be an obstruction that is a block of $T$. Let $W=(w_0,,\dots, w_{q-1}, w_0)$ be the output of $\walk{T,B'}$ (Algorithm~\ref{alg:closedwalk}).
	Then, for each $i\in \{0, \dots, q-1\}$, $w_i$ and $w_{i+2 \mod q}$ are distinct. 
\end{lem}
\begin{proof}
	All indices in this proof are considered to be modulo $q$. For any $i\in \{0, \dots, q-1\}$, our goal is to show $w_i\neq w_{i+2}$.
	We make a case distinction based on Observation~\ref{obs:alternatingwalk}.
	\begin{itemize}
		\item If $W$ is a cycle $C\in \cycles(B')$ then $w_i\neq w_{i+2}$ is immediate.
		\item Otherwise, for a positive integer~$j$, $W$ is of the form
		$W=Q_0+P_0+ Q_1+ P_1 \dots+Q_j+ P_j$ with the properties stated in Observation~\ref{obs:alternatingwalk}. 
		We consider the walk $(w_i, w_{i+1}, w_{i+2})$.
		\begin{myitemize}
			\item If for some $\ell\in [j]$, $(w_i, w_{i+1}, w_{i+2})$ is a subwalk of $Q_\ell$ then $w_i \neq w_{i+2}$ since, by Observation~\ref{obs:alternatingwalk},
			$Q_i$ is a subwalk of a cycle.
			\item If for some $\ell\in [j]$, $(w_i, w_{i+1}, w_{i+2})$ is a subwalk of $P_\ell$ then, since $P_\ell$ is a path, we have $w_i\neq w_{i+2}$.
			\item Otherwise, by Observation~\ref{obs:alternatingwalk}, there is no biconnected component that contains both $w_i$ and $w_{i+2}$ and consequently $w_i \neq w_{i+2}$. 
		\end{myitemize}
	\end{itemize} 
\end{proof}

The following lemma establishes (a stronger version of) the properties~\eqref{equ:CycleHardness+2} and~\eqref{equ:CycleHardness+3} for the walk returned by Algorithm~\ref{alg:closedwalk}, as required by Lemma~\ref{lem:CycleHardness+}.

\begin{lem}\label{lem:properWalkIntersection}
	Let $H$ be a connected $K_4$-minor-free graph. Let $T$ be a closed suitable subtree of $\bc(H)$. Let $B'$ be an obstruction that is a block of $T$. Let $W=(w_0,\dots, w_{q-1}, w_0)$ be the output of $\walk{T,B'}$ (Algorithm~\ref{alg:closedwalk}). By Lemma~\ref{lem:closedwalk}, $W$ is a closed walk and $q\ge 3$. For each $i\in \{0, \dots, q-1\}$, let $W_i=N_{W,H}(w_i)$ (Definition~\ref{def:walkneighbourset}). If $H$ has no hardness gadget then the following statement holds:
	\begin{equation*}
	\text{If $u\in W_{i-1 \bmod q}$ and $v\in W_{i+1 \bmod q}$ then $\NH{u} \cap \NH{v}= W_i$.}
	\end{equation*}
\end{lem}
\begin{proof} 
	All indices in this proof are considered to be modulo $q$.
	Let $i\in \{0, \dots, q-1\}$, $u\in W_{i-1}$ and $v\in W_{i+1}$. Our goal is to show that $\NH{u} \cap \NH{v}= W_i$. We split the proof into   
	two cases (Claims~A and~B).
	
	\medskip
	\noindent \textbf{Claim~A:} \textit{
		If there is no biconnected component of $H$ that contains both $w_{i-1}$ and $w_{i+1}$ then $\NH{u} \cap \NH{v}= W_i$.
	}
	\medskip
	
	\begin{claimproof}
		If there is no biconnected component that contains both $w_{i-1}$ and $w_{i+1}$ then, by the definition of $W_i$,   $W_i=\{w_i\}$. Since $w_{i-1}$ and $w_{i+1}$ are not in the same biconnected component every path from $w_{i-1}$ to $w_{i+1}$ goes through $w_i$. There is a path from $w_{i-1}$ to $u$ via $w_{i-2}$ and   there is a path from $v$ to $w_{i+1}$ via $w_{i+2}$. Since $w_{i-2}$ and $w_{i+2}$ are distinct from $w_i$ by Lemma~\ref{lem:distinctwiandwi+2} these paths do not go through $w_i$. Hence every path from $u$ to $v$ also goes through $w_i$. Thus, there is no biconnected component that contains both $u$ and $v$. Hence, $\NH{u} \cap \NH{v}= \{w_i\}= W_i$, as required.
		This concludes the proof of Claim~A.
	\end{claimproof}

	\medskip
	\noindent \textbf{Claim~B:} \textit{
		If there is a biconnected component $B$ such that $w_{i-1}$ and $w_{i+1}$ are in $B$ then $\NH{u} \cap \NH{v}= W_i$.
	}
	\medskip
	
	\begin{claimproof}
		By Lemma~\ref{lem:distinctwiandwi+2}, $w_{i-1}\neq w_{i+1}$. This together with the fact that $w_i$ is adjacent to both $w_{i-1}$ and $w_{i+1}$ implies that $w_i$ is also in $B$. If $u=w_{i-1}$ then it is trivial that $u$ is in $B$. If $u\neq w_{i-1}$ then $\abs{W_{i-1}}>1$. By the fact that $W_{i-1}=\NH{w_{i-2}}\cap \NH{w_i}$ and the fact that both $w_{i-1}$ and $w_i$ are in $B$, it follows that $W_{i-1}\subseteq V(B)$ and that $w_{i-2}$ is in $B$.
		Thus, we have established that $u$ is in $B$. Analogously, $v$ is in $B$. We state this formally so we can refer to it.
		
		\medskip
		\noindent \textbf{Fact~1:} \textit{
			If $\abs{W_{i-1}}>1$ then  
			every vertex in $W_{i-1}\cup \{w_{i-2}\}$ is in $B$. Similarly, if $\abs{W_{i+1}}>1$ then  
			every vertex in  $W_{i+1}\cup \{w_{i+2}\}$  is in $B$. Consequently, both $u$ and $v$ are in $B$. 
		}
		\medskip

		If $u=w_{i-1}$ and $v=w_{i+1}$ then $\NH{u} \cap \NH{v}= \NH{w_{i-1}} \cap \NH{w_{i+1}}= W_i$, as required.
		
		Therefore, we assume for the rest of the proof that $u\neq w_{i-1}$ (the case $v\neq w_{i+1}$ is symmetric).
		By Fact~1, the walk $(w_{i-2},w_{i-1}, w_i, w_{i+1})$ is in $B$.
		We show that $B$  is an obstruction. Suppose, for contradiction, that $B$ is an edge, diamond or impasse. Then by Observation~\ref{obs:alternatingwalk}, 
		there are cut-vertices~$a$ and~$b$ such that		
		the walk $(w_{i-2},w_{i-1}, w_i, w_{i+1})$ is a subpath of $\SP{a,b}$. This  contradicts Lemma~\ref{lem:obstructionfreepath}, which states that no four consecutive vertices of this path are part of the same biconnected component.
		
		Thus, we have established that $(w_{i-2},w_{i-1}, w_i, w_{i+1})$ is a walk in the obstruction~$B$. By Observation~\ref{obs:alternatingwalk} and the definition of $\cseg_C(\cdot)$ (Definition~\ref{def:DofC}), it is a subwalk of some cycle $C\in \cycles(B)$ following the order $D(C)$. It follows that $W_{i-1}=N_{C,H}(w_{i-1})$ and $W_{i}=N_{C,H}(w_{i})$. By Corollary~\ref{cor:diamonds_disjoint*}, from the fact that $H$ has no hardness gadget and $\abs{W_{i-1}}>1$ it follows that $W_i=\{w_i\}$. Let $\ell$ be the length of $C$. Since $C\in \cycles(B)$ we have $\ell= 3$ or $\ell>4$. We make a   case distinction depending on $\ell$.
		
		\begin{myitemize}
			\item Suppose $\ell=3$ so $w_{i-2} = w_{i+1}$. 
			Suppose, for contradiction that $\abs{W_{i+1}}>1$. Then by Fact~1, $w_{i+2}$ is also in $B$ and $(w_{i-2},w_{i-1}, w_i, w_{i+1}, w_{i+2})$ is a subwalk of $\cseg_C(\cdot)$. This gives a contradiction to the fact that all vertices of $\cseg_C(\cdot)$, apart from possibly its endpoints, are distinct (see Definition~\ref{def:DofC}). Therefore, $W_{i+1}=\{w_{i+1}\}$ and consequently $v=w_{i+1}$. 
			Since $u\neq w_{i-1}$,
			$(v,u,w_i,v)$ and $(v,w_{i-1},w_i,v)$ are two distinct triangles that share 
			the edge $\{w_i,v\}$. By  Lemma~\ref{lem:good_triangle_gadget},
			since $H$ has no hardness gadget, $u$ and $v$ have no common neighbour other than~$w_i$.

			\item Suppose $\ell>4$.
			Apply Lemma~\ref{lem:separators} to the 
			cycle~$C$ and the index~$i-1$.
			This shows that there is a separation $(A_1,A_2)$ of~$H$ 
			such that $C\setminus \{w_{i-1}\} \subseteq A_1$,
			$W_{i-1}=N_{C,H}(w_{i-1})\subseteq A_2$, and $A_1\cap A_2 = \{w_{i-2},w_i\}$.
			Since $u\in A_2$,  $u$ is not adjacent to any vertex in~$C\setminus \{w_{i-2},w_{i-1},w_i\}$.
			By the definition of $\cycles(B)$ (Definition~\ref{def:obstruction}) $C$ is an induced cycle of~$B$, so the 
			cycle $C'$ that is obtained from~$C$   by replacing $w_{i-1}$ with $u$ is also an induced cycle of~$B$.
			Also, $C'$ has length $\ell>4$. 
			Since $H$ has no hardness gadget,
			we can apply Corollary~\ref{cor:diamonds_disjoint*} to obtain    (using the fact that $N_{C',H}(u)= N_{C,H}(w_{i-1})= W_{i-1}$ has cardinality greater 
			than $1$) that $|N_{C',H}(w_i)|=1$.
			By definition, $N_{C',H}(w_i) = \NH{u} \cap \NH{w_{i+1}}$,
			so $\NH{u} \cap \NH{w_{i+1}}=\{w_i\}$.

			\begin{myitemize} 
				\item If $\abs{W_{i+1}}=1$ then $v=w_{i+1}$, so we are finished.

				\item Suppose that $\abs{W_{i+1}}>1$. By Fact~1,   $w_{i+1}$ and $w_{i+2}$ are in $B$, and consequently  
				the walk  $(w_{i-2},u, w_i, w_{i+1}, w_{i+2})$ is a subwalk of~$C'$.
				It follows that $W_{i+1} = N_{C',H}(w_{i+1})$.
				Apply Lemma~\ref{lem:separators} to the 
				cycle~$C'$ and the index~$i+1$.
				This shows that there is a separation $(A_3,A_4)$ of~$H$ 
				such that $C'\setminus \{w_{i+1}\} \subseteq A_3$,
				$W_{i+1} \subseteq A_4$, and $A_3\cap A_4 = \{w_{i},w_{i+2}\}$.
				
				Since $v\in A_4$,  $v$ is not adjacent to any vertex in~$C'\setminus \{w_{i},w_{i+1},w_{i+2}\}$.
				So the
				cycle $C''$ that is obtained from~$C'$   by replacing $w_{i+1}$ with $v$ is also an induced cycle of~$B$ with length $\ell>4$.
				Since $H$ has no hardness gadget,
				we can apply Corollary~\ref{cor:diamonds_disjoint*} to obtain    (using the fact that $N_{C'',H}(v)=   W_{i+1}$ has cardinality greater 
				than $1$) that $|N_{C'',H}(w_i)|=1$.
				By definition, $N_{C'',H}(w_i) = \NH{u} \cap \NH{v}$,
				so we are finished. 			\end{myitemize}
			
		\end{myitemize} 
		This concludes the proof of Claim~B.			
	\end{claimproof}	
	The lemma follows immediately from Claim~A and Claim~B.\end{proof}

The following lemma establishes property~\eqref{equ:CycleHardness+4} for the walk returned by Algorithm~\ref{alg:closedwalk}, as required by Lemma~\ref{lem:CycleHardness+}.

\begin{lem}\label{lem:noDwalk}
	Let $H$ be a connected graph. Let $T$ be a closed suitable subtree of $\bc(H)$. Let $B'$ be an obstruction that is a block of $T$. Let $W=(w_0,,\dots, w_{q-1}, w_0)$ be the output of $\walk{T,B'}$ (Algorithm~\ref{alg:closedwalk}). By Lemma~\ref{lem:closedwalk}, $W$ is a closed walk and $q\ge 3$.
	For each $i\in \{0, \dots, q-1\}$, let $W_i=N_{W,H}(w_i)$. Then there exists no closed walk $D=(d_0,\dots,d_{q-1},d_0)$ with $d_i\in \Gamma_H(W_i)\setminus (W_{i-1}\cup W_{i+1})$ for all $i$ (indices taken modulo $q$).
\end{lem}
\begin{proof}
	Assume for the sake of contradiction that such a walk $D$ exists. We distinguish two cases:
	\begin{enumerate}[(I)]
		\item $W$ is not entirely contained in a single biconnected component of $H$. In this case, there is an index $i$ such that no biconnected component contains both $w_{i-1}$ and $w_{i+1}$.
		Now consider $d_{i-1}\in \Gamma_H(w_{i-1})$ and $d_{i+1}\in \Gamma_H(w_{i+1})$. Note that $d_{i-1}\neq w_i$ and $d_{i+1}\neq w_i$ by the specification of $D$. Note further that $d_{i}\neq w_i$ as we do not allow self-loops in $H$. Consequently, there are two internally vertex disjoint 2-paths from $w_{i-1}$ to $w_{i+1}$; one passes through $d_{i-1}$, $d_i$ and $d_{i+1}$; and the other passes through $w_i$. This is a contradiction to the fact that no biconnected component contains both $w_{i-1}$ and $w_{i+1}$.

		\item $W$ is entirely contained in a biconnected component $B$. By Observation~\ref{obs:alternatingwalk}, the only possibility for this to be true is that $B$ is an obstruction and 
		$W$ is a cycle in $\cycles(B)$. By the definition of obstructions (and $\cycles(B)$), $W$ is thus an induced cycle of length $q$
		such that		
		$q\geq 3$ and $q\neq 4$.  The lemma will follow easily from the following claim.
		
		\medskip
		\noindent \textbf{Claim~A:} \textit{  $D\cap \left(\bigcup_i W_i\right)=\emptyset $.
		}\medskip

		\begin{claimproof}
			Assume for contradiction that $d_i\in W_j$ for some indices~$i$ and~$j$. Note that $j\notin\{i-1,i+1\}$ by the specification of $D$. 
			We cannot have $j=i$ since $d_i \in \Gamma_H(W_i)$ so $d_i\notin W_i$ (otherwise $H$ has a self-loop).	
			Since $j\notin\{i-1,i,i+1\}$, we have $q\geq 5$.
			We will show that $H$ has a $K_4$-minor. 
			Since $d_i$ is adjacent to $w_i$, and it is not equal to $w_{i-1}$ or $w_{i+1}$ 
			and since $W$ is an induced cycle in $B$,
			we conclude that $d_i$ is distinct from the vertices of $W$. $H$ therefore contains a $K_4$-minor containing the vertex~$d_i$ and its three neighbours~$w_i$,
			$w_{j-1}$ and $w_{j+1}$. There are disjoint paths between these three vertices along the cycle~$W$ and $d_i$ is not on these paths. See the following illustration.
			
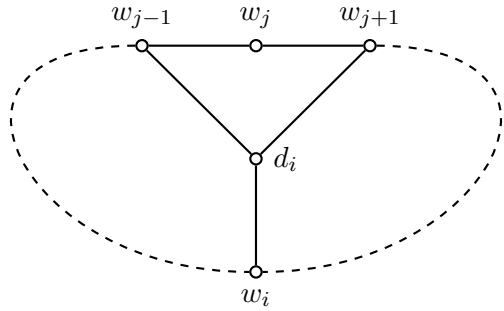
\begin{center}
				\begin{tikzpicture}[scale=1.5, node distance = 1.4cm,thick]
				\tikzstyle{dot}   =[fill=black, draw=black, circle, inner sep=0.15mm]
				\tikzstyle{vertex}=[  draw=black, circle, inner sep=1.5pt]
				\tikzstyle{dist}  =[fill=white, draw=black, circle, inner sep=2pt]
				\tikzstyle{pinned}=[draw=black, minimum size=10mm, circle, inner sep=0pt]	
				\node[vertex] (wi-1) at (-1  ,1) [label=90: $w_{j-1}$] {};
				\node[vertex] (wi+1) at (1  ,1) [label=90: $w_{j+1}$] {};
				\node[vertex] (di) at (0  ,1) [label=90: $w_j$] {};
				\node[vertex] (wip) at (0  ,0) [label=0: $d_i$] {}; 
				\node[vertex] (wjp) at (0  ,-1) [label=270: {$w_i$}] {};
				
				\draw[dashed] (wi-1) to[out=180,in=180,distance=2cm] (wjp); \draw (wi-1) -- (di); \draw (wi-1) -- (wip); 
				\draw (di) -- (wi+1); \draw (wip) -- (wi+1);\draw (wjp) -- (wip);
				\draw[dashed] (wi+1) to[out=0,in=0,distance=2cm] (wjp);
				\end{tikzpicture}
			\end{center}
			This concludes the proof of Claim~A.
		\end{claimproof}
		We have assumed for contradiction that $D$ exists, and proved Claim~A.	
		We obtain the  contradiction by using  Claim~A to construct a $K_4$-minor in~$H$.  Claim~A demonstrates that $W\cap D=\emptyset$.
		Now contract the walk $D$ to a single vertex. This yields a vertex $d\notin W$ which is adjacent to all vertices of $W$. As $W$ has length at least $3$, we have found a $K_4$-minor as promised. 
	\end{enumerate}
\end{proof}

\begin{lem}\label{lem:suitableTreeHardnessWithObstruction}
	Let $H$ be a connected $K_4$-minor-free graph. Let $T$ be a closed suitable subtree of $\bc(H)$. Let $B'$ be an obstruction that is a block of $T$. Then $H$ has a hardness gadget.
\end{lem}
\begin{proof} 
	Let $W=(w_0,\dots,w_{q-1}, w_0)$ be the output of $\walk{T,B'}$. 
	By Lemma~\ref{lem:closedwalk}, $W$ is a closed walk with $q\ge 3$ and $q\neq 4$.	
	Our goal is to use  Lemma~\ref{lem:CycleHardness+} to show that $H$ has a hardness gadget. To this end, we identify the sets $\mathcal{C}_i$ of Lemma~\ref{lem:CycleHardness+} with the sets $W_i=N_{W,H}(w_i)$. Let $S$ be the set of all $i$ such that $W_i$ has even cardinality. 
	
	\medskip
	\noindent \textbf{Claim A:}  \textit{
		For every $i\in S$, there is an obstruction $O_i$ such that, for $C_i=C(O_i, \Nof{T}(O_i))$, the following hold.
		\begin{myitemize}
			\item $w_{i-1}, w_i, w_{i+1}\in C_i$.
			\item $W_i=N_{C_i,H}(w_i)$.
			\item Every vertex in $W_i$ has degree $2$ in $O_i$.
			\item $w_i$ is an attachment point of $(T,O_i)$.
		\end{myitemize}
	}
	\medskip
	
	\begin{claimproof} Fix $i\in S$.
		By the definition of $W_i$ and the fact that $\abs{W_i}>1$, there is a biconnected component $O_i$ of~$H$ that contains $w_{i-1}$, $W_i$, and $w_{i+1}$. Suppose, for contradiction, that $O_i$ is an edge, diamond or impasse, then, by Observation~\ref{obs:alternatingwalk}, the walk $(w_{i-1}, w_i, w_{i+1})$ is a subpath of a path of the form $\SP{\cdot}$. However, since $\abs{W_i}\ge 2$, $w_{i-1}$ and $w_{i+1}$ have at least $2$ common neighbours in $H$, this is a contradiction to Lemma~\ref{lem:obstructionfreepath}.
		
		We have established that $O_i$ is an obstruction. Since $T$ is a suitable subtree, $(O_i, \Nof{T}(O_i))$ is a suitable connector and, by Definition~\ref{def:suitableconnector}, $C_i=C(O_i, \Nof{T}(O_i))$ is a cycle with $\Nof{T}(O_i)=\{c \in C_i \mid \text{the cardinality of } N_{C_i,H}(c)  \text{ is even}\}$. By Observation~\ref{obs:alternatingwalk}, $(w_{i-1}, w_i, w_{i+1})$ is a subwalk of a walk of the form $\cseg_{C_i}(\cdot)$. It follows that $w_{i-1}, w_i, w_{i+1}\in C_i$ and $W_i=N_{C_i,H}(w_i)$, as required. 
		
		As the cardinality of $W_i$ is even and $C_i\in \cycles(O_i)$, by Definition~\ref{def:obstruction},  
		every vertex in~$W_i$ has degree~$2$ in~$O_i$, 	 as required.
		
		The fact that the cardinality of $W_i$ is even also implies that $w_i\in \Nof{T}(O_i)$ (since $\Nof{T}(O_i)=\{c \in C_i \mid \text{the cardinality of } N_{C_i,H}(c)  \text{ is even}\}$). Thus, by Definition~\ref{def:exit}, $w_i$ is either an exit or an attachment point of $(T,O_i)$. However, by Observation~\ref{obs:alternatingwalk}, only the endpoints of $\cseg_{C_i}(\cdot)$ are exits, which means that $w_i$ is an attachment point, as required.
		This finishes the proof of Claim~A.
	\end{claimproof}

	In the remainder of this proof, for each $i\in S$, let $O_i$ and $C_i$ be as stated in Claim~A. 
	Next we use the fact that, for each $i\in S$, $w_i$ is an attachment point of $(T,O_i)$ to define a gadget $(\hat{J}_i,\hat{z}_i)$. Those gadgets will be used in the construction of the gadgets $(J_0,z_0),\dots,(J_{q-1},z_{q-1})$ required by Lemma~\ref{lem:CycleHardness+}. Recall that, by definition of attachment points (Definition~\ref{def:exit}), for each $i\in S$, there is a (unique) maximal-length proper obstruction-free path $P_i =(b^i_0,B^i_1,b^i_1,B^i_2,\dots,B^i_{q_i},b^i_{q_i} )$ in $T$ 
	such that $w_i = b^i_{q_i}$ and $b^i_0$ is a leaf in $T$. As $T$ is closed, we obtain that, for some block $R$ of $T$, the vertex $b^i_0$ is $R$-closed, i.e., $b^i_0$ has precisely one descendant in $\bc(H)$ with respect to $<_R$. Moreover, this descendant must be an edge. 
	We distinguish whether $q_i=0$ or $q_i\ge 1$:
	\begin{itemize}
		\item[$q_i=0$:] We have $w_i=b^i_{q_i} = b^i_0$. 
		Since $b^i_0=w_i$ is $R$-closed,  Lemma~\ref{lem:choiceofRdoesnotmatter} ensures that it is also  $O_i$-closed. Consequently, $w_i$ has precisely three neighbours in $H$: The two neighbours in $O_i$ (which are $w_{i-1}$ and $w_{i+1}$ --- these are distinct by Lemma~\ref{lem:distinctwiandwi+2}), as well as the other endpoint~$\ell_i$ of the edge  that is the unique descendant   of $w_i$ in $\bc(H)$.
		
		We define $\hat{J}_i$ to be a single edge, one endpoint of which is $z_i$, and the other endpoint of which is pinned to $w_i$. Observe that
		\[\{v\in V(H)\mid \enspace \abs{\hom{(\hat{J}_i,\hat{z}_i)}{(H,v)}} \text{ is odd.} \} = \{w_{i-1},w_{i+1},\ell_i\}\,.\]
		This concludes the definition of $(\hat{J}_i, \hat{z}_i)$ in the case that $q_i=0$.
		
		\item[$q_i\geq 1$:] By Lemma~\ref{lem:choiceofRdoesnotmatter}, $b^i_0$ is   $B^i_0$-closed. It follows that 
		$|\Gamma_H(b^i_0)\setminus B^i_0|=1$.
		Since $T$ is a suitable subtree and $P_i$ is obstruction-free,
		for each $j\in [q_i]$, $(B^i_j, \{b^i_{j-1}, b^i_j\})$ is a suitable connector in~$H$ and $B^i_j$ is an edge, diamond or impasse. Thus, we can invoke Lemma~\ref{lem:CaterpillarForBipChordalComponents}. We obtain that at least one of the following is true:
		\begin{myitemize}
			\item $H$ has a hardness gadget.
			\item $B^i_{q_i}$ is an edge or a diamond and $( L_i,b^i_{q_i})$ is a good start in~$B^i_{q_i}$ but not a good stop in~$B^i_{q_i}$, 
			where $L_i = 
			\{b^i_{q_i-1}\}$.		
			\item $B^i_{q_i}$ is an impasse, 		
			and $(L_i, b^i_{q_i})$ is a good start in $B^i_{q_i}$ but not a good stop in~$B^i_{q_i}$, where
			$L_i = \{d_i\}$ and 
			$d_{i}$ is  
			the unique common neighbour of $b^i_{q_i-1}$ and $b^i_{q_i}$	in~$H$. 			
		\end{myitemize}

		We are done in the first case, so suppose that one of other cases applies.
		By definition of good starts, we thus obtain a gadget $(\hat{J}_i,\hat{z}_i)$ such that, for $R_{i} = \Gamma_H(b^i_{q_i})\setminus V(B^i_{q_i})$,
		\[\{v\in V(H) \mid \enspace \abs{\hom{(\hat{J}_i,\hat{z}_i)}{(H,v)}} \text{ is odd.} \} = L_{i} \cup R_{i}.\]
		Note that   $L_{i}$ and $R_{i}$ are disjoint. Further, recall that $w_i = b^i_{q_i}$ and therefore $R_{i}\cap V(O_i)=\{w_{i-1},w_{i+1}\}$ (by Claim~A). As $(L_{i},b^i_{q_i})$ is not a good stop in $B^i_{q_i}$, we have that $R_{i}$ is of even cardinality, and thus $L_{i} \cup R_{i}$ is of odd cardinality.
		This concludes the definition of $(\hat{J}_i, \hat{z}_i)$ in the case that $q_i\geq 1$.
	\end{itemize}
	We now state the previously-established crucial property of the gadgets $(\hat{J}_i,\hat{z}_i)$ (which applies for all $q_i$, unless $H$ has a hardness gadget).
	
	\medskip
	\noindent \textbf{Fact 1:} \textit{For every $i\in S$, there is a gadget $(\hat{J}_i,\hat{z}_i)$ such that the set \[\hat{\Omega}_i=\{v\in V(H)\mid \enspace \abs{\hom{(\hat{J}_i,\hat{z}_i)}{(H,v)}} \text{ is odd.} \}\] 
		is a subset of $\NH{w_i}$, has odd cardinality, and contains precisely two vertices of $O_i$ --- the vertices $w_{i-1}$ and $w_{i+1}$.}
	\medskip
	
	We proceed by defining for each $i\in \{0, \dots, q-1\}$ the gadgets $(J_i,z_i)$ needed for Lemma~\ref{lem:CycleHardness+}. The definition of $(J_i,z_i)$ depends on whether or not $i-1$ and $i+1$ are in $S$. $(J_i,z_i)$ always contains the vertex $z_i$. Additionally, for $j\in \{i-1, i+1\}$, if $j \in S$ then $(J_i,z_i)$ also contains a copy of the gadget $(\hat{J}_j,\hat{z}_j)$ and $z_i$ is adjacent to $\hat{z}_j$. We provide an illustration of $(J_i,z_i)$ for the case $i-1\in S$ and $i+1\in S$ in Figure~\ref{fig:attachmentgadget}.

	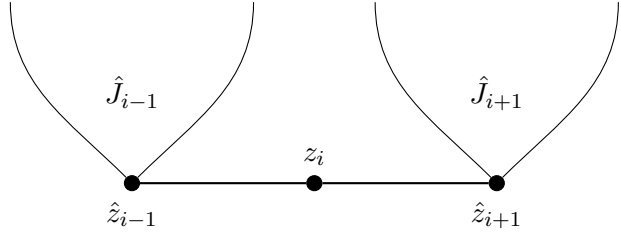
\begin{figure}
		\centering
		\begin{tikzpicture}
		[scale=0.8, node distance = 1.4cm]
		\tikzstyle{dot}   =[fill=black, draw=black, circle, inner sep=0.15mm]
		\tikzstyle{vertex}=[fill=black, draw=black, circle, inner sep=2pt]
		\tikzstyle{dist}  =[fill=white, draw=black, circle, inner sep=2pt]
		
		\node[vertex] (l) at (-3  ,0) [label=270:{$\hat{z}_{i-1}$}] {};
		\node[vertex] (r) at (3  ,0) [label=270:{$\hat{z}_{i+1}$}] {};
		\node[vertex] (m) at (0,0) [label=90:{$z_{i}$}] {};
		\node (jim1) at (-3,1.5) {$\hat{J}_{i-1}$};
		\node (jip1) at (3,1.5) {$\hat{J}_{i+1}$};
		
		\draw[thick] (m) -- (l); \draw[thick] (m) -- (r);
		\draw (l) to[out=135, in=270] (-5,3); \draw (l) to[out=45, in=270] (-1,3);
		\draw (r) to[out=45, in=270] (5,3); \draw (r) to[out=135, in=270] (1,3);
		\end{tikzpicture}
		\caption{The gadget $(J_i,z_i)$ for the case $i-1\in S$ and $i+1\in S$}
		\label{fig:attachmentgadget}
	\end{figure}

	Following the notation of Lemma~\ref{lem:CycleHardness+}, we set (for every $i\in\{0,\dots,q-1\}$):
	\[\Omega_i = \{v\in V(H)\mid \enspace \abs{\hom{(J_i,z_i)}{(H,v)}} \text{ is odd.} \} \]
	
	\medskip
	\noindent \textbf{Claim B:} \textit{For all $i\in \{0,\dots,q-1\}$ the following holds true, unless $H$ has a hardness gadget; indices are taken modulo $q$:
		\begin{itemize}
			\item If $i-1\in S$ then $W_{i-1}\cap \Omega_i =\{w_{i-1}\}$, and if $i-1\notin S$ then $W_{i-1}\cap \Omega_i = W_{i-1}$.
			\item If $i+1\in S$ then $W_{i+1}\cap \Omega_i =\{w_{i+1}\}$, and if $i+1\notin S$ then $W_{i+1}\cap \Omega_i = W_{i+1}$.
		\end{itemize}
		So, $W_{i-1}\cap \Omega_i$ and $W_{i+1}\cap \Omega_i$ have odd cardinality.
	}
	\medskip
	
	\begin{claimproof}
		We only show the first item; the second one is symmetric. We distinguish whether or not $i-1\in S$:
		\begin{enumerate}[(I)]
			\item If $i-1\in S$ then, by Claim~A, $w_{i-1}$ and $w_i$ are contained in the cycle $C_{i-1}$, and $W_{i-1}=N_{C_{i-1},H}(w_{i-1})$.
			Since $C_{i-1}\in\cycles(O_{i-1})$, $C_{i-1}$ is
			an  induced 
			cycle in~$O_{i-1}$ 	(hence in~$H$)		
			and is not a square. Therefore we can apply Corollary~\ref{cor:diamonds_disjoint*}. It follows that $H$ either has a hardness gadget, in which case we are done, or $\abs{N_{C_{i-1},H}(w_i)}=1$, i.e., $N_{C_{i-1},H}(w_i)=\{w_i\}$. This implies that $w_i$ is not an exit of $(T,O_{i-1})$, and thus $w_{i+1}\in C_{i-1}$ and consequently $W_i=N_{C_{i-1},H}(w_i)=\{w_i\}$. 
			
			We are now able to prove that $W_{i-1}\cap \Omega_i =\{w_{i-1}\}$. First, for $v\in W_{i-1}\setminus\{w_{i-1}\}$, we show that $v\notin \Omega_i$ and hence $v\notin W_{i-1}\cap \Omega_i$. By the construction of $J_i$, it suffices to show that there is an even number of vertices in $\hat{\Omega}_{i-1}$ that are adjacent to $v$. Recall from Fact~1 that $\hat{\Omega}_{i-1}\subseteq \NH{w_{i-1}}$. The vertex $v$ has precisely two common neighbours with $w_{i-1}$, namely $w_{i-2}$ and $w_i$ (any others would lead to a $K_4$-minor in $H$ induced by the vertices $\{v,w_{i-2}, w_{i-1}, w_i\}$). By Fact~1, we know that both of these are in $\hat{\Omega}_{i-1}$ and hence that there are two vertices in $\hat{\Omega}_{i-1}$ that are adjacent to $v$, as required.

			It remains to show that $w_{i-1}\in \Omega_i$ and hence $w_{i-1}\in W_{i-1}\cap \Omega_i$.
			By the construction of $J_i$, it suffices to show that there is an odd number of vertices in $\hat{\Omega}_{i-1}$ that are adjacent to $w_{i-1}$, and, in case $i+1\in S$, that there is an odd number of vertices in $\hat{\Omega}_{i+1}$ that are adjacent to $w_{i-1}$.
			
			\begin{itemize}
				\item By Fact~1, $\hat{\Omega}_{i-1}\subseteq \NH{w_{i-1}}$. Hence every element of $\hat{\Omega}_{i-1}$ is adjacent to $w_{i-1}$. By Fact~1, $\hat{\Omega}_{i-1}$ has odd cardinality, as required.
				
				\item Suppose that $i+1\in S$. By Fact~1, we have $\hat{\Omega}_{i+1}\subseteq\NH{w_{i+1}}$ and $w_i\in \hat{\Omega}_{i+1}$. Furthermore, $w_i$ is the only common neighbour of $w_{i-1}$ and $w_{i+1}$ in $H$ by the fact that $W_i=\{w_i\}$. Hence $w_i$ is the only vertex in $\hat{\Omega}_{i+1}$ that is adjacent to $w_{i-1}$, as required.
			\end{itemize}
			
			\item Consider $i-1\notin S$. 
			Our goal is to show that	$W_{i-1}\cap \Omega_i = W_{i-1}$.		
			If $i+1\notin S$, then $J_i$ contains only the single vertex $z_i$ and $\Omega_i=V(H)$ and we are finished.
			
			Hence we can assume $i+1\in S$. We first proceed as in Case~(I) to obtain either a hardness gadget or $W_i=\{w_i\}$. By Claim~A, $w_{i}$, $w_{i+1}$ and $w_{i+2}$ are contained in the cycle $C_{i+1}$, and $W_{i+1}=N_{C_{i+1},H}(w_{i+1})$.
			Since $C_{i+1}\in\cycles(O_{i+1})$, $C_{i+1}$ is induced and not a square and therefore we can apply Corollary~\ref{cor:diamonds_disjoint*}. It follows that $H$ either has a hardness gadget, in which case we are done, or $\abs{N_{C_{i+1},H}(w_i)}=1$, i.e., $N_{C_{i+1},H}(w_i)=\{w_i\}$. This implies that $w_i$ is not an exit of $(T,O_{i+1})$, and thus $w_{i-1}\in C_{i+1}$ and consequently $W_i=N_{C_{i+1},H}(w_i)=\{w_i\}$.  
			
			In order to show that $W_{i-1}\cap\Omega_{i} = W_{i-1}$ we show, for each $v\in W_{i-1}$, that $v \in \Omega_i$.
			By the construction of $J_i$ ($i-1\notin S$, $i+1\in S$), it suffices to show that there is an odd number of vertices in $\hat{\Omega}_{i+1}$ that are adjacent to $v$. Recall from Fact~1 that $\hat{\Omega}_{i+1}\subseteq \NH{w_{i+1}}$. By the fact that $v \in W_{i-1}$ and $w_{i+1}\in W_{i+1}$, from Lemma~\ref{lem:properWalkIntersection} we obtain that $\NH{v}\cap \NH{w_{i+1}} = W_i$. We have already established that $W_i=\{w_i\}$ and hence $w_i$ is the only vertex in $\hat{\Omega}_{i+1}$ that is adjacent to $w_{i-1}$, as required.

		\end{enumerate} 
		This concludes the proof of Case~(II) and of Claim~B.	\end{claimproof}
	We prove one final claim before we can apply Lemma~\ref{lem:CycleHardness+}:
	
	\medskip
	\noindent \textbf{Claim C:} \textit{Unless $H$ has a hardness gadget, there exists $k\in \{0,\dots,q-1\}$ such that both of the following are true; indices are taken modulo $q$:
		\begin{itemize}
			\item There are no edges between $W_k$ and $W_{k+3}$.
			\item $(W_k \cup W_{k+2})\cap \Omega_{k+1}$ and $(W_{k+1} \cup W_{k+3})\cap \Omega_{k+2}$ are of even cardinality.
		\end{itemize}
	}
	\begin{claimproof}
		We distinguish two cases.
		\begin{enumerate}[(I)]
			\item There is a biconnected component $B$ that contains $W$. Consequently, by Observation~\ref{obs:alternatingwalk}, there is a cycle $C\in \cycles(B)$ such that $W=C$. Since $C\in \cycles(B)$ it has length $q=3$ or $q\ge 5$. 
			In this case, we choose $k=0$.
			We first show that there is no edge between $W_k$ and $W_{k+3}$:
			\begin{myitemize}
				\item If $q=3$, we show that for $u,v\in W_0$ there cannot be an edge between $u$ and $v$. If $u=v$ there cannot be an edge since we do not allow self-loops in $H$. If $u\neq v$ there cannot be an edge, as otherwise $u,v,w_1,w_2$ induce a $K_4$-minor
				in~$H$, contradicting the fact that~$H$ has none.			 
				Thus, there are no edges between $W_0$ and $W_{3 \bmod q}=W_0$.
				
				\item If $q\ge 5$, consider $W_0$ and $W_3=W_{3 \bmod q}$. If $\abs{W_0}=\abs{W_3}=1$ then there are no edges between $W_0$ and $W_3$ since $C$ is induced by the definition of obstruction (Definition~\ref{def:obstruction}). 
				So suppose $|W_0|>1$ (the case $|W_3|>1$ is symmetric).				
				Since $C$ is an induced cycle of length $q>4$ in a biconnected graph $B$, we can apply Lemma~\ref{lem:separators} to 
				find a separation $(A,A')$ of~$H$ such that $C\setminus \{w_0\} \subseteq A$, $W_0 \subseteq A'$and $A\cap A' = \{w_q,w_1\}$.
				Since all of the vertices in $\bigcup_{i=1}^{q-1} W_i$
				have neighbours in  $C\setminus \{w_0\}$, 
				this implies
				that $w_{q-1}$ and $w_1$ are the only vertices in $\bigcup_{i=1}^{q-1} W_i$ that are adjacent to vertices in $W_0$. However, by Lemma~\ref{lem:disjoint}, $W_0, \dots W_{q-1}$ are pairwise disjoint and hence $w_{q-1}, w_1\notin W_3$. So, there are no edges between $W_0$ and $W_3$, as required.
			\end{myitemize}
			To establish the second bullet point, again use $k=0$ and the fact that
			$W_0, \dots W_{q-1}$ are pairwise disjoint. We have
			$\abs{(W_0 \cup W_{2})\cap \Omega_{1}} = \abs{W_0 \cap \Omega_{1}}+\abs{W_{2}\cap \Omega_{1}}$.
			By Claim~B, each of these terms is odd, so their sum  is even. The same argument applies to $(W_{1} \cup W_{3})\cap \Omega_{2}$.
			
			\item $W$ is not entirely contained in one biconnected component. If this is true, then by Observation~\ref{obs:alternatingwalk}, there exists an obstruction $B$ with cycle $C\in \cycles(B)$ such that, for some $k\in \{0, \dots, q-1\}$, $w_k$ and $w_{k+1}$ are contained in $C$, $w_{k+1}$ is an exit of $B$ (in particular, an articulation point), and $w_{k+2}$ and $w_{k+3}$ are not contained in $B$.
			
			Since $w_{k+2}\neq w_{k+4}$ by Lemma~\ref{lem:distinctwiandwi+2}, it follows that no $v\in W_{k+3}$ is in $B$, which implies that there is no edge between $W_k$ and $W_{k+3}$, as required.
			
			For the second item, observe that $W_k$ and $W_{k+2}$ must be disjoint, as $w_k$ and $w_{k+1}$ are in the biconnected component $B$, but $w_{k+2}$ is not. We further claim that $W_{k+1}$ and $W_{k+3}$ are disjoint. To see this, observe first that $W_{k+1}=\{w_{k+1}\}$ since $w_{k+1}$ is the only common neighbour of $w_k$ and $w_{k+2}$ as otherwise $w_{k+2}$ would be contained in $B$. Then we have already established that no $v\in W_{k+3}$ is in $B$, which implies $w_{k+1}\notin W_{k+3}$.
			
			Using the fact that $W_k$ and $W_{k+2}$ are disjoint, we conclude that 
			$\abs{(W_k \cup W_{k+2})\cap \Omega_{k+1}} = (\abs{W_k \cap \Omega_{k+1}}+\abs{W_{k+2}\cap \Omega_{k+1}})$.
			By Claim~B, each of these terms is odd, so their sum is even.
			Using the fact that $W_{k+1}$ and $W_{k+3}$ are disjoint, the same is true for $W_{k+1}$ and $W_{k+3}$.
		\end{enumerate}
	\end{claimproof}
	
	We are finally able to invoke Lemma~\ref{lem:CycleHardness+}: Recall first, that $q\geq 3$ and $q\neq 4$ from the beginning of the proof. 
	Recall that we identify the sets $\mathcal{C}_i$ of Lemma~\ref{lem:CycleHardness+} with the sets $W_i$.	
	Unless $H$ has a hardness gadget (in which case we are finished) the following hold.
	\begin{itemize}
		\item[] \eqref{equ:CycleHardness+1} holds by Claim B.
		\item[] \eqref{equ:CycleHardness+2} and~\eqref{equ:CycleHardness+3} hold by Lemma~\ref{lem:properWalkIntersection}.
		\item[] \eqref{equ:CycleHardness+4} is established by Lemma~\ref{lem:noDwalk}. 
		\item[] There is a $k$ such that  \eqref{equ:CycleHardness+5} and~\eqref{equ:CycleHardness+6} hold by Claim C.
	\end{itemize}
	Consequently, all conditions are satisfied and we obtain a hardness gadget by Lemma~\ref{lem:CycleHardness+}.
\end{proof}

\subsection{Proof of the Main Theorem}

We can now prove Theorem~\ref{thm:main}, which we restate for convenience.

\begin{thmmain}
\statethmmain
\end{thmmain}

\begin{proof}
By Theorem~\ref{thm:ReductionSuffices}, for every graph~$G$,
$|\hom{G}{H}| =|\hom{G}{H^\ast}|\mod 2$.
It is trivial to count homomorphisms to a graph with at most one vertex.
Suppose that $H^*$ has at least two vertices.
Then it suffices to show that  $\parHom{H^\ast}$
is $\parP$-complete and that 
 $\parHom{H^*}$ cannot be solved in time $\exp(o(\abs{G}))$, unless the rETH fails. 
	
	Since $H^\ast$ is involution-free and contains at least $2$ vertices, there is an involution-free connected component $H'$ of $H^\ast$ with at least $2$ vertices as well: If $H$ is disconnected, it has at least $2$ connected components, and at least one of those two components cannot be a single vertex, as otherwise, we obtain a non-trivial involution by switching those vertices. Furthermore, a connected component of an involution-free graph cannot have a non-trivial involution, as otherwise, the entire graph would  have a non-trivial involution.
	
	Next we claim that $H'$ has a hardness gadget:	Assume first that 
	$H'$ has a biconnected component that is not an edge, a diamond, an obstruction, or an impasse. By Lemma~\ref{lem:main_biconnected}, $H'$ has a hardness gadget. In the remaining case,  every biconnected component of $H'$ is an   edge, a diamond, an obstruction, or an impasse. By Lemma~\ref{lem:findsuitablesubtree}, 
	there is a closed suitable subtree~$T$ of the block-cut tree $\bc(H')$. If no block of $T$ is an obstruction, then 
	$H'$ has a  hardness gadget   by Lemma~\ref{lem:suitableTreeHardnessNoObstruction}.
	Otherwise,  $H'$ has a hardness gadget by Lemma~\ref{lem:suitableTreeHardnessWithObstruction}. 
	
	This allows us to invoke Theorem~\ref{thm:hardness-gadget} and we obtain that $\parRet{H'}$ is $\parP$-hard and cannot be solved in time $\exp(o\abs{J})$, unless the rETH fails.
	
	Since $H'$ is involution-free, we can reduce $\parRet{H'}$ to $\parHom{H'}$ by Theorem~\ref{thm:ETH_RetToHoms}, and we can reduce $\parHom{H'}$ to $\parHom{H^\ast}$ by Lemma~\ref{lem:ETH_connectivity}.  These reductions are tight in the sense that any subexponential-time algorithm for $\parHom{H^\ast}$ would yield a subexponential-time algorithm for $\parRet{H'}$; this is due to the fact that the size of the oracle queries in each reduction is bounded linearly in the input size (see, e.g., the proof of Theorem~\ref{thm:hardness-gadget} in which we made explicit that a linear bound on the size of the oracle queries is sufficient for the lower bound under rETH to transfer).
	We thus obtain $\parP$-hardness of $\parHom{H^\ast}$, and that $\parHom{H^\ast}$ cannot be solved in time $\exp(o(|G|))$, unless the rETH fails.
\end{proof}

	\section{Counting Homomorphisms mod $2$ to Graphs of Degree at most $3$}\label{sec:degreeAtmost3}

We explored the possibilities for constructing hardness gadgets in graphs containing two squares that share one edge  when we analysed $K_4$-minor-free and chordal bipartite graphs. It turns out that a similar strategy suffices to completely solve the case  where $H$ has degree at most $3$. We start with the following lemma.

\begin{lem}\label{lem:mainDeg3}
	Let $H$ be an involution-free graph of degree at most $3$  that contains a square. Then $H$ has a hardness gadget.
\end{lem}
\begin{proof}
	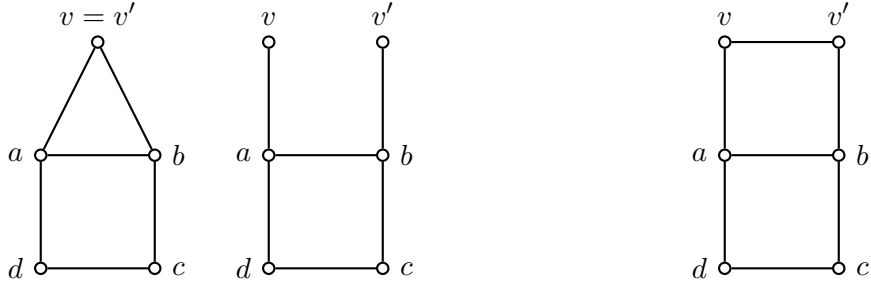
\begin{figure}
		\centering
		\begin{tikzpicture}[scale=1.5, node distance = 1.4cm,thick]
		\tikzstyle{dot}   =[fill=black, draw=black, circle, inner sep=0.15mm]
		\tikzstyle{vertex}=[  draw=black, circle, inner sep=1.5pt]
		\tikzstyle{dist}  =[fill=white, draw=black, circle, inner sep=2pt]
		\tikzstyle{pinned}=[draw=black, minimum size=10mm, circle, inner sep=0pt]	
		\node[vertex] (a) at (0  ,1) [label=180: $a$] {};
		\node[vertex] (b) at (1  ,1) [label=0: $b$] {};
		\node[vertex] (c) at (1  ,0) [label=0: $c$] {};
		\node[vertex] (d) at (0  ,0) [label=180: $d$] {}; 
		\node[vertex] (xy) at (0.5  ,2) [label=90: ${v=v'}$] {};
		
		\draw (a) -- (b); \draw (b) -- (c); \draw (c) -- (d); \draw (d) -- (a);
		\draw (a) -- (xy); \draw (b) -- (xy);
		
		\begin{scope}[xshift=2cm]
		\node[vertex] (a) at (0  ,1) [label=180: $a$] {};
		\node[vertex] (b) at (1  ,1) [label=0: $b$] {};
		\node[vertex] (c) at (1  ,0) [label=0: $c$] {};
		\node[vertex] (d) at (0  ,0) [label=180: $d$] {}; 
		\node[vertex] (x) at (0  ,2) [label=90: ${v}$] {};
		\node[vertex] (y) at (1  ,2) [label=90: ${v'}$] {};
		
		\draw (a) -- (b); \draw (b) -- (c); \draw (c) -- (d); \draw (d) -- (a);
		\draw (a) -- (x); \draw (b) -- (y);
		\end{scope}
		
		\begin{scope}[xshift=6cm]
		\node[vertex] (a) at (0  ,1) [label=180: $a$] {};
		\node[vertex] (b) at (1  ,1) [label=0: $b$] {};
		\node[vertex] (c) at (1  ,0) [label=0: $c$] {};
		\node[vertex] (d) at (0  ,0) [label=180: $d$] {}; 
		\node[vertex] (x) at (0  ,2) [label=90: ${v}$] {};
		\node[vertex] (y) at (1  ,2) [label=90: ${v'}$] {};
		
		\draw (a) -- (b); \draw (b) -- (c); \draw (c) -- (d); \draw (d) -- (a);
		\draw (a) -- (x); \draw (b) -- (y); \draw (x) -- (y); 
		\end{scope}

		\end{tikzpicture}
		\caption{\label{fig:CasesDeg3}Illustration of the two cases in the proof of Lemma~\ref{lem:mainDeg3}.  }
	\end{figure}
	\noindent Let $C=(a,b,c,d,a)$ be a square in $H$. Assume first that at least one of the edges $\{a,c\}$ or $\{b,d\}$ are present. W.l.o.g.\ let $\{a,c\}$ be present. Then $a$ and $c$ have degree $3$ and thus, by assumption, no further neighbours. Thus $(ac)$ is a non-trivial involution of~$H$.
	
	Now assume that none of $\{a,c\}$ or $\{b,d\}$ are edges of $H$. If both, $a$ and $c$ have degree $2$ then $(ac)$ is an involution. Similarly, if $b$ and $d$ have both degree $2$, we obtain the involution $(bd)$. W.l.o.g.\ we can thus assume that $a$ and $b$ have degree $3$. Let $v$ and $v'$ be the neighbours of $a$ and $b$, respectively, that are not contained in $C$. In what follows, we 
	consider cases based on whether the edge $\{v,v'\}$ is present, and, if not, we differentiate between $v=v'$ and $v\neq v'$; an illustration  is provided in Figure~\ref{fig:CasesDeg3}.
	
	\begin{enumerate}[(I)]
		\item $\{v,v'\}\notin E(H)$:   
		This case corresponds to the two illustrations to the left of Figure~\ref{fig:CasesDeg3}.
		Note first that $\{v',d\}$ cannot be an edge of $H$, as otherwise, $b$ and $d$ both have neighbours $\{a,v',c\}$
		(and no other neighbours, since they have degree~$3$), which means 
		that   
		$(bd)$ is a non-trivial involution of $H$. 
		Similarly, $\{v,c\}$ cannot be an edge of $H$, as otherwise $(ac)$ is a non-trivial involution of $H$. 
		Also,  at least one of the edges $\{v,d\}$ and $\{v',c\}$ must not be present in $H$, as otherwise $(ad)(bc)$ is a non-trivial involution of $H$. W.l.o.g.,  assume that $\{v,d\}$ is not present.
		We construct a hardness gadget of $H$ as follows:
		\begin{itemize}
			\item $I=\{a\}$.
			\item $S=\{b\}$.
			\item $J_1$ is a path of $4$ vertices: The first vertex is a $b$-pin, the second vertex is $y$, and the fourth vertex is 
			an $a$-pin.
			\item $J_2$ is a path of $3$ vertices: The first vertex is an $a$-pin, the second vertex is $z$, and the third vertex is a $c$-pin.
			\item $J_3$ is just the edge $\{y,z\}$.
		\end{itemize}
		We first claim that $\Omega_y=\{v',a\}$.
		A vertex of $H$ is in $\Omega_y$ if and only if it is adjacent to $b$ and  has an odd number of 2-paths to $a$. As $H$ has degree at most three, the neighbours of $b$ are precisely $v'$, $a$ and $c$. Note that $a$ has degree precisely $3$ and thus has an odd number of 2-paths to itself. Furthermore, there is only one 2-path from $v'$ to $a$: This path contains $b$ as internal vertex. There cannot be an additional 2-path from $v'$ to $a$, since, in this case, the internal vertex must either be $v$, which is not possible as $\{v,v'\}\notin E(H)$, or $d$, which is not possible as $\{v,d\}\notin E(H)$. Finally, there are precisely two 2-paths from $c$ to $a$: One has $b$ as internal vertex, and the other has $d$ as internal vertex. There cannot be a third one, as this 2-path would have $v$ as internal vertex, but we ruled out the existence of the edge $\{v,c\}$. This shows that $\Omega_y=\{v',a\}$.
		
		Our next claim is that $\Omega_z=\{b,d\}$. Observe that $\Omega_z$ contains precisely the common neighbours of $a$ and $c$. Thus $b$ and $d$ are included in $\Omega_z$. The only candidate for a third common neighbour would be $v$, but we ruled out the existence of the edge $\{v,c\}$.
		
		Finally, we observe that $|\Sigma_{v',d}|=0$ as $\{v',d\}$ is not an edge of $H$, and that $|\Sigma_{v',b}|=|\Sigma_{b,a}|=|\Sigma_{a,d}|=1$.
		
		\item $\{v,v'\}\in E(H)$: 
		This case corresponds to the illustration to the right of Figure~\ref{fig:CasesDeg3}.
		As in case (I), the edge $\{v,c\}$ is not present, as otherwise $(ac)$ is a non-trivial involution, and that the edge $\{v',d\}$ is not present, as otherwise $(bd)$ is a non-trivial involution. We construct a hardness gadget as follows:
		\begin{itemize}
			\item $I=\{a\}$.
			\item $S=\{b\}$.
			\item $J_1$ is a path of $3$ vertices: The first vertex is a $v$-pin, the second vertex is $y$, and the third vertex is a $b$-pin.
			\item $J_2$ is a path of $3$ vertices: The first vertex is an $a$-pin, the second vertex is $z$, and the third vertex is a $c$-pin.
			\item $J_3$ is just the edge $\{y,z\}$.
		\end{itemize}
		Note first that $\Omega_y$ contains precisely the common neighbours of $v$ and $b$. Thus $v'$ and $a$ are contained in $\Omega_y$. Recall further that $c$ is not adjacent to $v$. As the degree of $H$ is bounded by $3$, we thus have $\Omega_y=\{v',a\}$. Similarly, we obtain that $\Omega_z=\{b,d\}$.
		
		Finally, we observe that $|\Sigma_{v',d}|=0$ as $\{v',d\}$ is not an edge of $H$, and that $|\Sigma_{v',b}|=|\Sigma_{b,a}|=|\Sigma_{a,d}|=1$.
	\end{enumerate}
\end{proof}

\begin{thm}\label{thm:bddeg}
Let $H$ be a graph 
whose involution-free reduction $H^*$ has maximum degree at most~$3$.
If $H^*$ contains at most one vertex, then $\parHom{H}$ can be solved in polynomial time.
Otherwise,   $\parHom{H}$ is $\parP$-complete and, assuming the randomised Exponential Time Hypothesis,
it  cannot be solved in time $\exp(o(\abs{G}))$.
\end{thm}
\begin{proof}
By Theorem~\ref{thm:ReductionSuffices}, for every graph~$G$,
$|\hom{G}{H}| =|\hom{G}{H^\ast}|\mod 2$.
It is trivial to count homomorphisms to a graph with at most one vertex.
Suppose that $H^*$ has at least two vertices.
Then it suffices to show that  $\parHom{H^\ast}$
is $\parP$-complete and that 
$\parHom{H^*}$ cannot be solved in time $\exp(o(\abs{G}))$, unless the rETH fails.

If $H^\ast$ does not contain a square but has at least $2$ vertices, then it has a hardness gadget as shown in~\cite{squarefree}. If $H^\ast$  contains a square, then it has a hardness gadget by Lemma~\ref{lem:mainDeg3}. 

By Theorem~\ref{thm:hardness-gadget}, we obtain that $\parRet{H^*}$ is $\parP$-hard and 
that it cannot be solved in time $\exp(o\abs{J})$, unless the rETH fails.

Finally, since $H^*$ is involution-free, we can reduce $\parRet{H^*}$ to $\parHom{H^*}$ by Theorem~\ref{thm:ETH_RetToHoms}.  
As we have already noted, the size of the oracle queries in this reduction are bounded linearly in the input size,
so the reduction proves that any subexponential-time algorithm for $\parHom{H^\ast}$ would yield a subexponential-time algorithm for $\parRet{H^*}$,
completing the proof. \end{proof}

	\section{Counting List Homomorphisms modulo $2$}\label{sec:parLHom}
	
	Given graphs $G$ and $H$ together with a set of \emph{lists} $\boldS=\{S_v \subseteq V(H) \mid v \in V(G)\}$, a \emph{(list) homomorphism} from $(G, \boldS)$ to $H$ is a homomorphism $h$ from $G$ to $H$ such that for each $v\in V(G)$ we have $h(v)\in S_v$. 
	\label{def:listhomomorphisms}
We use $\hom{(G, \boldS)}{H}$ to denote the set of homomorphisms from $(G, \boldS)$ to $H$.
	List homomorphisms are a natural generalisation of both homomorphisms and retractions.
	
In this section we provide a complete complexity classification for the problem of counting list homomorphisms modulo~$2$ to a given graph~$H$. 
The classification determines for which graphs~$H$ the problem is feasible.
We strengthen the result by 
considering a wider class of graphs~$H$ than in the rest of the paper (where we required~$H$ to be a simple graph, without self-loops or parallel edges). 
Let $\calH$ be the set of all undirected graphs~$H$ which do not have parallel edges --- self-loops are allowed.

 Given a set $S$, let $\calP(S)$ denote its power set. 
 We consider the following problem, parameterised by a graph~$H\in \calH$  and by a set of lists
  $\calL\subseteq \calP(V(H))$.
	
\prob{$\parHom{H, \calL}$.}{A simple graph $G$ and a collection of lists $\boldS=\{S_v\in \calL \mid v\in V(G)\}$.}{$|\hom{(G,\boldS)}{H}| \bmod 2$.}

The input~$G$ to $\parHom{H,\calL}$ is assumed to be simple because this is standard in the field, and because it makes
results stronger. However, this restriction is not important for our result 
--- see Remark~\ref{rem:final}.
Taking $\calL = P(V(H))$, the problem	
  $\parHom{H, \calP(V(H))}$ is the problem of counting list homomorphisms to $H$ modulo~$2$.
  To simplify the notation, we also write $\parLHom{H}$ \label{def:parLhom} for this problem. The following lemma is well-known.

\begin{lem}\label{lem:P4hardness}
Let $H$ be a graph in~$\calH$  that contains a walk $(a,b,c,d)$ 
such that $a\neq c$, $b\neq d$, and $\{a,d\}\notin E(H)$. Let $\calL\subseteq \calP(V(H))$ be a set of lists with $\{\{a,c\}, \{b,d\}\}\subseteq \calL$. Then $\parHom{H,\calL}$ is $\parP$-complete.
	\end{lem}
\begin{proof}
The problem $\parbis$, of counting the independent sets of a bipartite graph, modulo~2,
is   known to be $\parP$-complete~\cite[Theorem 4.2.1]{Faben2008}.
We will reduce $\parbis$ to $\parHom{H,\calL}$.

Let $G$ be a bipartite graph (an input to $\parbis$) with vertex partition $V(G)=(L,R)$. 
For each $v\in L$, let $S_v=\{a,c\}$ and for each $v\in R$ let $S_v=\{b,d\}$. We set $\boldS=\{S_v \mid v\in V(G)\}$. Then every homomorphism $h$ from $(G,\boldS)$ to $H$ corresponds to an independent set in $G$ (and vice versa), where $h(v)\in \{a,d\}$ means that $v$ is \emph{in} the independent set and $h(v)\in \{b,c\}$ means that $v$ is \emph{out} of the independent set. (Since $a\neq c$ and $b\neq d$ it is well-defined whether $v$ is in or out.) Hence a single $\parLHom{H, \calL}$ oracle call with input $(G, \boldS)$ returns the number of independent sets of~$G$, modulo~$2$.
	\end{proof}
	
	\begin{thm}\label{thm:2lists}
Let $H$ be a connected graph in~$\calH$ and let $\calL\subseteq \calP(V(H))$ be a set of lists with $\{S\subseteq V(H) \mid \abs{S}=2\}\subseteq \calL$. 
If 
(i) $H$ is a complete bipartite graph with no self-loops, or 
(ii) $H$ is a complete graph in which every vertex has a self-loop,
then   $\parHom{H,\calL}$ can be solved in polynomial time. Otherwise, $\parHom{H,\calL}$ is $\parP$-complete.
	\end{thm}
\begin{proof}
The easiness result comes from Dyer and Greenhill~\cite[Theorem 1.1]{DG}. 
(Dyer and Greenhill's result is stated for homomorphisms rather than for list homomorphisms, but it is easy to see, and well known,
that it extends to list homomorphisms.)
For the hardness part we consider four cases.
		
\begin{description}
\item[Case 1:] $H$ contains at least one looped and one unlooped vertex.

The problem $\paris$, of counting the independent sets
of a graph, modulo~$2$,  is known to be $\parP$-complete~\cite{Valiant2006}.
In this case there is an easy reduction from $\paris$ to $\parLHom{H, \calL}$.
To see this, note that, since $H$ is connected, it 		
 contains a looped vertex $a$ which is adjacent to an unlooped vertex $b$. 
 Counting the homomorphisms from a graph~$G$  to $H[\{a,b\}]$ is well-known to be equivalent to counting the independent sets of~$G$ (see, 
 e.g.,~\cite{BW}). Since $\{a,b\} \in \calL$ we can use this list to restrict the image of  homomorphisms to $\{a,b\}$, giving the desired reduction. 
 			
\item[Case 2:] $H$ is a bipartite graph without self-loops but it is not a complete bipartite graph.
			
In this case, $H$ contains a path   $(a,b,c,d)$ such that $\{a,d\}\notin E(H)$ so $\parHom{H,\calL}$ is $\parP$-complete by Lemma~\ref{lem:P4hardness}.
			
			\item[Case 3:] $H$ is a graph without self-loops that contains a cycle of odd length. 
			
			Consider a shortest odd-length cycle $C$ in $H$. Due to minimality, $C$ has to be an induced cycle of $H$ (any additional edge between vertices of $C$ would give a shorter even-length cycle and a shorter odd-length cycle). If $C$ is not a triangle, then $C$ contains a path  $(a,b,c,d)$ such that $\{a,d\}\notin E(H)$. If otherwise $C$ is a triangle $(a,b,c,a)$, then $\{a,a\}\notin E(H)$ since $H$ does not have self-loops. In both cases $\parHom{H,\calL}$ is $\parP$-complete by Lemma~\ref{lem:P4hardness}.
			
			\item[Case 4:] $H$ is a graph with all self-loops present but $H$ is not a complete graph.
			
In this case, $H$ contains a path $(a,b,c)$ where $\{a,c\}\notin E(H)$. Since $\{b,b\}\in E(H)$ we can apply Lemma~\ref{lem:P4hardness} to the walk $(a,b,b,c)$ to obtain $\parP$-completeness of $\parHom{H,\calL}$.
		\end{description}
	\end{proof}
	
The following complexity classification for the problem $\parLHom{H}$ follows easily from Theorem~\ref{thm:2lists}.	
\begin{thm}\label{cor:LHomclassification}
Let $H$ be  graph in~$\calH$.
If every connected component $H'$ of~$H$ satisfies one of the following
\begin{enumerate}
\item
$H'$ is a complete bipartite graph with no self-loops, or 
\item $H'$ is a complete graph in which every vertex has a self-loop, 
\end{enumerate}
then   $\parLHom{H}$ can be solved in polynomial time. Otherwise, $\parLHom{H}$ is $\parP$-complete.
\end{thm}

\begin{proof}
The easiness result comes from Dyer and Greenhill~\cite[Theorem 1.1]{DG}. 
For the hardness part, let $H'$ be a connected component of~$H$ that  
is not a complete bipartite graph with no self-loops and is not a complete graph in which every vertex has a self-loop.
Let $\calL$ be the set of all size-2 subsets of $V(H')$.
From Theorem~\ref{thm:2lists}, 
$\parHom{H',\calL}$ is $\parP$-complete.
However, $\parHom{H',\calL}$ reduces trivially to $\parLHom{H}$ --- 
given an input~$(G,\boldS)$ to $\parHom{H',\calL}$
simply return the number of (list) homomorphisms from $(G,\boldS)$ to~$H$, modulo~$2$.

\end{proof}
	
\begin{rem}\label{rem:final}
Theorem~\ref{cor:LHomclassification}
would be unchanged if we changed the definition of  $\parLHom{H}$
so that the input~$G$ can be any graph in~$\calH$ (so it need not be simple). 
A self-loop on a vertex~$v$ of~$G$ simply enforces the constraint that a homomorphism  must map~$v$ to a vertex of~$H$ that has a self-loop. 
The same constraint can be enforced using a list. 
	\end{rem}

	\section*{Acknowledgements} 
	\pdfbookmark[1]{Acknowledgements}{acknowledgements}

We would like to thank Dave Richerby, our fellow member of the square-haters' club, for valuable discussions; this club will hopefully lose some members with the appearance of this work. Furthermore we thank Holger Dell for pointing out the tight conditional lower bound for counting independent sets modulo $2$ in~\cite{DellHMTW14}.

	\section{Index of Symbols and Definitions}\label{sec:index}
	 
	{
		\renewcommand{\arraystretch}{1.5}
		\begin{longtable}{@{}l|lll@{}}
			articulation point & removal increases number of connected components & Def.\ \ref{def:blockcuttree} & p.\ \pageref{def:blockcuttree}\\
			attachment point & {} & Def.\ \ref{def:exit} & p.\ \pageref{def:exit}\\
			$\bc(H)$& block-cut tree of $H$ & Def.\ \ref{def:blockcuttree} & p.\ \pageref{def:blockcuttree}\\ 
			biconnected component & maximal biconnected subgraph  & Def.\ \ref{def:blockcuttree} & p.\ \pageref{def:blockcuttree}\\
			biconnected graph & at least two vertices and no articulation points  & Def.\ \ref{def:blockcuttree} & p.\ \pageref{def:blockcuttree}\\ 
			block-cut tree & tree of biconn.\ components and articulation points & Def.\ \ref{def:blockcuttree} & p.\ \pageref{def:blockcuttree}\\ 
			chordal bipartite graph & all induced cycles are squares & Def.\ \ref{def:chordalbipartite} & p.\ \pageref{def:chordalbipartite}\\
			$\cycles(B)$ & set of distinguished cycles of obstruction $B$ & Def.\ \ref{def:obstruction} & p.\ \pageref{def:obstruction}\\
			$D(C)$ & order induced by cycle $C$ & Def.\ \ref{def:DofC} & p.\ \pageref{def:DofC}\\
			$\deg_H(v)$ & degree of $v$ in graph $H$ & ~ & p.\ \pageref{page:degree}\\
			destination & {} & Def.\ \ref{def:exit} & p.\ \pageref{def:exit}\\
			diamonds & distinguished family of chordal bipartite graphs & Def.\ \ref{def:diamond} & p.\ \pageref{def:diamond}\\ 
			exit & {} & Def.\ \ref{def:exit} & p.\ \pageref{def:exit}\\
			$F$ & graph with two squares sharing one edge & Def.\ \ref{def:graphF} & p.\ \pageref{def:graphF}\\
			good start & {} & Def.\ \ref{def:goodstartgoodstop} & p.\ \pageref{def:goodstartgoodstop}\\
			good stop & {} & Def.\ \ref{def:goodstartgoodstop} & p.\ \pageref{def:goodstartgoodstop}\\
			hardness gadget & substructure of a graph inducing $\parP$-hardness & Def.\ \ref{defn:hardness-gadget} & p.\ \pageref{defn:hardness-gadget}\\
			$H[S]$ & subgraph of $H$ induced by $S$ & ~ & p.\ \pageref{page:inducedgraph}\\
			homomorphism & edge-preserving mapping & ~ & p.\ \pageref{page:homomorphism}\\
			$\hom{G}{H}$ & set of homomorphisms from  $G$ to $H$ & ~ & p.\ \pageref{page:homsGH}\\
			$\hom{(J, x)}{(H, y)}$ & set of homomorphisms with pinned vertices& ~ & p.\ \pageref{page:partialhomset}\\
			impasses & distinguished family of chordal bipartite graphs & Def.\ \ref{def:impasse} & p.\ \pageref{def:impasse}\\
			involution & automorphism of order $\leq2$ & ~ & p.\ \pageref{page:involution} \\
			involution-free graph & graph without non-trivial involutions & ~ & p.\ \pageref{page:involutionfree} \\
			involution-free reduction & {} & ~ & p.\ \pageref{page:involutionfreereduction}\\
			$K_4$, $K_4$-minor-free & 4-vertex complete graph  &  & p.\ \pageref{page:K4} \\
			list homomorphism & {} & ~ & p.\ \pageref{def:listhomomorphisms} \\
			$N_{W,H}(w_i)$& walk-neighbour-set &Def.\ \ref{def:walkneighbourset} & p.\ \pageref{def:walkneighbourset}\\
			obstruction & distinguished biconnected $K_4$-minor-free graph& Def.\ \ref{def:obstruction} & p.\ \pageref{def:obstruction}\\
			obstruction-free path & path in the block-cut tree excluding obstructions & Def.\ \ref{def:obstructionfreePath} & p.\ \pageref{def:obstructionfreePath} \\
			pair of connectors & distinguished pair of vertices of an impasse & Def.\ \ref{def:impasse} & p.\ \pageref{def:impasse}\\
			partially $H$-labelled graph & pair consisting of a graph and a pinning function & ~ & p.\ \pageref{page:partiallyHlabelled}\\
			$\SP{a,b}$ & a particular shortest path from $a$ to $b$ in $H$ & Def.\ \ref{def:pathroutine} &p.\ \pageref{def:pathroutine}\\
			pinning function & partial function between vertices of two graphs & ~ & p.\ \pageref{page:pinningfunction}\\
			pre-hardness gadget & $\parP$-hardness for all $K_4$-minor-free (1,2)-supergraphs & Def.\ \ref{def:prehardness} & p.\ \pageref{def:prehardness}\\
			$R$-closed/$R$-open & {} & Def.\ \ref{def:Rclosed} & p.\ \pageref{def:Rclosed}\\

rETH & randomised Exponential Time Hypothesis   &  {} & p.\ \pageref{page:rETH}\\

			retraction & homomorphism from a partially labelled graph & ~ & p.\ \pageref{page:retractions} \\
			separation/separator & {} & Def.\ \ref{def:separator} & p.\ \pageref{def:separator}\\
			$S_{k,\ell}$ & distinguished $\TV$-typed supergraph of $F$ & Def.\ \ref{def:Skl} & p.\ \pageref{def:Skl} \\
			strong hardness gadget & $\parP$-hardness for all $K_4$-minor-free supergraphs & Def.\ \ref{def:stronghardnessgadget} & p.\ \pageref{def:stronghardnessgadget}\\
			suitable connector & {} & Def.\ \ref{def:suitableconnector} &p.\ \pageref{def:suitableconnector}\\
			suitable subtree & {} & Def.\ \ref{def:suitabletree} &p.\ \pageref{def:suitabletree}\\
			type $\TV$ & predicate for supergraphs of $F$ & Def.\ \ref{def:typeV} & p.\ \pageref{def:typeV}\\
			
	  walk-neighbour-set &&Def.\ \ref{def:walkneighbourset} & p.\ \pageref{def:walkneighbourset}\\

			$\cseg_C(a,b)$/$\cseg_C(a)$ & walk from $a$ to $b$ / from $a$ to $a$ along cycle $C$ & Def.\ \ref{def:DofC} & p.\ \pageref{def:DofC}\\
			$\Gamma_H(v)$ & neighbourhood of $v$ in graph $H$ & ~ & p.\ \pageref{page:neighbourhood}\\
			$\Gamma_H(S)$ & joint neighbourhood of $S$ & ~ &  p.\ \pageref{page:neighbourhood}\\
						$\Gamma_{H\setminus F}(i,j)$ & common neighbours of $v_i$ and $v_j$ in $H\setminus F$ & Def.\ \ref{def:graphF} & p.\ \pageref{def:graphF}\\
			$\parHom{H}$ & counting homomorphisms to $H$ mod 2 & ~ & p.\ \pageref{page:ParHom}\\
						
$\paris$ & counting independent sets mod 2 & ~ & p.\ \pageref{page:ParIS}\\

			$\parLHom{H}$ & counting list homomorphisms to $H$ mod 2 & ~ & p.\ \pageref{def:parLhom} \\
			$\parP$ & complexity class of parity problems & ~ & p.\ \pageref{page:parP} \\
			$\parRet{H}$ & counting retractions to $H$ mod 2 & ~ & p.\ \pageref{page:ParRet}\\
			(1,2)-supergraph & supergraph without new adjacencies and 2-paths & Def.\ \ref{def:12supergraph} & p.\ \pageref{def:12supergraph}\\
			$+$ & concatenation of walks & Def.\ \ref{def:plus} & p.\ \pageref{def:plus}
		\end{longtable}
	}	
	
	\clearpage
	
	\bibliography{\jobname}

\begin{thebibliography}{10}

\bibitem{Borgs}
Christian Borgs, Jennifer Chayes, L{\'a}szl{\'o} Lov{\'a}sz, Vera~T. S{\'o}s,
  and Katalin Vesztergombi.
\newblock Counting graph homomorphisms.
\newblock In Martin Klazar, Jan Kratochv{\'i}l, Martin Loebl, Ji{\v{r}}{\'i}
  Matou{\v{s}}ek, Pavel Valtr, and Robin Thomas, editors, {\em Topics in
  Discrete Mathematics}, pages 315--371, Berlin, Heidelberg, 2006. Springer
  Berlin Heidelberg.

\bibitem{BW}
Graham~R. Brightwell and Peter Winkler.
\newblock Graph homomorphisms and phase transitions.
\newblock {\em J. Comb. Theory, Ser. {B}}, 77(2):221--262, 1999.
\newblock \href {https://doi.org/10.1006/jctb.1999.1899}
  {\path{doi:10.1006/jctb.1999.1899}}.

\bibitem{ChenCD19}
Hubie Chen, Radu Curticapean, and Holger Dell.
\newblock The exponential-time complexity of counting (quantum) graph
  homomorphisms.
\newblock In Ignasi Sau and Dimitrios~M. Thilikos, editors, {\em
  Graph-Theoretic Concepts in Computer Science - 45th International Workshop,
  {WG} 2019, Vall de N{\'{u}}ria, Spain, June 19-21, 2019, Revised Papers},
  volume 11789 of {\em Lecture Notes in Computer Science}, pages 364--378.
  Springer, 2019.
\newblock \href {https://doi.org/10.1007/978-3-030-30786-8\_28}
  {\path{doi:10.1007/978-3-030-30786-8\_28}}.

\bibitem{CurticapeanDM17}
Radu Curticapean, Holger Dell, and D{\'{a}}niel Marx.
\newblock Homomorphisms are a good basis for counting small subgraphs.
\newblock In Hamed Hatami, Pierre McKenzie, and Valerie King, editors, {\em
  Proceedings of the 49th Annual {ACM} {SIGACT} Symposium on Theory of
  Computing, {STOC} 2017, Montreal, QC, Canada, June 19-23, 2017}, pages
  210--223. {ACM}, 2017.
\newblock \href {https://doi.org/10.1145/3055399.3055502}
  {\path{doi:10.1145/3055399.3055502}}.

\bibitem{DellHMTW14}
Holger Dell, Thore Husfeldt, D{\'{a}}niel Marx, Nina Taslaman, and Martin
  Wahlen.
\newblock Exponential time complexity of the permanent and the {T}utte
  polynomial.
\newblock {\em {ACM} Trans. Algorithms}, 10(4):21:1--21:32, 2014.
\newblock \href {https://doi.org/10.1145/2635812} {\path{doi:10.1145/2635812}}.

\bibitem{DiazST02}
Josep D{\'{\i}}az, Maria~J. Serna, and Dimitrios~M. Thilikos.
\newblock Counting {H}-colorings of partial k-trees.
\newblock {\em Theor. Comput. Sci.}, 281(1-2):291--309, 2002.
\newblock \href {https://doi.org/10.1016/S0304-3975(02)00017-8}
  {\path{doi:10.1016/S0304-3975(02)00017-8}}.

\bibitem{DiestelGT}
Reinhard Diestel.
\newblock {\em Graph Theory, 5th Edition}, volume 173 of {\em Graduate texts in
  mathematics}.
\newblock Springer, 2016.

\bibitem{Duffin}
Richard~J Duffin.
\newblock Topology of series-parallel networks.
\newblock {\em Journal of Mathematical Analysis and Applications},
  10(2):303--318, 1965.

\bibitem{DG}
Martin Dyer and Catherine Greenhill.
\newblock The complexity of counting graph homomorphisms.
\newblock {\em Random Structures \& Algorithms}, 17(3-4):260--289, 2000.

\bibitem{APred}
Martin~E. Dyer, Leslie~Ann Goldberg, Catherine~S. Greenhill, and Mark Jerrum.
\newblock The relative complexity of approximate counting problems.
\newblock {\em Algorithmica}, 38(3):471--500, 2004.
\newblock \href {https://doi.org/10.1007/s00453-003-1073-y}
  {\path{doi:10.1007/s00453-003-1073-y}}.

\bibitem{Faben2008}
John Faben.
\newblock The complexity of counting solutions to generalised satisfiability
  problems modulo k.
\newblock {\em arXiv preprint arXiv:0809.1836}, 2008.

\bibitem{FJ}
John Faben and Mark Jerrum.
\newblock {The Complexity of Parity Graph Homomorphism: An Initial
  Investigation}.
\newblock {\em Theory of Computing}, 11(2):35--57, 2015.
\newblock \href {https://doi.org/10.4086/toc.2015.v011a002}
  {\path{doi:10.4086/toc.2015.v011a002}}.

\bibitem{retr}
Tomas Feder and Pavol Hell.
\newblock List homomorphisms to reflexive graphs.
\newblock {\em J. Combin. Theory Ser. B}, 72(2):236--250, 1998.
\newblock URL:
  \url{https://ezproxy-prd.bodleian.ox.ac.uk:4563/10.1006/jctb.1997.1812},
  \href {https://doi.org/10.1006/jctb.1997.1812}
  {\path{doi:10.1006/jctb.1997.1812}}.

\bibitem{FlumGrohe}
J{\"{o}}rg Flum and Martin Grohe.
\newblock The parameterized complexity of counting problems.
\newblock {\em {SIAM} J. Comput.}, 33(4):892--922, 2004.
\newblock \href {https://doi.org/10.1137/S0097539703427203}
  {\path{doi:10.1137/S0097539703427203}}.

\bibitem{FGZRet}
Jacob Focke, Leslie~Ann Goldberg, and Stanislav {\v{Z}}ivn{\'y}.
\newblock The complexity of approximately counting retractions to square-free
  graphs.
\newblock {\em arXiv preprint arXiv:1907.02319}, 2019.

\bibitem{BIShard}
Andreas Galanis, Leslie~Ann Goldberg, and Mark Jerrum.
\newblock Approximately counting {H}-colorings is \#{BIS}-hard.
\newblock {\em {SIAM} J. Comput.}, 45(3):680--711, 2016.
\newblock \href {https://doi.org/10.1137/15M1020551}
  {\path{doi:10.1137/15M1020551}}.

\bibitem{ListApx}
Andreas Galanis, Leslie~Ann Goldberg, and Mark Jerrum.
\newblock A complexity trichotomy for approximately counting list
  \emph{H}-colorings.
\newblock {\em {ACM} Trans. Comput. Theory}, 9(2):9:1--9:22, 2017.
\newblock \href {https://doi.org/10.1145/3037381} {\path{doi:10.1145/3037381}}.

\bibitem{cactus}
Andreas G{\"{o}}bel, Leslie~Ann Goldberg, and David Richerby.
\newblock The complexity of counting homomorphisms to cactus graphs modulo 2.
\newblock {\em {ACM} Trans. Comput. Theory}, 6(4):17:1--17:29, 2014.
\newblock \href {https://doi.org/10.1145/2635825} {\path{doi:10.1145/2635825}}.

\bibitem{squarefree}
Andreas G{\"o}bel, Leslie~Ann Goldberg, and David Richerby.
\newblock Counting homomorphisms to square-free graphs, modulo 2.
\newblock {\em ACM Transactions on Computation Theory (TOCT)}, 8(3):12, 2016.

\bibitem{goebeltrees}
Andreas G{\"{o}}bel, J.~A.~Gregor Lagodzinski, and Karen Seidel.
\newblock Counting homomorphisms to trees modulo a prime.
\newblock In Igor Potapov, Paul~G. Spirakis, and James Worrell, editors, {\em
  43rd International Symposium on Mathematical Foundations of Computer Science,
  {MFCS} 2018, August 27-31, 2018, Liverpool, {UK}}, volume 117 of {\em
  LIPIcs}, pages 49:1--49:13. Schloss Dagstuhl - Leibniz-Zentrum f{\"{u}}r
  Informatik, 2018.
\newblock \href {https://doi.org/10.4230/LIPIcs.MFCS.2018.49}
  {\path{doi:10.4230/LIPIcs.MFCS.2018.49}}.

\bibitem{MarkTrees}
Leslie~Ann Goldberg and Mark Jerrum.
\newblock The complexity of approximately counting tree homomorphisms.
\newblock {\em {ACM} Trans. Comput. Theory}, 6(2):8:1--8:31, 2014.
\newblock \href {https://doi.org/10.1145/2600917} {\path{doi:10.1145/2600917}}.

\bibitem{GP86:Parallel}
L.~M. Goldschlager and I.~Parberry.
\newblock On the construction of parallel computers from various bases of
  {B}oolean functions.
\newblock {\em Theor. Comput. Sci.}, 43:43--58, 1986.

\bibitem{ChordalBipartite99}
Martin~Charles Golumbic and Clinton~F. Goss.
\newblock {Perfect Elimination and Chordal Bipartite Graphs}.
\newblock {\em Journal of Graph Theory}, 2(2):155--163, 1978.
\newblock \href {https://doi.org/10.1002/jgt.3190020209}
  {\path{doi:10.1002/jgt.3190020209}}.

\bibitem{Grohe07}
Martin Grohe.
\newblock The complexity of homomorphism and constraint satisfaction problems
  seen from the other side.
\newblock {\em J. {ACM}}, 54(1):1:1--1:24, 2007.
\newblock \href {https://doi.org/10.1145/1206035.1206036}
  {\path{doi:10.1145/1206035.1206036}}.

\bibitem{GroheSS01}
Martin Grohe, Thomas Schwentick, and Luc Segoufin.
\newblock When is the evaluation of conjunctive queries tractable?
\newblock In Jeffrey~Scott Vitter, Paul~G. Spirakis, and Mihalis Yannakakis,
  editors, {\em Proceedings on 33rd Annual {ACM} Symposium on Theory of
  Computing, July 6-8, 2001, Heraklion, Crete, Greece}, pages 657--666. {ACM},
  2001.
\newblock \href {https://doi.org/10.1145/380752.380867}
  {\path{doi:10.1145/380752.380867}}.

\bibitem{HellN90}
Pavol Hell and Jaroslav Nesetril.
\newblock On the complexity of \emph{H}-coloring.
\newblock {\em J. Comb. Theory, Ser. {B}}, 48(1):92--110, 1990.
\newblock \href {https://doi.org/10.1016/0095-8956(90)90132-J}
  {\path{doi:10.1016/0095-8956(90)90132-J}}.

\bibitem{ImpagliazzoP01}
Russell Impagliazzo and Ramamohan Paturi.
\newblock {On the Complexity of k-SAT}.
\newblock {\em J. Comput. Syst. Sci.}, 62(2):367--375, 2001.
\newblock \href {https://doi.org/10.1006/jcss.2000.1727}
  {\path{doi:10.1006/jcss.2000.1727}}.

\bibitem{Karp72}
Richard~M. Karp.
\newblock Reducibility among combinatorial problems.
\newblock In Raymond~E. Miller and James~W. Thatcher, editors, {\em Proceedings
  of a symposium on the Complexity of Computer Computations, held March 20-22,
  1972, at the {IBM} Thomas J. Watson Research Center, Yorktown Heights, New
  York, {USA}}, The {IBM} Research Symposia Series, pages 85--103. Plenum
  Press, New York, 1972.
\newblock \href {https://doi.org/10.1007/978-1-4684-2001-2\_9}
  {\path{doi:10.1007/978-1-4684-2001-2\_9}}.

\bibitem{BulatovModp}
Amirhossein Kazeminia and Andrei~A. Bulatov.
\newblock Counting homomorphisms modulo a prime number.
\newblock In Peter Rossmanith, Pinar Heggernes, and Joost{-}Pieter Katoen,
  editors, {\em 44th International Symposium on Mathematical Foundations of
  Computer Science, {MFCS} 2019, August 26-30, 2019, Aachen, Germany}, volume
  138 of {\em LIPIcs}, pages 59:1--59:13. Schloss Dagstuhl - Leibniz-Zentrum
  f{\"{u}}r Informatik, 2019.
\newblock \href {https://doi.org/10.4230/LIPIcs.MFCS.2019.59}
  {\path{doi:10.4230/LIPIcs.MFCS.2019.59}}.

\bibitem{KelkThesis}
Steven Kelk.
\newblock {\em {On the relative complexity of approximately counting
  $H$-colourings}}.
\newblock PhD thesis, Warwick University, 2003.

\bibitem{KolaitisV00}
Phokion~G. Kolaitis and Moshe~Y. Vardi.
\newblock Conjunctive-query containment and constraint satisfaction.
\newblock {\em J. Comput. Syst. Sci.}, 61(2):302--332, 2000.
\newblock \href {https://doi.org/10.1006/jcss.2000.1713}
  {\path{doi:10.1006/jcss.2000.1713}}.

\bibitem{Lovaszsurvey}
L\'{a}szl\'{o} Lov\'{a}sz.
\newblock Graph minor theory.
\newblock {\em Bull. Amer. Math. Soc. (N.S.)}, 43(1):75--86, 2006.
\newblock \href {https://doi.org/10.1090/S0273-0979-05-01088-8}
  {\path{doi:10.1090/S0273-0979-05-01088-8}}.

\bibitem{PZ82:Counting}
C.~H. Papadimitriou and S.~Zachos.
\newblock Two remarks on the power of counting.
\newblock In {\em Proc. 6th GI-Conference on Theoretical Computer Science},
  pages 269--275. Springer-Verlag, 1982.

\bibitem{RothW20}
Marc Roth and Philip Wellnitz.
\newblock Counting and finding homomorphisms is universal for parameterized
  complexity theory.
\newblock In Shuchi Chawla, editor, {\em Proceedings of the 2020 {ACM-SIAM}
  Symposium on Discrete Algorithms, {SODA} 2020, Salt Lake City, UT, USA,
  January 5-8, 2020}, pages 2161--2180. {SIAM}, 2020.
\newblock \href {https://doi.org/10.1137/1.9781611975994.133}
  {\path{doi:10.1137/1.9781611975994.133}}.

\bibitem{Tod91:PP-PH}
S.~Toda.
\newblock {PP} is as hard as the polynomial-time hierarchy.
\newblock {\em SIAM J. Comput.}, 20(5):865--877, 1991.

\bibitem{Valiant2006}
Leslie~G Valiant.
\newblock {Accidental algorthims}.
\newblock In {\em 2006 47th Annual IEEE Symposium on Foundations of Computer
  Science (FOCS'06)}, pages 509--517. IEEE, 2006.

\end{thebibliography}
	
	\appendix
	
	\pagebreak

\end{document}